\documentclass[letterpaper,11pt,oneside,reqno]{amsart}
\usepackage{amsfonts,amsmath, amssymb,amsthm,amscd}
\usepackage[usenames,dvipsnames]{color}
\usepackage{bm}
\usepackage{mathrsfs}
\usepackage{yfonts}
\usepackage[mpexclude,DIV13]{typearea}
\usepackage{verbatim}
\usepackage{hyperref}
\usepackage{graphicx}
\usepackage[latin1]{inputenc}
\usepackage{latexsym}
\usepackage{lscape}
\usepackage{epsfig}
\usepackage{subfigure}

\numberwithin{equation}{subsection}

\usepackage{setspace}
\include{bibtex}

\begin{document}

\bibliographystyle{alpha}
\newcommand{\cn}[1]{\overline{#1}}
\newcommand{\e}[0]{\epsilon}
\newcommand{\EE}{\ensuremath{\mathbb{E}}}
\newcommand{\qq}[1]{(q;q)_{#1}}
\newcommand{\A}{\ensuremath{\mathcal{A}}}
\newcommand{\GT}{\ensuremath{\mathbb{GT}}}
\newcommand{\link}{\ensuremath{Q}}
\newcommand{\PP}{\ensuremath{\mathbb{P}}}
\newcommand{\frakP}{\ensuremath{\mathfrak{P}}}
\newcommand{\frakQ}{\ensuremath{\mathfrak{Q}}}
\newcommand{\frakq}{\ensuremath{\mathfrak{q}}}
\newcommand{\R}{\ensuremath{\mathbb{R}}}
\newcommand{\Rplus}{\ensuremath{\mathbb{R}_{+}}}
\newcommand{\C}{\ensuremath{\mathbb{C}}}
\newcommand{\Z}{\ensuremath{\mathbb{Z}}}
\newcommand{\N}{\ensuremath{\mathbb{N}}}
\newcommand{\Weyl}[1]{\ensuremath{\mathbb{W}}^{#1}}
\newcommand{\LP}[1]{\ensuremath{\mathcal{C}}^{#1}}
\newcommand{\LPe}[1]{\ensuremath{\mathcal{C}}_{\e}^{#1}}
\newcommand{\LPsd}[1]{\ensuremath{\mathcal{\tilde{C}}}^{#1}}
\newcommand{\CP}[1]{\ensuremath{\mathcal{W}}^{#1}}
\newcommand{\Weylct}[1]{\ensuremath{\mathbb{\bar{W}}}^{#1}}
\newcommand{\Zgzero}{\ensuremath{\mathbb{Z}_{>0}}}
\newcommand{\Zgeqzero}{\ensuremath{\mathbb{Z}_{\geq 0}}}
\newcommand{\Zleqzero}{\ensuremath{\mathbb{Z}_{\leq 0}}}
\newcommand{\Q}{\ensuremath{\mathbb{Q}}}
\newcommand{\T}{\ensuremath{\mathbb{T}}}
\newcommand{\Y}{\ensuremath{\mathbb{Y}}}
\newcommand{\M}{\ensuremath{\mathbf{M}}}
\newcommand{\MM}{\ensuremath{\mathbf{MM}}}
\newcommand{\W}[1]{\ensuremath{\mathbf{W}}_{(#1)}}
\newcommand{\WM}[1]{\ensuremath{\mathbf{WM}}_{(#1)}}
\newcommand{\Zsd}{\ensuremath{\mathbf{Z}}}
\newcommand{\Fsd}{\ensuremath{\mathbf{F}}}
\newcommand{\symBM}{\ensuremath{\mathbf{W}}}
\newcommand{\symFE}{\ensuremath{\mathbf{S}}}
\newcommand{\tazrp}{\ensuremath{\nu}}

\newcommand{\Domain}{\ensuremath{\mathbf{D}}}

\newcommand{\const}{\ensuremath{c}}

\newcommand{\Real}{\ensuremath{\mathrm{Re}}}
\newcommand{\Imag}{\ensuremath{\mathrm{Im}}}
\newcommand{\re}{\ensuremath{\mathrm{Re}}}

\newcommand{\Sym}{\ensuremath{\mathrm{Sym}}}

\newcommand{\bfone}{\ensuremath{\mathbf{1}}}

\newcommand{\whitenoise}{\ensuremath{\mathscr{\dot{W}}}}
\newcommand{\alphaW}[1]{\ensuremath{\mathbf{\alpha W}}_{(#1)}}
\newcommand{\alphaWM}[1]{\ensuremath{\mathbf{\alpha WM}}_{(#1)}}
\newcommand{\malpha}{\ensuremath{\hat{\alpha}}}
\newcommand{\walpha}{\ensuremath{\alpha}}
\newcommand{\edge}{\textrm{edge}}
\newcommand{\dist}{\textrm{dist}}

\newcommand{\OO}[0]{\Omega}
\newcommand{\F}[0]{\mathfrak{F}}
\newcommand{\poly}[0]{R}
\def \Ai {{\rm Ai}}
\def \Pf {{\rm Pf}}
\def \sgn {{\rm sgn}}
\def \SS {\mathcal{S}}
\newcommand{\poles}{\mathbb{A}}
\def \ss {\mathcal{X}}
\newcommand{\var}{{\rm var}}

\newcommand{\Res}[1]{\underset{{#1}}{\mathbf{Res}}}
\newcommand{\Lim}[1]{\underset{{#1}}{\mathbf{lim}}}
\newcommand{\Resq}[1]{\underset{{#1}}{\mathbf{Res}^q}}
\newcommand{\Resc}[1]{\underset{{#1}}{\mathbf{Res}^c}}
\newcommand{\Resfrac}[1]{\mathbf{Res}_{{#1}}}
\newcommand{\Sub}[1]{\underset{{#1}}{\mathbf{Sub}}}
\newcommand{\Subq}[1]{\underset{{#1}}{\mathbf{Sub}^q}}
\newcommand{\Subc}[1]{\underset{{#1}}{\mathbf{Sub}^c}}

\newcommand{\llangle}[0]{\ensuremath{\big\langle}}
\newcommand{\rrangle}[0]{\ensuremath{\big\rangle_{\mathcal{W}}}}
\newcommand{\rranglesmall}[0]{\ensuremath{\big\rangle_{\mathcal{C}}}}
\newcommand{\rranglesd}[0]{\ensuremath{\big\rangle}_{\mathcal{\tilde{C}}}}

\newcommand{\psir}{\ensuremath{\Psi^{r}}}
\newcommand{\psil}{\ensuremath{\Psi^{\ell}}}
\newcommand{\psifwd}{\ensuremath{\Psi^{\textrm{fwd}}}}
\newcommand{\psimfwd}{\ensuremath{\Psi^{\textrm{cfwd}}}}
\newcommand{\psibwd}{\ensuremath{\Psi^{\textrm{bwd}}}}

\newcommand{\psire}{\ensuremath{\Psi^{r,\e}}}
\newcommand{\psiretilde}{\ensuremath{\hat{\Psi}^{r,\e}}}
\newcommand{\psile}{\ensuremath{\Psi^{\ell,\e}}}
\newcommand{\psifwde}{\ensuremath{\Psi^{\textrm{fwd},\e}}}
\newcommand{\psimfwde}{\ensuremath{\Psi^{\textrm{cfwd},\e}}}
\newcommand{\psibwde}{\ensuremath{\Psi^{\textrm{bwd},\e}}}

\newcommand{\psirezero}{\ensuremath{\Psi^{r,0}}}
\newcommand{\psiretildezero}{\ensuremath{\tilde{\Psi}^{r,0}}}
\newcommand{\psilezero}{\ensuremath{\Psi^{\ell,0}}}

\newcommand{\psirsd}{\ensuremath{\tilde{\Psi}^{r}}}
\newcommand{\psilsd}{\ensuremath{\tilde{\Psi}^{\ell}}}
\newcommand{\psifwdsd}{\ensuremath{\tilde{\Psi}^{\textrm{fwd}}}}
\newcommand{\psimfwdsd}{\ensuremath{\tilde{\Psi}^{\textrm{cfwd}}}}
\newcommand{\psibwdsd}{\ensuremath{\tilde{\Psi}^{\textrm{bwd}}}}

\newcommand{\psirct}{\ensuremath{\bar{\Psi}^{r}}}
\newcommand{\psilct}{\ensuremath{\bar{\Psi}^{\ell}}}

\newcommand{\Id}{\ensuremath{\mathrm{Id}}}
\newcommand{\V}{\ensuremath{\Delta}}

\newcommand{\HqBoson}{\ensuremath{\mathcal{H}^{\textrm{q-Boson}}}}

\newcommand{\KqBoson}{\ensuremath{\mathcal{K}^{\textrm{q-Boson}}}}
\newcommand{\KqBosonDual}{\ensuremath{\mathcal{M}^{\textrm{q-Boson}}}}
\newcommand{\FqBoson}{\ensuremath{\mathcal{F}^{\textrm{q-Boson}}}}
\newcommand{\JqBoson}{\ensuremath{\mathcal{J}^{\textrm{q-Boson}}}}
\newcommand{\JqBosonlarge}{\ensuremath{\mathcal{J}_{\textrm{large}}^{\textrm{q-Boson}}}}

\newcommand{\KqBosone}{\ensuremath{\mathcal{K}^{\textrm{q-Boson},\e}}}
\newcommand{\KqBosoneDual}{\ensuremath{\mathcal{M}^{\textrm{q-Boson},\e}}}
\newcommand{\FqBosone}{\ensuremath{\mathcal{F}^{\textrm{q-Boson},\e}}}
\newcommand{\JqBosone}{\ensuremath{\mathcal{J}^{\textrm{q-Boson},\e}}}

\newcommand{\KqBosonezero}{\ensuremath{\mathcal{K}^{\textrm{q-Boson},0}}}
\newcommand{\KqBosonezeroDual}{\ensuremath{\mathcal{M}^{\textrm{q-Boson},0}}}
\newcommand{\FqBosonezero}{\ensuremath{\mathcal{F}^{\textrm{q-Boson},0}}}
\newcommand{\JqBosonezero}{\ensuremath{\mathcal{J}^{\textrm{q-Boson},0}}}

\newcommand{\KqBosonsd}{\ensuremath{\mathcal{K}^{\textrm{S-D}}}}
\newcommand{\KqBosonsdDual}{\ensuremath{\mathcal{M}^{\textrm{S-D}}}}
\newcommand{\FqBosonsd}{\ensuremath{\mathcal{F}^{\textrm{S-D}}}}
\newcommand{\JqBosonsd}{\ensuremath{\mathcal{J}^{\textrm{S-D}}}}

\newcommand{\KqBosonct}{\ensuremath{\mathcal{K}^{\textrm{F-Y}}}}
\newcommand{\KqBosonctDual}{\ensuremath{\mathcal{M}^{\textrm{F-Y}}}}
\newcommand{\FqBosonct}{\ensuremath{\mathcal{F}^{\textrm{F-Y}}}}
\newcommand{\JqBosonct}{\ensuremath{\mathcal{J}^{\textrm{F-Y}}}}

\renewcommand{\i}{\mathbf i}
\newcommand{\qfac}[2]{\ensuremath{#2}_{#1}!}
\newcommand{\Abwd}{\ensuremath{\mathcal{H}^{\textrm{bwd}}}}
\newcommand{\Afwd}{\ensuremath{\mathcal{H}^{\textrm{fwd}}}}
\newcommand{\Amfwd}{\ensuremath{\mathcal{H}^{\textrm{cfwd}}}}

\newcommand{\Abwde}{\ensuremath{\mathcal{H}^{\textrm{bwd},\e}}}
\newcommand{\Afwde}{\ensuremath{\mathcal{H}^{\textrm{fwd},\e}}}
\newcommand{\Amfwde}{\ensuremath{\mathcal{H}^{\textrm{cfwd},\e}}}

\newcommand{\Abwdsd}{\ensuremath{\mathcal{\tilde{H}}^{\textrm{bwd}}}}
\newcommand{\Afwdsd}{\ensuremath{\mathcal{\tilde{H}}^{\textrm{fwd}}}}
\newcommand{\Amfwdsd}{\ensuremath{\mathcal{\tilde{H}}^{\textrm{cfwd}}}}

\newcommand{\difbwd}{\ensuremath{\nabla^{\textrm{bwd}}}}
\newcommand{\diffwd}{\ensuremath{\nabla^{\textrm{fwd}}}}

\newcommand{\difbwde}{\ensuremath{\nabla^{\textrm{bwd},\e}}}
\newcommand{\diffwde}{\ensuremath{\nabla^{\textrm{fwd},\e}}}

\newcommand{\Freebwd}{\ensuremath{\mathcal{L}^{\textrm{bwd}}}}
\newcommand{\Freefwd}{\ensuremath{\mathcal{L}^{\textrm{fwd}}}}

\newcommand{\Freebwde}{\ensuremath{\mathcal{L}^{\textrm{bwd},\e}}}
\newcommand{\Freefwde}{\ensuremath{\mathcal{L}^{\textrm{fwd},\e}}}

\newcommand{\Freebwdsd}{\ensuremath{\mathcal{\tilde{L}}^{\textrm{bwd}}}}
\newcommand{\Freefwdsd}{\ensuremath{\mathcal{\tilde{L}}^{\textrm{fwd}}}}

\newcommand{\ul}[2]{\underline{#1}_{#2}}
\newcommand{\qhat}[1]{\widehat{#1}^{q}}
\newcommand{\La}[0]{\Lambda}
\newcommand{\la}[0]{\lambda}
\newcommand{\ta}[0]{\theta}
\newcommand{\w}[0]{\omega}
\newcommand{\ra}[0]{\rightarrow}
\newcommand{\vectoro}{\overline}
\newtheorem{theorem}{Theorem}[section]
\newtheorem{partialtheorem}{Partial Theorem}[section]
\newtheorem{conj}[theorem]{Conjecture}
\newtheorem{lemma}[theorem]{Lemma}
\newtheorem{proposition}[theorem]{Proposition}
\newtheorem{corollary}[theorem]{Corollary}
\newtheorem{claim}[theorem]{Claim}
\newtheorem{formal}[theorem]{Critical point derivation}
\newtheorem{experiment}[theorem]{Experimental Result}
\newtheorem{prop}{Proposition}

\def\todo#1{\marginpar{\raggedright\footnotesize #1}}
\def\change#1{{\color{green}\todo{change}#1}}
\def\note#1{\textup{\textsf{\Large\color{blue}(#1)}}}

\theoremstyle{definition}
\newtheorem{remark}[theorem]{Remark}

\theoremstyle{definition}
\newtheorem{example}[theorem]{Example}

\theoremstyle{definition}
\newtheorem{definition}[theorem]{Definition}

\theoremstyle{definition}
\newtheorem{definitions}[theorem]{Definitions}

\begin{abstract}
We develop spectral theory for the generator of the $q$-Boson (stochastic) particle system. Our central result is a Plancherel type  isomorphism theorem for this system. This theorem has various implications. It proves the completeness of the Bethe ansatz for the $q$-Boson generator and consequently enables us to solve the Kolmogorov forward and backward equations for general initial data. Owing to a Markov duality with $q$-TASEP, this leads to moment formulas which characterize the fixed time distribution of $q$-TASEP started from general initial conditions. The theorem also implies the biorthogonality of the left and right eigenfunctions.

We consider limits of our $q$-Boson results to a discrete delta Bose gas considered previously by van Diejen, as well as to another discrete delta Bose gas that describes the evolution of  moments of the semi-discrete stochastic heat equation (or equivalently, the O'Connell-Yor semi-discrete directed polymer partition function). A further limit takes us to the delta Bose gas which arises in studying moments of the stochastic heat equation / Kardar-Parisi-Zhang equation.
\end{abstract}

\title{Spectral theory for the $q$-Boson particle system}
\author[A. Borodin]{Alexei Borodin}
\address{A. Borodin,
Massachusetts Institute of Technology,
Department of Mathematics,
77 Massachusetts Avenue, Cambridge, MA 02139-4307, USA, and Institute for Information Transmission Problems, Bolshoy Karetny per. 19, Moscow 127994, Russia}
\email{borodin@math.mit.edu}

\author[I. Corwin]{Ivan Corwin}
\address{I. Corwin, Columbia University,
Department of Mathematics,
2990 Broadway,
New York, NY 10027, USA,
and Clay Mathematics Institute, 10 Memorial Blvd. Suite 902, Providence, RI 02903, USA,
and Massachusetts Institute of Technology,
Department of Mathematics,
77 Massachusetts Avenue, Cambridge, MA 02139-4307, USA}
\email{ivan.corwin@gmail.com}

\author[L. Petrov]{Leonid Petrov}
\address{L. Petrov,
Department of Mathematics, Northeastern University, 360 Huntington ave.,
Boston, MA 02115, USA, and Institute for Information Transmission Problems, Bolshoy Karetny per. 19, Moscow, 127994, Russia}
\email{lenia.petrov@gmail.com}

\author[T. Sasamoto]{Tomohiro Sasamoto}
\address{T. Sasamoto,
Chiba University, Department of Mathematics, 1-33 Yayoi-cho, Inage, Chiba, 263-8522, Japan, and Zentrum Mathematik, Technische Universit\"at Mu\"nchen, Boltzmannstrasse 3 85748 Garching Germany}
\email{sasamoto@math.s.chiba-u.ac.jp}

\maketitle

\setcounter{tocdepth}{3}
\tableofcontents
\hypersetup{linktocpage}

%

\section{Introduction}

In this work we develop spectral theory for the $q$-Boson (stochastic) particle system\footnote{In \cite{SasWad} this particle system was referred to as the $q$-Boson totally asymmetric diffusion model while in \cite{BorCor,BCS,Hyun} it was also referred to as $q$-TAZRP. The term ``stochastic'' is included here to differentiate this with the non-stochastic quantum particle system considered in earlier work \cite{BBT,BIK} under the name $q$-Bosons. That earlier studied system is a special limit of the more general system presently considered (see Section \ref{vandiejen}). Despite this, in what follows we will generally suppress the term ``stochastic'', though still always referring to the stochastic particle system}. This is an interacting particle system whose generator is a stochastic representation of the generalization of the $q$-Boson Hamiltonian introduced by Sasamoto-Wadati in 1998 \cite{SasWad} (see Section \ref{afortiori} for more details). The system (in fact, a totally asymmetric zero range process) consists of $k\geq 1$ particles on $\Z$ with locations labeled by $\vec{n}=(n_1\geq \cdots \geq n_k)\in \Z^k$. In continuous time, each cluster of particles with the same location transfers one particle to the left by one at rate $(1-q^c)$, where $c$ is the size of the cluster and $q$ is a parameter fixed between 0 and 1. In order to preserve the ordering of $\vec{n}$, the highest index particle in a cluster is always the one which moves left.

This particle system can be understood as being a discrete space, $q$-deformation of the continuum delta Bose gas on $\R$ with attractive coupling constant (see Section \ref{deltabosesec}). This delta Bose gas has a rich history, going back to the foundational work of Lieb-Liniger in 1963 \cite{LL}, and it has recently played an important role in the physics literature \cite{McGuire,K,Dot,CDR,CDprl,ProS1,ProS2,CorwinQuastel,ProSpoComp,ImSa,ImSaKPZ,CDlong,ImSaKPZ2,Dot2,Dot3,Dot4,Dot5,ISS} surrounding the Kardar-Parisi-Zhang (KPZ) equation and universality class (which includes random growth models, interacting particle systems and directed polymers -- see the review \cite{ICReview}). More exactly the moments (for a fixed time $t$ but possibly different spatial locations $x_1,\ldots, x_k\in \R$) of the solution to the stochastic heat equation (whose logarithm is the KPZ equation that models a randomly growing interface) satisfy the delta Bose gas with initial data corresponding to the initial data of the stochastic heat equation. For more details see Section \ref{deltabosesec}.

The moment problem for the solution to the stochastic heat equation is not well-posed since its moments, though all finite, grow too fast to characterize the distribution of the solution. Despite this mathematical limitation, there has been a significant amount of non-rigorous work using these moment formulas to extract distributional information about the solutions to KPZ equation -- this sometimes goes by the name of the polymer replica method. In the instances for which rigorous results were available via other means \cite{ACQ,CQ,BCF,BCFV} it has been checked that these computations have yielded the correct answer. There are now many non-rigorous KPZ distribution computations -- such as those involving different types of initial data (flat / half-flat \cite{CDprl,CDlong}, stationary \cite{ImSa, ImSaKPZ,ImSaKPZ2}, or more general \cite{CorwinQuastel}) or different times \cite{Dot4} -- which are based on this technique and which do not yet have rigorous counterparts. Such computations involve some level of guessing (as to how to sum certain divergent series) which varies problem to problem, hence it is hard to be confident (let alone prove) whether the outcome of each additional computation will yield the correct answer.

One way to put this line of work on a firm, even rigorous footing is to find discrete regularizations of the KPZ equation which are integrable and well-posed in the sense that formulas for moments characterize the distribution. One such system is the $q$-deformed totally asymmetric simple exclusion process ($q$-TASEP) which was introduced by Borodin-Corwin \cite{BorCor} in 2011 via the framework of Macdonald processes. It was later observed by Borodin-Corwin-Sasamoto \cite{BCS} that moments of $q$-TASEP solve the $q$-Boson particle system with initial data corresponding to the initial condition of $q$-TASEP. Drawing on \cite{BorCor}, in \cite{BCS} this system was explicitly solved for one family of initial data (corresponding to $q$-TASEP started from half-stationary initial condition). Until now, it was not clear how to analyze this system for general initial data, as would be necessary to rigorously approach the many problems studied via the polymer replica method in the past few years. 

The reason why the delta Bose gas can be solved for general initial data is that it is known how to diagonalize the Hamiltonian. Though the eigenfunctions for this Hamiltonian have been known via coordinate Bethe ansatz \cite{Bethe} since the work of Lieb-Liniger \cite{LL}, alone they do not suffice to solve the time evolution equation. Ultimately this requires a Plancherel formula which shows how to decompose general initial data onto a subset of the algebraic eigenfunctions (and hence shows completeness of the Bethe ansatz as well). Various forms of such a result have been proved in \cite{Oxford,HO,ProSpoComp} for the delta Bose gas.

In this present work we prove a Plancherel formula for the $q$-Boson particle system. This allows us to explicitly compute the moments of $q$-TASEP for general initial data which will hopefully serve as the basis for further rigorous asymptotic work confirming the results of the previously made non-rigorous calculations about the KPZ equation. Our proof of the Plancherel formula is fairly simple, relying in part on the {\it contour shift argument} used in Heckman-Opdam's proof of the delta Bose gas Plancherel formula \cite{HO} (see Section \ref{plansec} for more history on this argument). In that previous work, the Hermiticity of the Hamiltonian played an important rule. In our work we consider a non-Hermitian Hamiltonian, however we find that a symmetry (known sometimes as PT-invariance or joint space-reflection and time-reversal symmetry) between our Hamiltonian and its adjoint is a suitable replacement for Hermiticity, see Remark \ref{commrem}.

A large part of the paper is devoted to modifications, implications and degenerations of this core result. For example, the Plancherel formula immediately implies a (spatial) orthogonality of the left and right eigenfunctions and suggests a second (spectral) orthogonality which we prove separately by ultimately appealing to the Cauchy-Littlewood formula for Hall-Littlewood (multivariate symmetric) polynomials (though we should emphasize that our eigenfunctions are not Hall-Littlewood polynomials and only degenerate to them in a certain limit). With minor modifications all results (though not necessarily the probabilistic interpretations) apply for more general parameter $q\in \C$. Such $q$ (especially on the complex unit circle) arise naturally from a quantum physics perspective  (see Section \ref{complexq}).

It is possible to take a limit of the $q$-Boson particle system so as to recover the continuum delta Bose gas and its Plancherel formula. There are also two discrete space limits we study -- one leads to a delta Bose gas and Plancherel formula previously studied by van Diejen \cite{vd} and another leads to the delta Bose gas and Plancherel formula related to the semi-discrete stochastic heat equation (or equivalently the O'Connell-Yor semi-discrete directed polymer partition function).

There exist other integrable discrete regularizations of the KPZ equation (and delta Bose gas) such as the asymmetric simple exclusion process (ASEP) or $q$-PushASEP \cite{BorPet,CorPet}. In a companion paper \cite{BCPS2} we prove (in a similar manner as here) a Plancherel formula for ASEP and the Heisenberg XXZ quantum spin chain on $\Z$ (thus recovering the XXZ results of \cite{BabThom, BabGut, Gut} in a different manner).


\subsection{Main results}\label{mainsec}

Let $\Weyl{k} = \Big\{\vec{n} = (n_1,\ldots,n_k)\in \Z^k\big\vert n_1\geq \cdots \geq n_k\Big\}$. The backward generator for the $q$-Boson particle system is written $\Abwd$ and defined via its action on functions $f:\Weyl{k}\to \C$ as
\begin{equation*}
\big(\Abwd f \big)(\vec{n}) = \sum_{i=1}^{M} (1-q^{c_i}) \big(f(\vec{n}_{c_1+\cdots +  c_{i}}^{-}) - f(\vec{n})\big).
\end{equation*}
Here $M$ denotes the number of clusters of $\vec{n}\in \Weyl{k}$ and $(c_1,\ldots, c_M)$ denote the sizes of these clusters (so that $n_1=\cdots=n_{c_1}>n_{c_1+1}=\cdots =n_{c_1+c_2}>\cdots >n_{c_1+\cdots +c_{M-1}+1} = \cdots = n_{c_1+\cdots +c_M}$) and $\vec{n}_{i}^{-}=(n_1,\ldots, n_{i}-1,\ldots, n_k)$. The forward generator of the $q$-Boson particle system is written $\Afwd$ and given by the matrix transpose of $\Abwd$ (see Definition \ref{bwdgendef} for its explicit action).

We will be concerned with eigenfunctions of $\Abwd$ and $\Afwd$. Define the function $C_{q}:\Weyl{k}\to \R$  by
$$
C_{q}(\vec{n}) = (-1)^k q^{-\frac{k(k-1)}{2}} \prod_{i=1}^{M} (c_i)!_q,
$$
where $(c_1,\ldots, c_M)$ are as above, and $(c)!_q = \prod_{j=1}^{c} \tfrac{1-q^j}{1-q}$ is the $q$-factorial. This is an invariant measure for the $q$-Boson particle system and the backward and forward generators are related by
\begin{equation}\label{ptsym}
\Abwd = (R C_q) \Afwd (R C_q)^{-1}
\end{equation}
where $\big(R f\big)(n_1,\ldots, n_k) = f( -n_k,\ldots, -n_1)$ and $C_q$ is the multiplication operator $\big(C_q f\big)(\vec{n}) = C_q(\vec{n})f(\vec{n})$. This relationship is sometimes called PT-invariance, see Remark \ref{commrem}.

For all $z_1,\ldots, z_k\in \C\setminus \{1\}$, set
\begin{eqnarray*}
\psibwd_{\vec{z}}(\vec{n})  &=& \sum_{\sigma\in S_k} \prod_{1\leq B<A\leq k} \frac{z_{\sigma(A)}-q z_{\sigma(B)}}{z_{\sigma(A)}- z_{\sigma(B)}} \, \prod_{j=1}^{k} (1-z_{\sigma(j)})^{-n_j},\\
\psifwd_{\vec{z}}(\vec{n}) &=& C_q^{-1}(\vec{n})\sum_{\sigma\in S_k} \prod_{1\leq B<A\leq k} \frac{z_{\sigma(A)}-q^{-1} z_{\sigma(B)}}{z_{\sigma(A)}- z_{\sigma(B)}} \, \prod_{j=1}^{k} (1-z_{\sigma(j)})^{n_j}.
\end{eqnarray*}

Our first result is that these are the right eigenfunctions of $\Abwd$ and $\Afwd$ (respectively).
\begin{proposition}[Proposition \ref{prop211} below]
For all $z_1,\ldots, z_k\in \C\setminus \{1\}$, we have
$$
\big(\Abwd \psibwd_{\vec{z}}\big)(\vec{n}) = (q-1)(z_1+\cdots +z_k) \psibwd_{\vec{z}}(\vec{n}),\qquad \big(\Afwd \psifwd_{\vec{z}}\big)(\vec{n}) = (q-1)(z_1+\cdots +z_k) \psifwd_{\vec{z}}(\vec{n}).
$$
\end{proposition}
The proof of this result is based on the fact that the backward and ($C_q$-conjugated) forward generators can be rewritten as free generators subject to $(k-1)$ two-body boundary conditions. The coordinate Bethe ansatz then readily implies the result. Since we are working on $\Z$ (as opposed to a finite or periodic interval) there are no Bethe equations to be solved.

Since $\psifwd_{\vec{z}}(\vec{n})$ is a right eigenfunction for the forward generator $\Afwd$, and $\Afwd$ is the matrix transpose of $\Abwd$, it follows that $\psibwd_{\vec{z}}(\vec{n})$ is also a left eigenfunction for $\Afwd$. Thus, we use the alternative notation
$$
\psil_{\vec{z}}(\vec{n}) = \psibwd_{\vec{z}}(\vec{n}), \qquad  \psir_{\vec{z}}(\vec{n}) = \psifwd_{\vec{z}}(\vec{n}).
$$

We now proceed to our main result -- a Plancherel type isomorphism theorem.

\begin{definition}
Let $\CP{k}$ be the space of functions $f:\Weyl{k}\to \C$ of compact support, and let $\LP{k}$ be the space of symmetric Laurent polynomials $G:\C^k\to \C$ in the variables $1-z_1,\ldots, 1-z_k$. Let $\gamma_1,\ldots,\gamma_k$ be positively oriented, closed contours chosen so that they all contain $1$, so that the $\gamma_A$ contour contains the image of $q$ times the $\gamma_B$ contour for all $B>A$, and so that $\gamma_k$ is a small enough circle around 1 that does not contain $q$ (see Figure \ref{circontours} in Section \ref{plansec} below).

Define the (symmetric) bilinear pairing $\llangle \cdot ,\cdot \rrangle$ on functions $f,g\in\CP{k}$ via
\begin{equation*}
\llangle f,g\rrangle = \sum_{\vec{n}\in \Weyl{k}} f(\vec{n})g(\vec{n}),
\end{equation*}
and the (symmetric) bilinear pairing $\llangle \cdot ,\cdot \rranglesmall$ on functions $F,G\in \LP{k}$ via
\begin{equation*}
\llangle F,G\rranglesmall = \sum_{\lambda\vdash k}\, \oint_{\gamma_k} \cdots \oint_{\gamma_k} d\mu_{\lambda}(\vec{w}) \prod_{j=1}^{\ell(\lambda)} \frac{1}{(w_j;q)_{\lambda_j}}  F(\vec{w}\circ\lambda) G(\vec{w}\circ \lambda).
\end{equation*}
Here $\lambda=(\lambda_1\geq \lambda_{2}\geq \cdots \geq 0)\vdash k$ is a partition of size $k$ (i.e., $\sum \lambda_i = k$) and
\begin{equation}\label{dmulambda}
d\mu_{\lambda}(\vec{w}) = \frac{(1-q)^{k}(-1)^k q^{-\frac{k^2}{2}}}{m_1! m_2!\cdots} \det\left[\frac{1}{w_i q^{\lambda_i} -w_j}\right]_{i,j=1}^{\ell(\lambda)} \prod_{j=1}^{\ell(\lambda)} w_j^{\lambda_j} q^{\frac{\lambda_j^2}{2}} \frac{dw_j}{2\pi \i};
\end{equation}
where $m_i=|\{j:\lambda_j=i\}|$, the $q$-Pochhammer symbol is $(a;q)_{n} = (1-a)(1-qa)\cdots (1-q^{n-1}a)$ and we use the notation
\begin{equation*}
\vec{w} \circ \lambda = (w_1,qw_1,\ldots, q^{\lambda_1-1}w_1, w_2, q w_2,\ldots, q^{\lambda_2-1}w_2,\ldots, w_{\lambda_{\ell}},q w_{\lambda_{\ell}},\ldots, q^{\lambda_{\ell}-1} w_{\lambda_{\ell}}).
\end{equation*}

\begin{remark}\label{smallremarkintro}
The pairing $\llangle \cdot,\cdot\rranglesmall$ on $\LP{k}$  has a simpler alternative definition:
\begin{equation*}
\llangle F,G\rranglesmall = \oint_{\gamma} \cdots \oint_{\gamma} d\mu_{(1)^k}(\vec{w}) \prod_{j=1}^{k} \frac{1}{1-w_j}  F(\vec{w}) G(\vec{w}),
\end{equation*}
where $\gamma$ can be chosen to be a circle containing both 1 and 0 and $(1)^k$ is the partition with $k$ ones. However, it turns out that the more involved definition above is more useful for our present purposes (in particular when taking various limits).
\end{remark}

The {\it $q$-Boson transform} $\FqBoson$ takes functions $f\in\CP{k}$ into functions $\FqBoson f\in~\LP{k}$ via
\begin{equation*}
\big(\FqBoson f\big)(\vec{z}) = \llangle f, \psir_{\vec{z}}\rrangle.
\end{equation*}

The (candidate) {\it $q$-Boson inverse transform} $\JqBoson$ takes functions $G\in\LP{k}$ into functions $\JqBoson G\in \CP{k}$ via
\begin{equation*}
\big(\JqBoson G\big)(\vec{n}) = \oint_{\gamma_1} \frac{dz_1}{2\pi \i} \cdots \oint_{\gamma_k} \frac{dz_k}{2\pi \i}  \prod_{1\leq A<B\leq k} \frac{z_A-z_B}{z_A-qz_B}\, \prod_{j=1}^{k} (1-z_{j})^{-n_j-1}\, G(\vec{z}).
\end{equation*}
By shrinking the nested contours so that they all lie upon $\gamma_k$ (cf. Lemma \ref{expandlem}) this can also be written as
\begin{equation}
\big(\JqBoson G\big)(\vec{n}) = \llangle \psil(\vec{n}),G\rranglesmall,
\end{equation}
where $\psil(\vec{n})$ is the function which maps $\vec{z}\mapsto \psil_{\vec{z}}(\vec{n})$.
\end{definition}

\begin{theorem}[Theorem \ref{isothm} below]
The $q$-Boson transform $\FqBoson$ induces an isomorphism between $\CP{k}$ and $\LP{k}$ with inverse given by $\JqBoson$. Moreover, for any $f,g\in \CP{k}$
\begin{equation*}
\llangle f,g \rrangle = \llangle \FqBoson(\mathcal{P} f),\FqBoson g\rranglesmall,
\end{equation*}
and for any $F,G\in \LP{k}$
\begin{equation*}
\llangle \mathcal{P}^{-1}(\JqBoson F),\JqBoson G \rrangle = \llangle F,G\rranglesmall.
\end{equation*}
Here $\mathcal{P}:\CP{k}\to \CP{k}$ is defined via its action $(\mathcal{P}g)(\vec{n}) = (-1)^k C_q(\vec{n})(R g)(\vec{n})$ and is the operator which maps right eigenfunctions to left eigenfunctions.
\end{theorem}

One immediate corollary of this Plancherel isomorphism theorem is the completeness of the coordinate Bethe ansatz.

\begin{corollary}[Corollary \ref{completenesscor} below]\label{compl}
For all $f\in \CP{k}$,
\begin{equation*}
f(\vec{n}) = \big(\JqBoson \FqBoson f\big)(\vec{n}) = \sum_{\lambda\vdash k}\, \oint_{\gamma_k} \cdots \oint_{\gamma_k} d\mu_{\lambda}(\vec{w}) \prod_{j=1}^{\ell(\lambda)} \frac{1}{(w_j;q)_{\lambda_j}}  \psil_{\vec{w}\circ\lambda}(\vec{n}) \llangle f,\psir_{\vec{w}\circ \lambda}\rrangle.
\end{equation*}
\end{corollary}

This shows that the space $\CP{k}$ is decomposed onto left (and right) eigenfunctions with spectral variables $\vec{w}\circ\lambda$ over all $\lambda\vdash k$ and $\vec{w}$ of length $\ell(\lambda)$, with respect to the (complex-valued) Plancherel measure given above in (\ref{dmulambda}). The lack of positivity of this measure should be compared to the case of Hermitian Hamiltonians such as the Heisenberg XXZ quantum spin chain on $\Z$ or delta Bose gas on $\R$. Thus our Plancherel isomorphism theorem is not an isomorphism between $L^2$ spaces. There are some degenerations of this system which admit $L^2$ space isomorphisms, namely the system considered in Sections \ref{vandiejen} and \ref{deltabosesec}. The system considered in Section \ref{semidiscsec} involves a complex-valued Plancherel measure, and likewise does not admit an $L^2$ isomorphism.

In applying a contour shift argument (used, for instance, in Heckman-Opdam's proof of the delta Bose gas Plancherel formula \cite{HO}) in this proof of the Plancherel formula, the non-Hermitian nature of our Hamiltonian introduces some additional non-triviality. However, the symmetry given earlier in (\ref{ptsym}) between $\Abwd$ and $\Afwd$ (see also Remark \ref{commrem}) is a suitable replacement for Hermiticity.

The Plancherel isomorphism theorem also contains within it the biorthogonality (with respect to the spatial variable $\vec{n}$ and spectral variable $\vec{w}$) of the left and right eigenfunctions (see Corollary \ref{north} and Proposition \ref{specorth} below).

The completeness result above along with the fact that $\psil_{\vec{z}}(\vec{n})$ is a right eigenfunction for $\Abwd$, enables us to explicitly solve the Kolmogorov backward equation for the $q$-Boson particle system with general initial data in $\CP{k}$.

The $q$-Boson particle system plays an analogous role to $q$-TASEP, as the continuum delta Bose gas plays to the Kardar-Parisi-Zhang equation (cf. Section \ref{deltabosesec}). More precisely, due to a Markov duality (see Section \ref{qtasepsec} or \cite{BCS}), if we write the trajectory of $q$-TASEP as $\vec{x}(t)$, then
\begin{equation*}
h(t;\vec{n}):=\EE\Big[ \prod_{i=1}^{k} q^{x_{n_i}(t)+n_i}\Big]
\end{equation*}
solves the $q$-Boson particle system Kolmogorov backward equation
$$
\frac{d}{dt} h(t;\vec{n}) = \big(\Abwd h\big)(t;\vec{n})
$$
with initial data depending on the initial condition $\vec{x}(0)$ for $q$-TASEP.

Since we know how to solve the backward equation (via the spectral decomposition above), this gives us access to exact formulas for moments of $q$-TASEP started from any initial data. However, such an exact formula does involve evaluating $\big(\FqBoson h_0\big)(\vec{z})$, where $h_0(\vec{n}) = h(0;\vec{n})$. For certain types of initial data this summation can be explicitly calculated (see for instance Corollary \ref{abovecor2} and Lemma \ref{belowlemma}). This is due to the fact that if $h_0(\vec{n}) = \big(\JqBoson G\big)(\vec{n})$ for some function $G(\vec{z})$, then the Plancherel isomorphism theorem implies that $\big(\FqBoson h_0\big)(\vec{z})= G(\vec{z})$. This provides a systematic way to discover and proof many new combinatorial formulas.

In particular, for $q$-TASEP with step initial condition (i.e., $x_i(0)=-i$ for $i\geq 1$), this approach leads to (see Proposition \ref{dualiyresult} below)
\begin{equation}\label{page6star}
\EE\Big[ \prod_{i=1}^{k} q^{x_{n_i}(t)+n_i}\Big]  = (-1)^k q^{\frac{k(k-1)}{2}} \oint_{\gamma_1} \frac{dz_1}{2\pi \i}\cdots \oint_{\gamma_k} \frac{dz_k}{2\pi \i} \prod_{1\leq A<B\leq k} \frac{z_A-z_B}{z_A-qz_B} \prod_{j=1}^{k} (1-z_j)^{-n_j} \frac{e^{(q-1)tz_j}}{z_j},
\end{equation}
where $\gamma_1,\ldots, \gamma_k$  are as above, with the additional condition that they do not include $0$. Proving this formula requires us to go beyond the functional spaces $\LP{k}$ and $\CP{k}$, which in this case is readily doable.

This moment formulas for step initial data has appeared previously (though not via this spectral approach). With all $n_i\equiv n$ it was first proved in \cite[Section 3.3]{BorCor} using the theory of Macdonald processes. The general $\vec{n}$ result was then proved in \cite[Theorem 2.11]{BCS} by simply guessing the above formula and checking that it satisfied the backward equation and initial data. The Macdonald process approach was then extended in \cite{BCGS} to also cover general $\vec{n}$. Extending from step to general initial data was unclear until the present work.

Since $q\in (0,1)$, the above moment formulas completely characterize the distribution of $q$-TASEP at time $t$. So far, this has been used to write a Fredholm determinant formula for the location of a particle $x_n(t)$ (cf. \cite[Theorem 3.2.11]{BorCor} or \cite[Theorem 3.12]{BCS}). In turn, this formula have served as effective starting point for asymptotic analysis of a variety of models related to $q$-TASEP (cf. \cite{BorCor,BCS, BCF,BCFV}, and related work \cite{BCR}). The direct asymptotics (for instance to the Tracy-Widom GUE distribution and KPZ universality class) of $q$-TASEP has not yet been performed, though it would seem that the above mentioned Fredholm determinant formula is well suited for that.

Thus, as pointed out earlier, $q$-TASEP and the $q$-Boson particle process serve as discrete regularizations of the Kardar-Parisi-Zhang equation and the continuum delta Bose gas, for which the physics replica method can be put on firm mathematical ground.

\subsection{Motivations}\label{motsec}

There are three related motivations which led to the present work: The polymer replica method from physics, Tracy-Widom's ASEP transition probability formulas, and measures on partitions and Gelfand-Tsetlin patterns (sequences of interlacing partitions) such as Schur, Whittaker and Macdonald processes.

\subsubsection{Polymer replica method}
We have already discussed much of our motivation coming from the polymer replica method. Let us just add that the connection between the moments of the stochastic heat equation and the delta Bose gas does not seem to be firmly established in the mathematical literature and some of the types of initial data one wishes to study for the KPZ equation (such as narrow wedge) fall outside of the realm of the presently proved Plancherel formulas. Thus, by working with a discrete regularization, one may hope to avoid such issues as well.

\subsubsection{Tracy-Widom's ASEP transition probability formulas}
In 2008, Tracy-Widom \cite{TW1} computed an exact formula for the transition probability of the $k$-particle ASEP. Though the coordinate Bethe ansatz was central to their work, they did not diagonalize the $k$-particle ASEP Hamiltonian. We address this diagonalization in the companion paper \cite{BCPS2} as well as the relationship between Tracy-Widom's work and the ASEP Plancherel formula. Tracy-Widom have recorded further progress in computing transition probabilities for variants of ASEP on $\Z$ (such as on $\Z_{\geq 0}$ \cite{TWhalfspace}). Lee \cite{Lee} has further developed Tracy-Widom's methods. Indeed, soon after the first posting of the present work, Korhonen-Lee \cite{Hyun} demonstrated how the Tracy-Widom approach can be employed to find transition probabilities for the $q$-Boson particle system (see also Section \ref{transprobsec} below). It would be reasonable to try to develop analogous results to ours in all of these other settings.

\subsubsection{Measures on partitions and Gelfand-Tsetlin patterns}
Measures on partitions and Gelfand-Tsetlin patterns have played an important role in a wide variety of probabilistic systems including random matrix theory, random growth processes, interacting particle systems, directed polymer models, and random tilings (see, for example, the recent review \cite{BorGor} and references therein). During the last 15 years there was significant activity in the analysis of determinantal measures -- in particular the so-called Schur processes \cite{Ok, OkResh} which are written in terms of Schur symmetric functions. It is possible to define generalizations of the Schur processes by replacing Schur symmetric functions with Macdonald symmetric functions (their two-parameter $(q,t)$-generalizations). Until recently, little was done with respect to these generalizations, owing largely to the fact that they no longer are determinantal.

In 2009, O'Connell \cite{OCon} introduced the Whittaker process and related it to the O'Connell-Yor semi-discrete directed polymer partition function. Furthermore, he utilized a Whittaker function integral identity to compute an exact formula for the Laplace transform of the partition function. As Whittaker functions are limits (for $t=0$ and $q\to 1$) of Macdonald symmetric functions \cite{GLOqlim}, this development was a source of motivation for Borodin-Corwin's \cite{BorCor} subsequent 2011 work on Macdonald processes.

The work of Borodin-Corwin \cite{BorCor} gives answers to two basic questions about Macdonald processes (and their various degenerations): what is their probabilistic content and how can one compute meaningful information and asymptotics. In order to endow the Macdonald processes with probabilistic content, \cite{BorCor} showed how to construct Markov dynamics which map Macdonald processes to other Macdonald processes (with evolved sets of parameters). The construction in \cite{BorCor} is an adaption of one from the Schur process context due to Borodin-Ferrari \cite{BF, twosides} (and based on an idea of Diaconis-Fill \cite{DiaconisFill}). There are other Markov dynamics (at various levels of degeneration of Macdonald processes) besides those considered in \cite{BorCor} -- see \cite{BorPet, OConPei, OCon, COSZ}. The $q$-TASEP was discovered as a one-dimensional marginal of the Markov dynamics corresponding to setting the Macdonald parameters $t=0$, and $q\in (0,1)$. When $q\to 0$, $q$-TASEP becomes the usual TASEP, and the above Macdonald process turns into a Schur process.

The second challenge in studying Macdonald processes was to compute meaningful information and asymptotics. In \cite{BorCor} this was accomplished by directly appealing to the integrable structure of the Macdonald symmetric polynomials -- in particular their explicit eigenfunction relationships with the Macdonald difference operators. This structure naturally led to nested contour integral formulas for expectations of various observables of the Macdonald processes. With regards to $q$-TASEP, this provided formulas for the expectation $h(t;\vec{n})=\EE[\prod_{j=1}^{k} q^{x_{n_j}(t)+n_j}]$ like (\ref{page6star}) above.

As a stochastic process, $q$-TASEP has a scaling limit to the logarithm of the semi-discrete stochastic heat equation (or free energy of the O'Connell-Yor semi-discrete directed polymer) and further to the KPZ equation (logarithm of continuum stochastic heat equation). The limits of the observables studied for $q$-TASEP correspond to products of values of the solution to these stochastic heat equations at fixed time and varying spatial location, and thus \cite{BorCor} found similar nested contour integral formulas for these limits. As explained earlier, for the stochastic heat equation, these observables satisfy the delta Bose gas and have a significant literature surround them. The nested contour integral formulas for the delta Bose gas seem to have first appeared in 1985 physics work of Yudson \cite{Yudson} and reemerged independently in 1997 work of Heckman-Opdam \cite{HO} proving the Plancherel formula for the delta Bose gas.

A limitation of the Macdonald process approach to studying $q$-TASEP is that it only works for initial conditions of $q$-TASEP which arise as marginals of Macdonald processes. There are certainly many examples of such initial conditions (an infinite-dimensional family that, in particular, includes the step initial condition). Generic initial conditions, and certain important ``examples'' of initial conditions (such as flat or half-flat) do not arise in this manner and hence cannot be treated via this method.

Inspired by the connection between the nested contour integral formulas and the delta Bose gas, Borodin-Corwin-Sasamoto \cite{BCS} discovered that it was possible to implement the ideas of the polymer replica method at the level of $q$-TASEP (and ASEP). In particular, they showed that for any initial conditions, the expectations $h(t;\vec{n})$ solve a discrete $q$-deformed (or regularized) version of the delta Bose gas in which the $q$-Boson particle system generator replaces the delta Bose gas Hamiltonian. The initial data for the $q$-Boson particle system backward equation directly corresponds with the initial condition for $q$-TASEP, and \cite{BCS} showed that for step and half-stationary $q$-TASEP initial conditions the solution $h(t;\vec{n})$ was given by nested contour integral formulas. The step solution came directly from \cite{BorCor} while the half-stationary did not (though should correspond to a two-sided generalization of Macdonald processes such as developed in the Schur context in \cite{twosides}).

The question of how to solve the $q$-Boson particle system for general initial data (and hence solve for $h(t;\vec{n})$ for $q$-TASEP with general initial conditions) is the central motivation of the present work. The Plancherel formula and its various consequences provide a solution to this question.

\subsubsection{Algebraic motivations}\label{afortiori}
A fortiori, this work provides some motivation to better understand the relationship between various algebraic structures. One question which was answered in \cite{BCdiscrete} was how the theory of Macdonald processes implies the fact that $h(t;\vec{n})$ solves the $q$-Boson particle system. This relationship turned out to be a consequence of  a commutation relation satisfied by the Macdonald first order difference operators. However, there remain a few algebraic mysteries which we now touch upon.

The $q$-Boson particle system on a periodic lattice arises as a stochastic representation of the $q$-Boson Hamiltonian considered in \cite{SasWad}. This Hamiltonian is a generalization of an earlier version considered by Bogoliubov-Bullough-Timonen \cite{BBT} and Bogoliubov-Izergin-Kitanine \cite{BIK}. The algebraic $q$-Boson Hamiltonian from \cite{SasWad} is written as
\begin{equation*}
\HqBoson = - \sum_{j=1}^{M} (B^{\dag}_{j-1} +\gamma B^{\dag}_j) B_j.
\end{equation*}
Here $\{B_j\}_{j=1}^{M}$, $\{B^{\dag}_j\}_{j=1}^{M}$ and another set $\{N_j\}_{j=1}^{M}$ belong to the $q$-Boson algebra generated by the relations
$$
[N_j,B_k^{\dag}] = B_j^{\dag} \bfone_{j=k}, \qquad [N_j,B_k] = -B_j \bfone_{j=k}, \qquad [B_j,B_k^{\dag}] = q^{N_j} \bfone_{j=k}.
$$
One imposes periodic (of length $M$) boundary conditions by defining $B_{0}^{\dag} = B_M$.

A periodic version of the $q$-Boson particle system backward generator corresponds to a particular representation of $-(1-q)\HqBoson$ in which the operators act on functions $f:(\Z_{\geq 0})^{M}\to \C$ as
$$
\big(B_j^{\dag} f\big) (\vec{\tazrp}) = f(\vec{\tazrp}_{j}^{+}), \qquad \big(B_j f\big)(\vec{\tazrp}) = \frac{1-q^{\tazrp_j}}{1-q} f(\vec{\tazrp}_{j}^{-}), \qquad \big(Nf\big)(\vec{\tazrp})= \tazrp_j f(\vec{\tazrp}).
$$
One checks that these operators satisfy the $q$-Boson algebra. The operator $-(1-q)\HqBoson$ acts on functions $f:(\Z_{\geq 0})^{M}\to \C$ as
$$
\Big(-(1-q)\HqBoson f\Big)(\vec{\tazrp}) = \sum_{j=1}^{M} (1-q^{\tazrp_j}) \big(f(\vec{\tazrp}^{j,j-1}) - f(\vec{\tazrp})\big)
$$
where $\vec{\tazrp}^{j,j-1} = (\tazrp_1,\ldots,\tazrp_{j-1}+1,\tazrp_j-1,\ldots, \tazrp_M)$, and where $\vec{\tazrp}^{1,0} = (\tazrp_1-1,\ldots, \tazrp_M+1)$. This is the backward generator of the $q$-Boson particle system (see Section \ref{stochrepsec}) on a periodic portion of $\Z$.

This algebraic $q$-Boson Hamiltonian arises from the $L$-matrix given in (5.1) of \cite{SasWad} and its integrability (in the sense of satisfying the Yang-Baxter equations) was shown there as well. The version studied earlier in \cite{BBT,BIK} corresponds to $\gamma=0$ and is only stochastic\footnote{By a stochastic Hamiltonian, we mean a linear operator in the space of functions on the (discrete) state space that, when viewed as a matrix with rows and columns indexed by the states of the system, has nonnegative off-diagonal matrix elements, and that maps the constant functions to zero. Under certain regularity assumptions (which we do not check here), such a Hamiltonian generates a continuous time Markov jump process on the state space, see e.g. \cite{Lig}.} in the limit when $q\to 1$. Due to our probabilistic motivations, we primarily consider $\gamma=-1$ here, though in Section \ref{epdef} we consider the general $\gamma$ Hamiltonian (under the identification of $\gamma=-\e$) and show how our results extend. The spectral theory for the periodic $\gamma=0$ case has received attention recently in \cite{Tsilevich,VDcircle,Korff}, and likewise for the infinite lattice (as considered herein) $\gamma=0$ case in \cite{vd,VDEE,VDEE2}. In these cases, Hall-Littlewood polynomials play the role of left and right eigenfunctions.

Another algebraic mystery has to do with two sets of symmetric functions. The Macdonald process related to $q$-TASEP is defined in terms of Macdonald polynomials at $t=0$ and $q\in (0,1)$ (these are sometimes called $q$-Whittaker functions \cite{GLO}). Via the connection to the $q$-Boson particle system we naturally arrive at another set of symmetric polynomials -- the eigenfunctions for the $q$-Boson particle system -- which appear to be quite remarkable and (though not the same) are close relatives of Hall-Littlewood symmetric polynomials (and degenerate to them in the limit corresponding to $\gamma=0$ above). We do not presently have an algebraic explanation for this transition from one type of symmetric polynomials to another. Perhaps a further investigation into the relation of $q$-Bosons to degenerations of the double affine Hecke algebra might shed light on this (see developments in this direction in \cite{HO,vd,Korff,takeyama} and references therein).

The work of \cite{BorCor,BCS} showed that for the particular step and half-stationary initial conditions, the relationship between $q$-TASEP and the $q$-Boson particle system persists when inhomogeneous particle jumping rate parameters $a_i$ are introduced. These $a$-parameters have an important meaning from the Macdonald perspective as the variables of the Macdonald polynomials. From the perspective of the $q$-Boson particle system, \cite{BCS} showed that the backward generator (similar holds true for the conjugated forward generator) was equivalent to a free generator with inhomogeneous rates depending on the $a$'s and the same two-body boundary condition as occurs when all $a_i \equiv 1$. Despite this evidence, we do not presently know how to develop a general $a$-parameter version of the eigenfunctions and Plancherel formula. One reason why a general $a$-version of our present result could be quite enlightening is that we do not have an analog of Macdonald processes or $a$-parameters in relation to ASEP and the Heisenberg XXZ quantum spin chain, and this may help guide that investigation.

In \cite{BCdiscrete} two discrete time versions of $q$-TASEP were introduced and studied. Their moments solve discrete variants of the $q$-Boson particle system evolution equations. It would be interesting to develop the parallel of the results of this paper in the discrete-time context. There is another generalization of $q$-TASEP which is called $q$-PushASEP (cf. \cite{BorPet,CorPet}). It seems likely that the analog of the $q$-Boson particle system for these variants of $q$-TASEP should have the eigenfunctions as above (with different eigenvalues). As explained in \cite{SasWad}, this should also hold true for any operator which arises out of specializing the spectral parameter of the $q$-Boson transfer matrix $\tau(u)$ (since transfer matrices commute for different spectral parameters). It is reasonable, then, to hope to show that these discrete systems arise from the transfer matrix.

At the same time as the present work was posted, Povolotsky \cite{Pov2} has posted a work in which he introduces a class of models solvable via coordinate Bethe ansatz which includes the $q$-Boson particle system (and also those of \cite{BCdiscrete}) as special cases (see also earlier related work \cite{PPH,Pov1}). Similar to here, he computes the eigenfunctions via coordinate Bethe ansatz. He also conjectures their completeness and orthogonality. Section \ref{appsec} below shows exactly how our Plancherel theorem proves completeness and biorthogonality for the $q$-Boson particle system. We hope our methods will extend to the full class of models discussed above.

\subsection{Notation}\label{notations}
We collect many of the definitions and notations used throughout the paper (including some already discussed earlier in the introduction).

Define
\begin{equation*}
\Weyl{k} = \Big\{\vec{n} = (n_1,\ldots,n_k)\in \Z^k\big\vert n_1\geq \cdots \geq n_k\Big\}.
\end{equation*}

For $\vec{n}\in \Z^k$ set
\begin{equation*}
\vec{n}_i^{\pm} = (n_1, \cdots , n_{i}\pm 1, \cdots , n_k).
\end{equation*}
The backward difference operator $\difbwd$ and forward difference operator $\diffwd$ act on functions $f:\Z\to \C$ as
\begin{equation*}
\big(\difbwd f \big)(n) = f(n-1)-f(n),\qquad \big(\diffwd f \big)(n) = f(n+1)-f(n).
\end{equation*}
For a function $f:\Z^k\to \C$, and $1\leq i\leq k$, we write $\difbwd_i$ and $\diffwd_i$ as the application of the respective operator in the variable $n_i$.

The space-reflection operator $R$ acts on functions $f:\Z\to \C$ as
\begin{equation*}
\big(R f\big)(n_1,\ldots, n_k) = f( -n_k,\ldots, -n_1).
\end{equation*}

Define a symmetric bilinear pairing $\llangle \cdot ,\cdot \rrangle$ on functions $f,g:\Weyl{k}\to \C$ via
\begin{equation*}
\llangle f,g\rrangle = \sum_{\vec{n}\in \Weyl{k}} f(\vec{n})g(\vec{n})
\end{equation*}
(assuming the series converges) and note that $\llangle f,g\rrangle = \llangle f h,g h^{-1} \rrangle$ for any function $h\neq 0$, and that
\begin{equation}\label{bilinflip}
\llangle R f,R g\rrangle = \llangle f,g\rrangle.
\end{equation}

To $\vec{n}\in \Weyl{k}$ of the form
$$n_1= \cdots = n_{c_1} > n_{c_1+1}= \cdots = n_{c_1+c_2} > \cdots >n_{c_1+\cdots +c_{M-1}+1} = \cdots = n_{c_1+\cdots c_M}$$
we associate the list of ordered cluster sizes $\vec{c}=(c_1,\ldots, c_M)$ (where $M$ is the number of clusters). We also associate to $\vec{n}$ the ordered list $\vec{g} = (g_1,\ldots ,g_M)$ of gaps between the clusters of $\vec{n}$. Thus for $2\leq i\leq M$, $g_i= n_{c_{i-1}}-n_{c_{i}}$ and by convention $g_1=+\infty$. We write these functions of $\vec{n}$ as $\vec{c}(\vec{n})$, $M(\vec{n})$ and $\vec{g}(\vec{n})$. To illustrate, for $\vec{n}= (2,1,-2,-2,-2)$, $c(\vec{n}) = (1,1,3)$, $M(\vec{n})=3$ and $g(\vec{n}) = (+\infty,1,3)$. The pair $(\vec{c},\vec{g})$ identifies $\vec{n}$ modulo a global shift.

A partition $\lambda=(\lambda_1\geq \lambda_2\geq \cdots \geq 0)$ is a weakly decreasing set of non-negative integers with finitely many nonzero entries. The size of $\lambda$ is $|\lambda|=\sum_i \lambda_i$ and the length of $\lambda$ is $\ell(\lambda)=|\{i:\lambda_i\neq 0\}|$. If $|\lambda|=k$ we write $\lambda\vdash k$. To a partition we associate its multiplicities $m_i = |\{j:\lambda_j=i\}|$ and sometimes write $\lambda = 1^{m_1}2^{m_2}\cdots$. For a partition $\lambda\vdash k$ of length $\ell$ and a set of variables $\vec{w}=w_1,\ldots, w_\ell$ we define
\begin{equation}\label{wlambda}
\vec{w} \circ \lambda = (w_1,qw_1,\ldots, q^{\lambda_1-1}w_1, w_2, q w_2,\ldots, q^{\lambda_2-1}w_2,\ldots, w_{\lambda_{\ell}},q w_{\lambda_{\ell}},\ldots, q^{\lambda_{\ell}-1} w_{\lambda_{\ell}}).
\end{equation}
We will also use an additive version of $\vec{w}\circ\lambda$ defined as
\begin{equation}\label{wlambdasd}
\vec{w} \tilde{\circ} \lambda = (w_1,w_1+1,\ldots, w_1+\lambda_1-1, w_2, w_2+1,\ldots, w_2+\lambda_2-1,\ldots, w_{\lambda_{\ell}},w_{\lambda_{\ell}}+1,\ldots, w_{\lambda_{\ell}}+\lambda_{\ell}-1).
\end{equation}
In general, an arrow over a variable (such as $\vec{z}$) denotes a vector $\vec{z}=(z_1,\ldots, z_k)$ whose length is generally $k$ (or else explicitly stated).

We define the space $\CP{k}$ of compactly supported functions $f:\Weyl{k}\to \C$. We also define the space $\LP{k}$ of symmetric Laurent polynomials in $1-z_1,\ldots, 1-z_k$. In particular, any $G\in \LP{k}$ can be written as
$$
\sum_{\vec{n}\in \Weyl{k}} \sum_{\sigma\in S_k} a_{\vec{n}} \prod_{i=1}^{k} (1-z_{\sigma(i)})^{n_i}
$$
or as
$$\frac{1}{\V(\vec{z})}
\sum_{\vec{n}\in \Weyl{k}} \sum_{\sigma\in S_k} \sgn(\sigma) b_{\vec{n}} \prod_{i=1}^{k} (1-z_{\sigma(i)})^{n_i}
$$
with all but finitely many coefficients $a_{\vec{n}},b_{\vec{n}}\in \C$ fixed to be zero, and $\V(\vec{z}) = \prod_{A<B} (z_A-z_B)$ is the Vandermonde determinant.

Throughout we assume that the parameter $q\in (0,1)$, though in Section \ref{complexq} we explain how our results (or modifications of them) hold for more general $q\in \C$. The $q$-Pochhammer symbol and $q$-factorial are defined as
\begin{equation*}
(a;q)_{n} = \prod_{i=0}^{n-1} (1-q^i a), \qquad n!_q = \frac{(q;q)_n}{(1-q)^n}.
\end{equation*}
The $q\to 1$ limit of $(a;q)_n$ is written $(a)_n = a (a+1)\cdots (a+n-1)$.

The symmetric group on $k$ elements is denoted by $S_k$, the indicator function for an event $E$ by $\bfone_{E}$, and the identity operator by $\Id$.

\subsection{Acknowledgements}

This work was started at the Simons Symposium on the Kardar-Parisi-Zhang Equation and advanced during the Lebesgue Centre's summer school on the KPZ equation and Rough Paths and the Cornell probability summer school. IC appreciates discussions on Bethe ansatz for non-Hermitian Hamiltonians with Gunter Sch\"{u}tz while visiting the Hausdorff Institute. The authors all appreciate illuminating comments and discussions of a draft of this work with Jan Felipe van Diejen, Erdal Emsiz, Christian Korff, Alexander Povolotsky and Gunter Sch\"{u}tz.

AB was partially supported by the NSF grant DMS-1056390. IC was partially supported by the NSF through DMS-1208998 as well as by Microsoft Research through the Schramm Memorial Fellowship, and by the Clay Mathematics Institute through a Clay Research Fellowship.  LP was partially supported by the RFBR-CNRS grant 11-01-93105. TS was supported by KAKENHI (22740054).

\subsection{Outline}

In Section \ref{qBosonsyst} we use the coordinate Bethe ansatz to determine algebraic eigenfunctions for the forward and backward generators of the $q$-Boson particle system. The main result of this paper -- the Plancherel isomorphism theorem involving these eigenfunctions -- is proved in Section \ref{plansec}. That section also contains results about the completeness and biorthogonality of these eigenfunctions. Section \ref{appsec} applies the Plancherel formula in order to solve the Kolmogorov forward and backward equations, and as an application compute moments of the $q$-TASEP. Section \ref{complexq} records how the earlier results persist when $q$ is complex. Section \ref{semidisc} describes two limits of the $q$-Boson particle system, one of which relates to earlier work of van Diejen \cite{vd} and the other of which relates to the semi-discrete stochastic heat equation.
Finally, in Section \ref{appendixsec} (Appendix) we briefly touch on the continuum limit of the work of this paper (as relates to the continuum delta Bose gas and stochastic heat equation), prove a proposition related to the expansion of nested contour integrals, and provide a direct combinatorial proof of Lemma \ref{belowlemma}. 

\section{Coordinate Bethe ansatz}\label{qBosonsyst}

In this section we introduce the $q$-Boson particle system and use the Bethe ansatz to write down algebraic eigenfunctions for its forward and backward generators. The main idea here is that it is possible (and straightforward) to rewrite the generator of the system in terms of a constant coefficient, separable {\it free generator} with $(k-1)$ two-body boundary conditions imposed. The coordinate Bethe ansatz is then readily applied to provide the desired eigenfunctions. It is worth recalling that since we work on $\Z$ (as opposed to a finite or periodic interval) there are no Bethe equations which must be solved in computing the algebraic eigenfunctions. The usefulness of this approach for interacting particle systems goes back to the work of \cite{Schutz,AKK}.

\subsection{The $q$-Boson particle system}\label{stochrepsec}

The $q$-Boson particle system is a stochastic interacting particle system $\vec{\tazrp}(t)$ (in fact, a certain totally asymmetric zero range process) with state space $\vec{\tazrp}(t)\in (\Z_{\geq 0})^{\Z}$ (see \cite{Lig} for general background). Here $\tazrp_i(t)$ represents the non-negative number of particles at a given site $i\in\Z$. In continuous time, independently for each $i\in \Z$, one particle may move from site $i$ to site $i-1$ at rate $1-q^{\tazrp_i(t)}$ (i.e., according to an exponential waiting time of that rate). This corresponds to changing $(\tazrp_{i-1}(t),\tazrp_i(t))$ to $(\tazrp_{i-1}(t)+1,\tazrp_{i}(t)-1)$.

We will focus on the restriction of this process to a finite number $k\geq 1$ of particles. Since the dynamics conserve the number of particles, this restriction is itself a stochastic interacting particle system (see Figure \ref{qBoson} for an illustration of $k=9$).

To any $\vec{\tazrp}\in (\Z_{\geq 0})^{\Z}$ such that $\sum_{i\in \Z}\tazrp_i = k$ we may associate a (weakly) ordered list of particle locations $\vec{n}=(n_1\geq \cdots \geq n_k)\in \Weyl{k}$. We write $\vec{n}(\vec{\tazrp})$ to denote this function. For example, if $k=5$ and $\tazrp_i\equiv 0$ except for $\tazrp_{-2} = 3$, $\tazrp_{1}=1$ and $\tazrp_{2}=1$, then $\vec{n}(\vec{\tazrp}) = (2,1,-2,-2,-2)$. Thus, we may write $\vec{n}(t) = \vec{n}(\tazrp(t))$ as the $q$-Boson particle system at time $t$ in these coordinates.

We may write, in these $\vec{n}$-coordinates, the backward and forward generators for the $k$ particle restriction of the $q$-Boson particle system (see also Section \ref{appsec}). In what follows we write $\vec{c} = \vec{c}(\vec{n})$, $M=M(\vec{n})$, and $\vec{g}= \vec{g}(\vec{n})$ (recall from Section \ref{notations}) thus suppressing the $\vec{n}$ dependence.

\begin{figure}
\begin{center}
\includegraphics[scale=1]{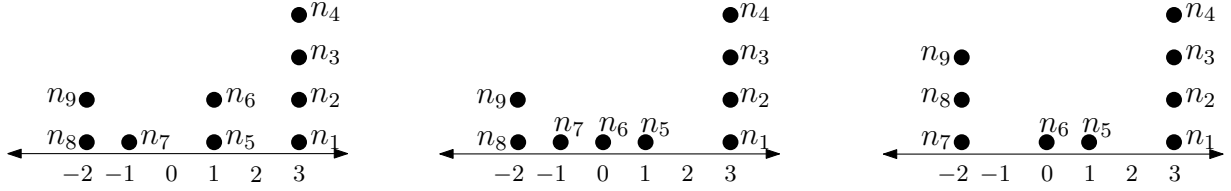}
\end{center}
\caption{Three states of the $q$-Boson particle system with 9 particles. In the left-most state, $\tazrp_3=4$, $\tazrp_2=0$, $\tazrp_1=2$, $\tazrp_0=0$, $\tazrp_{-1}=1$ and $\tazrp_{-2}=2$. The value of $\vec{n}$ associated to this state is given by reading off the location of the labels $n_1,\ldots, n_9$. In the middle state, the particle corresponding with $n_6$ has decreased location by 1 (as occurs at rate $1-q^2)$, and in the right-most state the particle corresponding with $n_7$ has decreased location by 1 (as occurs at rate $1-q$).}\label{qBoson}
\end{figure}

\begin{definition}\label{bwdgendef}
The $q$-Boson backward generator acts on functions $f:\Weyl{k}\to \C$ as
\begin{equation*}
\big(\Abwd f \big)(\vec{n}) = \sum_{i=1}^{M} (1-q^{c_i}) \big(f(\vec{n}_{c_1+\cdots +  c_{i}}^{-}) - f(\vec{n})\big)
\end{equation*}
where $\vec{n}_j^-$ is defined in Section \ref{notations}.

The $q$-Boson forward generator (the matrix transpose of $\Abwd$) acts on function $f:\Weyl{k}\to \C$ as
\begin{equation}\label{Afwddef}
\big(\Afwd f \big)(\vec{n}) = \sum_{i=1}^{M} \left( \left( (1-q^{c_{i-1}+1}) \bfone_{g_i=1} + (1-q) \bfone_{g_i>1}\right) f(\vec{n}_{c_1+\cdots+c_{i-1}+1}^{+})  - (1-q^{c_i}) f(\vec{n})\right)
\end{equation}
where by convention, for $i=1$ we set $c_{i-1}+1=c_1+\cdots + c_{i-1}+1=1$.

The function $C_{q}:\Weyl{k}\to \R$ depends only on the list $\vec{c}(\vec{n})=(c_1,\ldots, c_M)$ of cluster sizes for $\vec{n}$ via
\begin{equation}\label{Cqdef}
C_{q}(\vec{n}) = (-1)^k q^{-\frac{k(k-1)}{2}} \prod_{i=1}^{M} (c_i)!_q.
\end{equation}
Let $C_q^{-1}(\vec{n})$ denote $(C_q(\vec{n}))^{-1}$. We will also write $C_q$ and $C_q^{-1}$ as (diagonal) multiplication operators defined so that $(C_qf)(\vec{n}) = C_q(\vec{n}) f(\vec{n})$ and $(C_q^{-1}f)(\vec{n}) = C_q^{-1}(\vec{n}) f(\vec{n})$. One readily sees that $C_q$ and the space-reflection operator $R$ commute.

The $q$-Boson conjugated forward generator is defined as
$$\Amfwd = C_q \Afwd C_q^{-1}.$$
\end{definition}

\begin{lemma}\label{equalAfwdlem}
We have that
\begin{equation}\label{equalAfwd}
\big(\Amfwd f\big)(\vec{n}) = \sum_{i=1}^{M} (1-q^{c_i}) \left(f(\vec{n}_{c_1+\cdots+c_{i-1}+1}^{+}) - f(\vec{n})\right).
\end{equation}
\end{lemma}
\begin{proof}
Since $\Amfwd$ is defined via conjugation of $\Afwd$ by the diagonal matrix $C_q$, it suffices to check the coincidence of the off-diagonal matrix elements of $\Amfwd$ with those of the right-hand side of (\ref{equalAfwd}).

Observe that if $g_i=1$ then $\vec{c}(\vec{n}_{c_1+\cdots+c_{i-1}+1}^{+}) = (c_1,\ldots,c_{i-2}, c_{i-1}+1,c_{i}-1,c_{i+1},\ldots)$. From the definition of $C_q$ it then follows that
$$
\frac{C_q(\vec{n})}{C_q(\vec{n}_{c_1+\cdots+c_{i-1}+1}^{+})} = \frac{1-q^{c_i}}{1-q^{c_{i-1}+1}}.
$$
Similarly, if $g_i>2$ then $\vec{c}(\vec{n}_{c_1+\cdots+c_{i-1}+1}^{+}) = (c_1,\ldots,c_{i-2}, c_{i-1}, 1,c_{i}-1,c_{i+1},\ldots)$. Then
$$
\frac{C_q(\vec{n})}{C_q(\vec{n}_{c_1+\cdots+c_{i-1}+1}^{+})} = \frac{1-q^{c_i}}{1-q}.
$$
These are exactly the ratios of the off-diagonal matrix elements of (\ref{Afwddef}) and (\ref{equalAfwd}).
\end{proof}

Let us note that in the above proof the constant $(-1)^k q^{-\frac{k(k-1)}{2}}$ in the definition of $C_q$ played no role. This choice of constant is fixed in the proof of Theorem \ref{KqBosonId} below, by (\ref{kqbosonthmend}) and the computation following it. Another choice of constant would result in proving that $\KqBoson$ acts as a constant (not equal to one) times the identity operator.

\begin{remark}\label{commrem}
Lemma \ref{equalAfwdlem} implies that
$R^{-1} \Abwd R = \Amfwd = C_q \Afwd C_q^{-1}$ or equivalently
$$
\Abwd = (R C_q) \Afwd (R C_q)^{-1}
$$
showing that $\Abwd$ and $\Afwd$ are related via a similarity transform. Up to a constant, the function $C_q(\vec{n})$ is the invariant measure for the $q$-Boson particle process (see \cite[Corollary 3.3.12]{BorCor}). In the study of interacting particle systems, the conjugated forward generator $\Amfwd$ is sometimes known as the adjoint or time-reversed generator. The above relationship shows that our present model is PT-invariant, i.e., invariant under joint space-reflection and time-reversal. This type of invariance is of fundamental significance in many areas of physics and features prominently in some papers on interacting particle systems such as \cite{GJL,TS}. It may be possible to extend the present approach to similarly study and develop spectral theory for other PT-invariant Bethe ansatz solvable models. 
We are grateful to Gunter Sch\"{u}tz for noting that the above relation is equivalent to the PT-invariance of the model.

PT-invariance is a property which holds true for all spatially homogeneous nearest neighbor asymmetric zero range processes, so long as $C_q$ is replaced by the product invariant measure for the process. The proof of this general fact follows the same lines as Lemma \ref{equalAfwdlem}.
\end{remark}


\subsection{Free generators with two-body boundary conditions}\label{freegensec}
The key to finding the eigenfunctions for the backward and conjugated forward generators is to reexpress them in terms of constant-coefficient, separable {\it free generators} subject to $(k-1)$ two-body boundary conditions. It is certainly not typical that the generators of an interacting particle system can be reexpressed in this manner (see however \cite{TW1,BCS} for ASEP). This is the hallmark of integrable systems of Bethe-ansatz type.

Recall the definitions of $\difbwd_i$ and $\diffwd_i$ from Section \ref{notations}.
\begin{definition}\label{freedef}
The $q$-Boson backward free generator $\Freebwd$ acts on functions $u:\Z^k\to \C$ as
\begin{equation}\label{freebwddef}
\big(\Freebwd u \big)(\vec{n}) = (1-q) \sum_{i=1}^{k} \big(\difbwd_i u\big)(\vec{n}).
\end{equation}
We say that the function $u:\Z^k\to \C$  satisfies the $(k-1)$ $q$-Boson backward two-body boundary conditions if
\begin{equation}\label{star1}
\textrm{for all } 1\leq i\leq k-1\qquad \big(\difbwd_i - q\difbwd_{i+1}\big)u\big\vert_{\vec{n}:n_i=n_{i+1}} \equiv 0.
\end{equation}
The $q$-Boson forward free generator $\Freefwd$ acts on functions $u:\Z^k\to \C$ as
\begin{equation}\label{freefwddef}
\big(\Freefwd u\big)(\vec{n}) =(1-q) \sum_{i=1}^{k} \big(\diffwd_i u\big)(\vec{n}).
\end{equation}
We say that the function $u:\Z^k\to \C$  satisfies the  $(k-1)$ $q$-Boson forward two-body boundary conditions
\begin{equation}\label{star2}
\textrm{for all } 1\leq i\leq k-1\qquad \big(q\diffwd_i - \diffwd_{i+1}\big)u\big\vert_{\vec{n}:n_i=n_{i+1}} \equiv 0.
\end{equation}
\end{definition}

\begin{proposition}\label{freetrueequiv}
If $u:\Z^k\to \C$ satisfies the $(k-1)$ $q$-Boson backward (respectively, forward) two-body boundary conditions, then for $\vec{n}\in \Weyl{k}$,
\begin{equation*}
\big(\Freebwd u\big)(\vec{n}) = \big(\Abwd u\big)(\vec{n})\qquad \big(\textrm{respectively, } \big(\Freefwd u\big)(\vec{n}\big) = \big(\Amfwd u\big)(\vec{n})\big).
\end{equation*}
\end{proposition}

\begin{proof}
Assume that $u$ satisfies (\ref{star1}). If $\vec{n}$ has $n_{1}=\cdots =n_{c_1}$, then
$$(1-q) \sum_{i=1}^{c_1} \big(\difbwd_i u\big)(\vec{n}) = (1-q)\sum_{i=1}^{c_1} q^{i-1} \big(\difbwd_{c_1} u\big)(\vec{n}) = (1-q^{c_1})\big(\difbwd_{c_1} u\big)(\vec{n}).$$
The same reasoning applies to the other clusters of $\vec{n}$ readily implying that $\big(\Freebwd u\big)(\vec{n})~=~\big(\Abwd u\big)(\vec{n})$.

Likewise assume that $u$ satisfies (\ref{star2}). If $\vec{n}$ has $n_{1}=\cdots =n_{c_1}$, then it follows that
$$(1-q) \sum_{i=1}^{c_1} \big(\diffwd_i u\big)(\vec{n}) = (1-q)\sum_{i=1}^{c_1} q^{c_1-i} \big(\diffwd_{1} u\big)(\vec{n}) = (1-q^{c_1})\big(\diffwd_{1} u\big)(\vec{n}).$$
The same reasoning applies to the other clusters of $\vec{n}$ readily implying (we have also appealed here to Lemma \ref{equalAfwdlem}) that $\big(\Freefwd u\big)(\vec{n}) = \big(\Amfwd u\big)(\vec{n})$.
\end{proof}

\begin{definition}\label{eigdeffree}
A function $\psibwd:\Z^k\to \C$ is called an eigenfunction of the $q$-Boson backward free generator with $(k-1)$ two-body boundary conditions if $\psibwd$ is an eigenfunction for the $q$-Boson backward free generator and it satisfies the $(k-1)$ $q$-Boson backward two-body boundary conditions (\ref{star1}). We likewise define what it means for a function $\psimfwd:\Z^k\to \C$ to be an eigenfunction of the $q$-Boson forward free generator with $(k-1)$ two-body boundary conditions (\ref{star2}).
\end{definition}

The following is a corollary of Proposition \ref{freetrueequiv}.

\begin{corollary}\label{eigcorgen}
Any eigenfunction $\psibwd:\Z^k\to \C$ for the $q$-Boson backward free generator with $(k-1)$ two-body boundary conditions is, when restricted to $\vec{n}\in \Weyl{k}$, an eigenfunction for the $q$-Boson backward generator $\Abwd$ with the same eigenvalue.

Similarly any eigenfunction $\psimfwd:\Z^k\to \C$ for the $q$-Boson forward free evolution equation with $(k-1)$ two-body boundary conditions is, when restricted to $\vec{n}\in \Weyl{k}$, an eigenfunction for the $q$-Boson conjugated forward generator $\Amfwd$ with the same eigenvalue. In turn, $C_q^{-1}\psimfwd$ is an eigenfunction for the $q$-Boson forward generator $\Afwd$ with the same eigenvalue.
\end{corollary}

\subsection{Coordinate Bethe ansatz}

\subsubsection{General review}\label{corbetherev}
We briefly review the coordinate Bethe ansatz \cite{Bethe}. Given an operator $L$ acting on functions from $X\to \C$ (where $X$ is an arbitrary space) we form the $k$-particle operator
$\mathcal{L}$ which acts on functions $f:X^k\to \C$ as
\begin{equation}\label{calA}
(\mathcal{L} f)(\vec{x}) = \sum_{i=1}^{k} (L_i f)(\vec{x})
\end{equation}
where $L_i$ acts as $L$ on the coordinate $x_i$. Write $\psi_z$ to denote any eigenfunction with eigenvalue $z\in \C$ such that $L\psi_z = z\psi_z$ (we are not assuming all $z$ correspond to an eigenfunction or that all eigenvalues are simple). Then it follows that
\begin{equation}\label{Psibethe}
\Psi_{\vec{z}}(\vec{x}) = \sum_{\sigma\in S_k} \mathbf{A}_{\sigma}(\vec{z}) \prod_{j=1}^{k} \psi_{z_{\sigma(j)}}(x_j)
\end{equation}
is an eigenfunction of $\mathcal{L}$ with eigenvalue $(z_1+\cdots +z_k)$. Presently $\mathbf{A}_{\sigma}(\vec{z})\in \C$ can be arbitrary.

Consider an operator $B$ which acts on functions $g:X^2\to \C$. We say that $g$ satisfies the two-body boundary condition corresponding to $B$ if for all $x\in X$, $\big(Bg\big)(x,x)=0$.
Let $B_{i,i+1}$ act as $B$ in the variables $x_i$ and $x_{i+1}$ and as the identity for all other variables. We wish to find $\Psi$ such that
\begin{equation}\label{BCbethe}
\textrm{for all }1\leq i\leq k-1\qquad B_{i,i+1} \Psi\big\vert_{\vec{x}:x_i=x_{i+1}} \equiv 0.
\end{equation}

This can be accomplished by choosing $\mathbf{A}_{\sigma}(\vec{z})$ suitably. Define the function
\begin{equation*}
S(z_1,z_2) = \frac{\big(B (\psi_{z_1}\otimes \psi_{z_2})\big)(x,x)}{(\psi_{z_1}\otimes \psi_{z_2})(x,x)},
\end{equation*}
where $B(f_1\otimes f_2)$ is defined as $B$ applied to the function which takes $(x_1,x_2)\mapsto f_1(x_1)f_2(x_2)$.

\begin{lemma}\label{Abetheconst}
If
\begin{equation*}
\mathbf{A}_{\sigma}(\vec{z}) = \sgn(\sigma) \prod_{1\leq B<A\leq k}\frac{S(z_{\sigma(A)},z_{\sigma(B)})}{S(z_{A},z_{B})}
\end{equation*}
then $\Psi_{\vec{z}}(\vec{x})$ from (\ref{Psibethe}) is an eigenfunction of the operator $\mathcal{L}$ from (\ref{calA}) with eigenvalue $(z_1+\cdots +z_k)$, for which the $(k-1)$ two-body boundary conditions in (\ref{BCbethe}) hold.
\end{lemma}
\begin{proof}
For any $\sigma\in S_k$ consider the pair of terms in (\ref{Psibethe}) corresponding to $\sigma$ and $\tau_i\sigma$ where $\tau_i$ represents the transposition of $i$ and $i+1$. We will show that $B_{i,i+1}$ applied to this pair
$$
\mathbf{A}_{\sigma}(\vec{z}) \prod_{j=1}^{k} \psi_{z_{\sigma(j)}}(x_j)  +  \mathbf{A}_{\tau_i\sigma}(\vec{z}) \prod_{j=1}^{k} \psi_{z_{(\tau_i\sigma)(j)}}(x_j)
$$
is identically zero if $x_i=x_{i+1}$. If this holds for all pairs, this implies the same holds true for $\Psi_{\vec{z}}(\vec{x})$. Call the first term above $T$ and the second $T_{\tau}$. Then it follows from the definition of the function $S$ that
$$
(B_{i,i+1} T)(\vec{x}) = S(z_{\sigma(i)},z_{\sigma(i+1)})T(\vec{x}),\qquad (B_{i,i+1} T_{\tau})(\vec{x}) = S(z_{\sigma(i+1)},z_{\sigma(i)})T_{\tau}(\vec{x}).
$$
For $\vec{x}$ such that $x_i=x_{i+1}$
$$
T_{\tau}(\vec{x}) = - \frac{S(z_{\sigma(i)},z_{\sigma(i+1)})}{S(z_{\sigma(i+1)},z_{\sigma(i)})} T(\vec{x})
$$
and thus
\begin{eqnarray*}
\big(B_{i,i+1} (T+T_{\tau})\big)(\vec{x}) &=& S(z_{\sigma(i)},z_{\sigma(i+1)})T(\vec{x})  +  S(z_{\sigma(i+1)},z_{\sigma(i)})T_{\tau}(\vec{x}) \\
&=&  S(z_{\sigma(i)},z_{\sigma(i+1)})T(\vec{x}) -  S(z_{\sigma(i+1)},z_{\sigma(i)}) \frac{S(z_{\sigma(i)},z_{\sigma(i+1)})}{S(z_{\sigma(i+1)},z_{\sigma(i)})} T =0.
\end{eqnarray*}
\end{proof}

\subsubsection{Application to $q$-Boson particle system generators}\label{applytoq}

We may apply the coordinate Bethe ansatz of Section \ref{corbetherev} to the $q$-Boson backward and forward free generators with $(k-1)$ two-body boundary conditions of Section \ref{freegensec}. One should take care to notice that the eigenfunction associated to $\vec{z}$ has eigenvalue $(q-1)(z_1+\cdots+z_k)$ and not $(z_1+\cdots+z_k)$ as in the above general construction.

\begin{definition}\label{eigdefn}
For all $z_1,\ldots, z_k\in \C\setminus \{1\}$, set
\begin{eqnarray*}
\psibwd_{\vec{z}}(\vec{n})  &=& \sum_{\sigma\in S_k} \prod_{1\leq B<A\leq k} \frac{z_{\sigma(A)}-q z_{\sigma(B)}}{z_{\sigma(A)}- z_{\sigma(B)}} \, \prod_{j=1}^{k} (1-z_{\sigma(j)})^{-n_j}\\
\psimfwd_{\vec{z}}(\vec{n}) &=& \sum_{\sigma\in S_k} \prod_{1\leq B<A\leq k} \frac{z_{\sigma(A)}-q^{-1} z_{\sigma(B)}}{z_{\sigma(A)}- z_{\sigma(B)}} \, \prod_{j=1}^{k} (1-z_{\sigma(j)})^{n_j}\\
\psifwd_{\vec{z}}(\vec{n}) &=& C_q^{-1}(\vec{n}) \psimfwd_{\vec{z}}(\vec{n}).
\end{eqnarray*}
Observe that for any fixed $\vec{n}\in \Weyl{k}$ these are symmetric Laurent polynomials in $(1-z_1),\ldots, (1-z_k)$ thus elements of $\LP{k}$.
\end{definition}

\begin{proposition}\label{prop211}
For all $z_1,\ldots, z_k\in \C\setminus \{1\}$, $\psibwd_{\vec{z}}(\vec{n})$
is an eigenfunction for the $q$-Boson backward free generator with $(k-1)$ two-body boundary conditions (Definition \ref{eigdeffree})  with eigenvalue $(q-1)(z_1+\cdots +z_k)$. The restriction of $\psibwd_{\vec{z}}(\vec{n})$ to $\vec{n}\in \Weyl{k}$ is consequently an eigenfunction for the $q$-Boson backward generator $\Abwd$ with the same eigenvalue.

Similarly, for all $z_1,\ldots, z_k\in \C\setminus \{1\}$, $\psimfwd_{\vec{z}}(\vec{n})$
is an eigenfunction  for the $q$-Boson forward free generator with $(k-1)$ two-body boundary conditions (Definition \ref{eigdeffree}) with eigenvalue $(q-1)(z_1+\cdots +z_k)$. The restriction of $\psimfwd_{\vec{z}}(\vec{n})$ to $\vec{n}\in \Weyl{k}$ is consequently an eigenfunction for the $q$-Boson conjugated forward generator $\Amfwd$ with the same eigenvalue. The restriction of $\psifwd_{\vec{z}}(\vec{n})$ to $\vec{n}\in \Weyl{k}$ is likewise an eigenfunction for the $q$-Boson forward generator $\Afwd$ with the same eigenvalue.
\end{proposition}

\begin{proof}
Let us focus on the backward case first. Observe that $(1-z)^{-n}$ is an eigenfunction for $(1-q)\difbwd$ with eigenvalue $(q-1)z$. We may apply Lemma \ref{Abetheconst} to construct an eigenfunction for the backward free generator which satisfies the two-body boundary conditions of Definition \ref{freedef}. In this application of Lemma \ref{Abetheconst} we find that $S(z_1,z_2) = z_1-q z_2$. Comparing the outcome of the lemma to the expression above for $\psibwd$ we see that they differ by an overall multiplicative factor
$$
\prod_{1\leq B<A\leq k} \frac{z_A-q z_B}{z_A-z_B}.
$$
However, since the $z$-variables are fixed, this factor is just a constant, and hence the desired result follows. The implication for the backward generator $\Abwd$ follows from Corollary \ref{eigcorgen}.

Similar reasoning applies in the forward case. Observe that $(1-z)^{n}$ is an eigenfunction for $(1-q)\diffwd$ with eigenvalue $(q-1)z$. The forward two-body boundary condition can be multiplied by $q^{-1}$ yielding
\begin{equation*}
(\diffwd_i - q^{-1}\diffwd_{i+1})u = 0.
\end{equation*}
We may apply Lemma \ref{Abetheconst} to construct an eigenfunction for the forward free generator which satisfies the above two-body boundary conditions. In this application of Lemma \ref{Abetheconst} we find that $S(z_1,z_2) = z_1-q^{-1} z_2$. Just as in the backward case, this differs from $\psimfwd$ by an overall multiplicative constant (that depends only on the $z$-variables) hence the desired result follows. The implication for the conjugated forward generator $\Amfwd$ and the forward generator $\Afwd$ follows from Corollary \ref{eigcorgen}.
\end{proof}

\begin{remark}\label{qBosonHamiltoniansym}
We may extend the eigenfunctions of Definition \ref{eigdefn} so as to be defined for all of $\Z^k$ (rather than $\Weyl{k}$) by fixing that the value for a general $\vec{n}\in \Z^k$ is the same as the value of $\sigma \vec{n}\in \Weyl{k}$ (where $\sigma\in S_k$ is a permutation of the elements of $\vec{n}$ taking it into $\Weyl{k}$). It is possible to write down an operator on all of $\Z^k$ for which these extensions (which we write with the same notation) are still eigenfunctions. For instance (cf. \cite[Proposition 2.7, C]{BCS}) one readily observes that
$$
(1-q)\left[\sum_{i=1}^{k} \difbwd_i + (1-q^{-1}) \sum_{1\leq i<j\leq k} \bfone_{n_i=n_j} q^{j-i} \difbwd_i\right] \psibwd_{\vec{z}}(\vec{n}) = (1-q) (z_1+\cdots +z_k)\, \psibwd_{\vec{z}}(\vec{n}).
$$
The extension of $\psibwd_{\vec{z}}(\vec{n})$ to $\Z^k$ is through symmetric extension and it is not clear how to modify the formula for the eigenfunction so that the $\vec{n}$ coordinates do not require permuting. For the case of the delta Bose gas (a particular limit of the present system considered later in Section \ref{deltabosesec}) such a formula can be seen in \cite[Equation (A.12)]{Dot}.
\end{remark}

\subsection{Left and right eigenfunctions}\label{leftrightsec}

Proposition \ref{prop211} along with the fact that $\Afwd$ is the transpose of $\Abwd$ implies that
\begin{equation}\label{eqn8}
\big(\Afwd \psifwd_{\vec{z}}\big)(\vec{n}) = \Big((1-q) \sum_{i=1}^{k} z_i\Big) \psifwd_{\vec{z}}(\vec{n}), \qquad \big(\psibwd_{\vec{z}} \Afwd \big)(\vec{n})  = \psibwd_{\vec{z}}(\vec{n}) \Big((1-q) \sum_{i=1}^{k} z_i\Big)
\end{equation}
showing the $\psifwd_{\vec{z}}(\vec{n})$ and $\psibwd_{\vec{z}}(\vec{n})$ are (respectively) right and left eigenfunctions for $\Afwd$ with eigenvalue $(q-1)(z_1+\cdots +z_k)$. This motivates the following.

\begin{definition}\label{leftrighteig}
For any $\vec{z}=(z_1,\ldots, z_k)\in \big(\C\setminus\{1\}\big)^k$ define
\begin{equation*}
\psir_{\vec{z}}(\vec{n})=  \psifwd_{\vec{z}}(\vec{n})=  C^{-1}_q(\vec{n}) \, \sum_{\sigma\in S_k} \prod_{1\leq B<A\leq k} \frac{z_{\sigma(A)}-q^{-1} z_{\sigma(B)}}{z_{\sigma(A)}- z_{\sigma(B)}} \, \prod_{j=1}^{k} (1-z_{\sigma(j)})^{n_j}
\end{equation*}
where $C_q(\vec{n})$ is given in (\ref{Cqdef}). Likewise define
\begin{equation*}
\psil_{\vec{z}}(\vec{n})=  \psibwd_{\vec{z}}(\vec{n})= \sum_{\sigma\in S_k} \prod_{1\leq B<A\leq k} \frac{z_{\sigma(A)}-q z_{\sigma(B)}}{z_{\sigma(A)}- z_{\sigma(B)}} \, \prod_{j=1}^{k} (1-z_{\sigma(j)})^{-n_j}.
\end{equation*}
\end{definition}

We could have just as well defined the right and left eigenfunctions with respect to the operator $\Abwd$ (see Remark \ref{wrtbwdremark} for the implications of this). Finally, observe the symmetry of $\psil$ and $\psir$ with respect to the space-reflection operator
\begin{equation}\label{psisym}
\big(R \psil_{\vec{z}}\big)(\vec{n}) =q^{\frac{k(k-1)}{2}} C_q(\vec{n}) \psir_{\vec{z}}(\vec{n}), \qquad\qquad
\big(R\psir_{\vec{z}}\big)(\vec{n}) = q^{-\frac{k(k-1)}{2}} C_q^{-1}(\vec{n}) \psil_{\vec{z}}(\vec{n}).
\end{equation}

\section{Plancherel formulas}\label{plansec}

In this section we prove a Plancherel formula (Theorem \ref{KqBosonId}) related to the $q$-Boson particle system. We define the transform $\FqBoson$ in which functions are paired with right eigenfunctions. The (candidate) inverse $\JqBoson$ is defined in terms of nested contour integrals, though can also be expressed via a (different) pairing with left eigenfunctions with respect to a (complex) Plancherel measure. The composition of these transforms is written as $\KqBoson:=\JqBoson \FqBoson$. Theorem \ref{KqBosonId} shows that on $\CP{k}$, $\KqBoson$ acts as the identity operator. The proof of this theorem is quite simple and relies on two steps. In the first step we demonstrate a certain symmetry of $\KqBoson$ and in the second step we use elementary residue considerations and this symmetry to prove that $\KqBoson$ acts as the identity.


A key step in our proof has a history and is sometimes called the {\it contour shift argument}. In the related continuum delta Bose gas (cf. Section \ref{deltabosesec}), Heckman and Opdam \cite{HO} used a similar approach to prove a Plancherel formula for that system. They, in turn, attribute it to van den Ban and Schlichtkrull \cite{VS} and ultimately to Helgason, Gangolli and Rozenberg's argument in the proof of the Plancherel theorem for a Riemannian symmetric space $G/K$ \cite{Helbook,GV}. It seems, in fact, that the original instance of this argument goes back to Helgason's 1966 work \cite{Helgason66b}.

The other result we prove in this section is a dual Plancherel formula (Theorem \ref{KqBosonIdDual}) which shows that on $\LP{k}$, $\KqBosonDual:=\FqBoson\JqBoson $ acts as the identity operator. The proof of this dual Plancherel formula goes through a spectral orthogonality of the eigenfunctions $\psir$ and $\psil$  which is given in Proposition \ref{specorth}.

Combining the Placherel formula and the dual Plancherel formula leads to the Plancherel isomorphism given in Theorem \ref{isothm}.

\subsection{The $q$-Boson transform and inverse transform}
We introduce the $q$-Boson transform and the $q$-Boson inverse transform, as well as the contours $\gamma_1,\ldots,\gamma_k$, and $\gamma,\gamma'$. Recall the definitions of the function spaces $\CP{k}$ and $\LP{k}$ from Section \ref{notations}.

\begin{definition}\label{KqTASEPdef}
The $q$-Boson transform $\FqBoson$ takes functions $f\in\CP{k}$ into functions $\FqBoson f\in~\LP{k}$ via
\begin{equation*}
\big(\FqBoson f\big)(\vec{z}) = \llangle f, \psir_{\vec{z}}\rrangle.
\end{equation*}

Fix any set of positively oriented, closed contours $\gamma_1,\ldots, \gamma_k$ chosen so that they all contain $1$, so that the $\gamma_A$ contour contains the image of $q$ times the $\gamma_B$ contour for all $B>A$, and so that $\gamma_k$ is a small enough circle around 1 so as not to contain $q$. Let us also fix contours $\gamma$ and $\gamma'$ where $\gamma$ is a positively oriented closed contour which contains 1 and its own image under multiplication by $q$, and $\gamma'$ contains $\gamma$ and is such that for all $z\in \gamma$ and $w\in \gamma'$, $|1-w|>|1-z|$. See Figure \ref{circontours} for an illustration of a possible set of such contours.

\begin{figure}
\begin{center}
\includegraphics[scale=1.3]{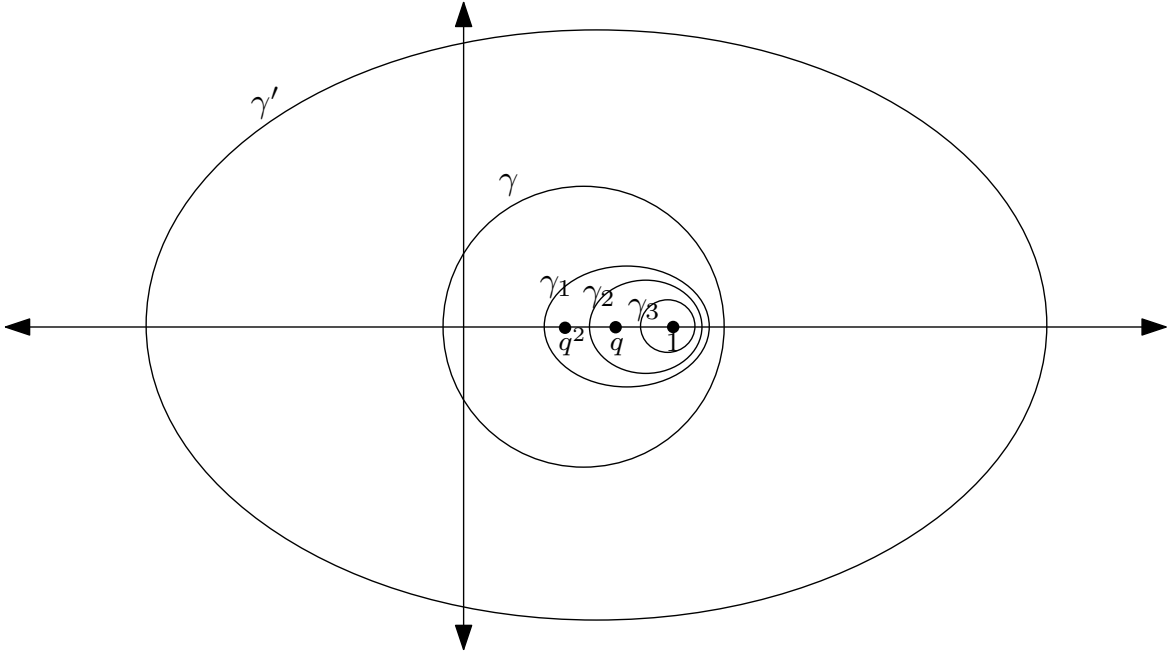}
\end{center}
\caption{For $k=3$ and $q\in (0,1)$, a possible set of contours $\gamma_1,\gamma_2,\gamma_3$ as well as possible contours for $\gamma$ and $\gamma'$ (see Definition \ref{KqTASEPdef}).}\label{circontours}
\end{figure}

The (candidate) $q$-Boson inverse transform $\JqBoson$ takes functions $G\in\LP{k}$ into functions $\JqBoson G\in \CP{k}$ via
\begin{equation}\label{JqBosontrans}
\big(\JqBoson G\big)(\vec{n}) = \oint_{\gamma_1} \frac{dz_1}{2\pi \i} \cdots \oint_{\gamma_k} \frac{dz_k}{2\pi \i}  \prod_{1\leq A<B\leq k} \frac{z_A-z_B}{z_A-qz_B}\, \prod_{j=1}^{k} (1-z_{j})^{-n_j-1}\, G(\vec{z}).
\end{equation}

The composition of the transform and (candidate) inverse transform takes functions $f\in \CP{k}$ into functions $\KqBoson f\in \CP{k}$ via
\begin{eqnarray}\label{KqBosondef}
\big(\KqBoson f \big)(\vec{n}) &=& \big(\JqBoson\FqBoson f\big)(\vec{n}) \\
\nonumber &=&  \oint_{\gamma_1} \frac{dz_1}{2\pi \i} \cdots \oint_{\gamma_k} \frac{dz_k}{2\pi \i}  \prod_{1\leq A<B\leq k} \frac{z_A-z_B}{z_A-qz_B}\, \prod_{j=1}^{k} (1-z_{j})^{-n_j-1} \, \llangle f,\psir_{\vec{z}}\rrangle.
\end{eqnarray}

The composition of the (candidate) inverse transform and the transform takes functions $G\in \LP{k}$ into functions $\KqBosonDual G\in \LP{k}$ via
\begin{eqnarray}\label{KqBosondefDual}
\big(\KqBosonDual G \big)(\vec{n}) &=& \big(\FqBoson \JqBoson F\big)(\vec{z}) \\
\nonumber &=& \sum_{\vec{n}\in \Weyl{k}} \psir_{\vec{z}}(\vec{n}) \oint_{\gamma_1} \frac{dz_1}{2\pi \i} \cdots \oint_{\gamma_k} \frac{dz_k}{2\pi \i}  \prod_{1\leq A<B\leq k} \frac{z_A-z_B}{z_A-qz_B}\, \prod_{j=1}^{k} (1-z_{j})^{-n_j-1} G(\vec{z}).
\end{eqnarray}
\end{definition}

\begin{remark}
It is an easy residue calculation to check that $\FqBoson$ maps $\CP{k}$ into $\LP{k}$ and $\JqBoson$ maps $\LP{k}$ into $\CP{k}$. It is convenient to work with these function spaces since it avoids certain issues such as showing convergence of the summation defining $\FqBoson$. We expect (though do not attempt to prove) that the results which we now develop can be shown to hold true with respect to larger function spaces.
\end{remark}

The operator $\JqBoson$ can be written in terms of $\psil_{\vec{z}}(\vec{n})$ in two ways.
\begin{lemma}\label{321prime}
Consider a symmetric function $G:\C^k\to\C$ and positively oriented, closed contours $\gamma_1,\ldots,\gamma_k$ and $\gamma$ such that:
\begin{itemize}
\item The contour $\gamma$ is a circle around 1 and 0 which contains the image of $\gamma$ multiplied by $q$;
\item For all $1\leq A<B\leq k$, the interior of $\gamma_A$ contains $1$ and the image of $\gamma_B$ multiplied by $q$;
\item For all $1\leq j\leq k$, there exist deformations $D_j$ of $\gamma_j$ to $\gamma$ so that for all $z_1,\ldots, z_{j-1},z_{j+1},\ldots, z_k$ with $z_i\in \gamma$ for $1\leq i<j$, and $z_i\in \gamma_i$ for $j<i\leq k$, the function $z_j\mapsto \V(\vec{z}) G(z_1,\ldots ,z_j,\ldots, z_k)$ is analytic in a neighborhood of the area swept out by the deformation $D_j$.
\end{itemize}
Then,
\begin{equation*}
\big(\JqBoson G\big)(\vec{n}) =  \oint_{\gamma} \cdots \oint_{\gamma} d\mu_{(1)^k}(\vec{w}) \prod_{j=1}^{k} \frac{1}{1-w_j}  \psil_{\vec{w}}(\vec{n}) \, G(\vec{w}),
\end{equation*}
where $(1)^k$ is the partition with $k$ ones and $d\mu_{\lambda}$ is defined in (\ref{dmulambda}).
\end{lemma}
\begin{proof}
Due to the hypothesis on $G$ and the contours, the nested contours in (\ref{JqBosontrans}) can be sequentially (starting with $\gamma_1$ to $\gamma_k$) deformed to the single contour $\gamma$ without changing the value of the integral. Since all contours are now the same, we may symmetrize the integrand. Towards this end we rewrite the product
$$ \prod_{1\leq A<B\leq k} \frac{z_A-z_B}{z_A-qz_B} = \prod_{1\leq A\neq B\leq k} \frac{z_A-z_B}{z_A-qz_B}\, \prod_{1\leq B\neq A\leq k} \frac{z_A-q z_B}{z_A-z_B}$$
where the first product on the right-hand side is symmetric. The eigenfunction $\psil_{\vec{z}}(\vec{n})$ comes from this symmetrization, and the Cauchy determinant identity yields the term $d\mu_{(1)^k}(\vec{z})$, thus completing the proof.
\end{proof}

It is useful (especially for later asymptotic purposes) to also record how the operator $\JqBoson$ can be expressed in terms of contour integrals along the single contour $\gamma_k$. The proof of this result is more involved that that above, and given in the appendix.

\begin{lemma}\label{expandlem}
Consider a symmetric function $G:\C^k\to \C$ and positively oriented, closed contours $\gamma_1,\ldots, \gamma_k$ such that:
\begin{itemize}
\item The contour $\gamma_k$ is a circle around 1 and small enough so as not to contain $q$;
\item For all $1\leq A<B\leq k$, the interior of $\gamma_A$ contains $1$ and the image of $\gamma_B$ multiplied by $q$;
\item For all $1\leq j\leq k$, there exist deformations $D_j$ of $\gamma_j$ to $\gamma_k$ so that for all $z_1,\ldots, z_{j-1},z_{j+1},\ldots, z_k$ with $z_i\in \gamma_i$ for $1\leq i<j$, and $z_i\in \gamma_k$ for $j<i\leq k$, the function $z_j\mapsto \V(\vec{z}) G(z_1,\ldots ,z_j,\ldots, z_k)$ is analytic in a neighborhood of the area swept out by the deformation $D_j$.
\end{itemize}
Then,
\begin{equation*}
\big(\JqBoson G\big)(\vec{n}) = \sum_{\lambda\vdash k}\, \oint_{\gamma_k} \cdots \oint_{\gamma_k} d\mu_{\lambda}(\vec{w}) \prod_{j=1}^{\ell(\lambda)} \frac{1}{(w_j;q)_{\lambda_j}}  \psil_{\vec{w}\circ\lambda}(\vec{n}) \, G(\vec{w}\circ \lambda),
\end{equation*}
where the definition of the notation $\vec{w}\circ \lambda$, and $(w_j;q)_{\lambda_j}$ is given in Section \ref{notations} and where
$d\mu_{\lambda}$ is defined in (\ref{dmulambda}).
\end{lemma}
\begin{proof}
This lemma follows immediately from Proposition \ref{321}.
\end{proof}

It follows from Lemma \ref{321prime} that for $f\in \CP{k}$
\begin{equation}\label{Kexpformsmall}
\big(\KqBoson f\big)(\vec{n}) = \oint_{\gamma} \cdots \oint_{\gamma} d\mu_{(1)^k}(\vec{w}) \prod_{j=1}^{k} \frac{1}{1-w_j}  \psil_{\vec{w}}(\vec{n}) \llangle f,\psir_{\vec{w}}\rrangle,
\end{equation}
and likewise it follows from Lemma \ref{expandlem} that
\begin{equation}\label{Kexpform}
\big(\KqBoson f\big)(\vec{n}) = \sum_{\lambda\vdash k}\, \oint_{\gamma_k} \cdots \oint_{\gamma_k} d\mu_{\lambda}(\vec{w}) \prod_{j=1}^{\ell(\lambda)} \frac{1}{(w_j;q)_{\lambda_j}}  \psil_{\vec{w}\circ\lambda}(\vec{n}) \llangle f,\psir_{\vec{w}\circ \lambda}\rrangle.
\end{equation}

\begin{definition}\label{bilinearprime}
Define the (symmetric) bilinear pairing $\llangle \cdot ,\cdot \rranglesmall$ on functions $F,G\in \LP{k}$ via
\begin{eqnarray*}
\llangle F,G\rranglesmall &=& \oint_{\gamma} \cdots \oint_{\gamma} d\mu_{(1)^k}(\vec{w}) \prod_{j=1}^{k} \frac{1}{1-w_j}  F(\vec{w}) G(\vec{w})\\
 &=& \sum_{\lambda\vdash k}\, \oint_{\gamma_k} \cdots \oint_{\gamma_k} d\mu_{\lambda}(\vec{w}) \prod_{j=1}^{\ell(\lambda)} \frac{1}{(w_j;q)_{\lambda_j}}  F(\vec{w}\circ\lambda) G(\vec{w}\circ \lambda).
\end{eqnarray*}
\end{definition}
The equivalence of the two expressions on the right-hand sides above is due to Lemmas \ref{321prime} and \ref{expandlem}. Furthermore they imply that we can write
\begin{equation*}
\big(\JqBoson G\big)(\vec{n}) = \llangle \psil(\vec{n}),G\rranglesmall,
\end{equation*}
where $\psil(\vec{n})$ is the function which maps $\vec{z}\mapsto \psil_{\vec{z}}(\vec{n})$.

\begin{remark}
The operator $\JqBoson$ applied to $G$ amounts to integrating $\psil_{\vec{z}}(\vec{n}) G(\vec{z})$ against a measure which is supported on a disjoint sum of subspaces (or contours and strings of specializations) indexed by $\lambda\vdash k$. This measure may be called the Plancherel measure, though it is most certainly complex valued. This should be compared to the case of the Heisenberg XXZ quantum spin chain on $\Z$ \cite{BabThom, BabGut, Gut}, or the continuum delta Bose gas \cite{Oxford,HO} in which the underlying Hamiltonian is Hermitian and the Plancherel measure is positive, and where one can hope to prove $L^2$ isometries with respect to this measure.
\end{remark}

\subsection{Main results}

We may now state and prove the main results of this work.

\begin{theorem}\label{KqBosonId}
The $q$-Boson transform $\FqBoson$ induces an isomorphism between the space $\CP{k}$ and its image with inverse given by $\JqBoson$. Equivalently, $\KqBoson$ acts as the identity operator on $\CP{k}$.
\end{theorem}

\begin{proof}
In order to prove this theorem we must show that on functions in $\CP{k}$ the operator $\KqBoson=\Id$. Before proving that we show the following property (which follows from the PT-invariance of the eigenfunctions, cf. Remark \ref{commrem}) of $\KqBoson$: For any functions $f,g\in \CP{k}$,
\begin{equation}\label{claimedkbos}
\llangle \KqBoson f,g\rrangle = \llangle f ,(C_q R)^{-1} \KqBoson (C_q R g)\rrangle.
\end{equation}
In order to prove this claim we utilize the formula (\ref{Kexpformsmall}) for $\big(\KqBoson f\big)(\vec{z})$ (we could just as well use (\ref{Kexpform} instead with immediately modifications to the next set of equations). From this as well as (\ref{bilinflip}) and (\ref{psisym}) it follows that
\begin{eqnarray*}
\llangle \KqBoson f,g\rrangle &=& \oint_{\gamma} \cdots \oint_{\gamma} d\mu_{(1)^k}(\vec{w}) \prod_{j=1}^{k} \frac{1}{1-w_j}  \llangle \psil_{\vec{w}},g\rrangle\, \llangle f,\psir_{\vec{w}}\rrangle\\
&=&   \oint_{\gamma} \cdots \oint_{\gamma} d\mu_{(1)^k}(\vec{w}) \prod_{j=1}^{k} \frac{1}{1-w_j}  \llangle q^{\frac{k(k-1)}{2}}C_q \psir_{\vec{w}},R g\rrangle\, \llangle Rf, q^{-\frac{k(k-1)}{2}} C_q^{-1} \psil_{\vec{w}}\rrangle\\
&=&  \oint_{\gamma} \cdots \oint_{\gamma} d\mu_{(1)^k}(\vec{w}) \prod_{j=1}^{k} \frac{1}{1-w_j} \llangle \psil_{\vec{w}}, C_q^{-1} Rf \rrangle \,  \llangle C_q R g, \psir_{\vec{w}}\rrangle\\
&=& \llangle C_q^{-1} R f , \KqBoson(C_q Rg)\rrangle.
\end{eqnarray*}
Moving the $C_q^{-1}R$ operator to the right-hand side and using the fact that $R=R^{-1}$, yields (\ref{claimedkbos}).

We now show that for any $\vec{x}\in \Weyl{k}$, if $f(\vec{n}) = \bfone_{\vec{n}=\vec{x}}$ then for $\vec{y}\in \Weyl{k}$
\begin{equation*}
\big(\KqBoson f\big)(\vec{y}) = \bfone_{\vec{y}=\vec{x}}.
\end{equation*}
This claim (which by linearity of $\KqBoson$ completes the proof of the theorem) is shown in two parts: First we show that $\big(\KqBoson f\big)(\vec{y})$ is zero unless $\vec{y}=\vec{x}$, and then we compute the value of $\big(\KqBoson f\big)(\vec{x})$.

Let $g(\vec{n}) = \bfone_{\vec{n}=\vec{y}}$. Then showing the first part is equivalent to showing that $\llangle \KqBoson f ,g\rrangle$ may only be non-zero if $\vec{y}=\vec{x}$.
With these choices for $f$ and $g$ we have that
\begin{equation}\label{z1res}
\llangle \KqBoson f ,g\rrangle =  \oint_{\gamma_1} \frac{dz_1}{2\pi \i} \cdots \oint_{\gamma_k} \frac{dz_k}{2\pi \i}  \prod_{1\leq A<B\leq k} \frac{z_A-z_B}{z_A-qz_B}\, \prod_{j=1}^{k} (1-z_{j})^{-y_j-1} \psir_{\vec{z}}(\vec{x}).
\end{equation}
Since we may (without crossing poles) expand the $\gamma_1$ contour to infinity, the integral in $z_1$ can be evaluated by taking the residue at $z_1=\infty$. From this and the definition of $\psir$ we find that this integral may be non-zero only if $y_1\leq x_1$.
On the other hand, by (\ref{claimedkbos}),
\begin{align}\label{zkres}
&\llangle \KqBoson f ,g\rrangle = \llangle f ,(C_q R)^{-1} \KqBoson (C_q R g)\rrangle \\
\nonumber&=\frac{C_q(\vec{y})}{C_q(\vec{x})}\oint_{\gamma_1} \frac{dz_1}{2\pi \i} \cdots \oint_{\gamma_k} \frac{dz_k}{2\pi \i}  \prod_{1\leq A<B\leq k} \frac{z_A-z_B}{z_A-qz_B}\, \prod_{j=1}^{k} (1-z_{j})^{x_{k-j+1}-1} \psir_{\vec{z}}(-y_k,\ldots,-y_1).
\end{align}
Since we may (without crossing poles) shrink the $\gamma_k$ contour to $1$, the integral in $z_k$ can be evaluated by taking the residue at $z_k=1$. From this and the definition of $\psir$ we find that this integral may be non-zero only if $x_1\leq y_1$.

Combining the above two observations, we find that for $\llangle \KqBoson f ,g\rrangle$ to be non-zero, we must have $x_1=y_1$. Now assuming that $x_1=y_1$, we may evaluate the residue at $z_1=\infty$ in (\ref{z1res}) and proceed to expand (without crossing poles) the $\gamma_2$ contour to infinity. From this we find that this integral may be non-zero only if $y_2\leq x_2$. Likewise, we may evaluate the residue at $z_k=1$ in (\ref{zkres}) and proceed to shrink the $\gamma_{k-1}$ contour to 1. From this we find that this integral may be non-zero only if $x_2\leq y_2$. Thus $x_2=y_2$. Proceeding in this manner we arrive at the first part of the claim, that $\big(\KqBoson f\big)(\vec{y})$ is zero unless $\vec{y}=\vec{x}$.

Now we evaluate $\big(\KqBoson f\big)(\vec{x})$. Observe that by sequentially deforming the contours $\gamma_1$ through $\gamma_k$ to infinity, we find that
\begin{align*}
&\big(\KqBoson f\big)(\vec{x}) =C_q^{-1}(\vec{x}) (-1)^k \\
& \times \Res{z_k=\infty}\cdots \Res{z_1=\infty} \prod_{1\leq A<B\leq k} \frac{z_A-z_B}{z_A-qz_B}\, \sum_{\sigma\in S_k} \prod_{1\leq B<A\leq k} \frac{z_{\sigma(A)}-q^{-1}z_{\sigma(B)}}{z_{\sigma(A)}-z_{\sigma(B)}} \prod_{j=1}^{k} (1-z_{j})^{-x_j+x_{\sigma^{-1}(j)}-1}.
\end{align*}
The residue will be non-zero only for those terms in the summation over $S_k$ in which $\sigma$ only permutes within the clusters of $\vec{x}$ (i.e., if $k=5$ and $x_1=x_2=x_3>x_4=x_5$, then $\sigma$ should stabilize the set $(1,2,3)$ and the set $(4,5)$). Let $\vec{c}=(c_1,\ldots, c_M)$ record the cluster sizes of $\vec{x}$ (see Section \ref{notations}), and  define the set $C(\vec{c},i)=\big\{c_1+\cdots +c_{i-1}+1, \ldots, c_{1}+\cdots + c_{i}\big\}$. We may write such a $\sigma\in S_k$ as a product $\sigma_1\cdots \sigma_M$ where each $\sigma_i$ permutes the elements of $C(\vec{c},i)$ and fixes $\{1,\ldots, k\}\setminus C(\vec{c},i)$. We denote the set of all such permutations $\sigma_i \in S(\vec{c},i)$.

Owing to this decomposition (and the fact that for such $\sigma$, $x_j = x_{\sigma^{-1}(j)}$) we may rewrite
\begin{align*}
&\sum_{\sigma\in S_k} \prod_{1\leq B<A\leq k} \frac{z_{\sigma(A)}-q^{-1}z_{\sigma(B)}}{z_{\sigma(A)}-z_{\sigma(B)}} \prod_{j=1}^{k} (1-z_{j})^{-x_j+x_{\sigma^{-1}(j)}-1}\\
&=\prod_{j=1}^{k} \frac{1}{1-z_j} \left( \prod_{i=1}^{M} \sum_{\sigma_i \in S(\vec{c},i)} \prod_{\substack{1\leq B\leq A\leq k\\A,B\in C(\vec{c},i)}}  \frac{z_{\sigma_i(A)}-q^{-1}z_{\sigma_i(B)}}{z_{\sigma_i(A)}-z_{\sigma_i(B)}} \right)
\left( \prod_{1\leq j<i\leq M} \prod_{\substack{A\in C(\vec{c},i)\\B\in C(\vec{c},j)}}  \frac{z_{A}-q^{-1}z_{B}}{z_{A}-z_{B}}\right).
\end{align*}

Note the identity \cite[III.1(1.4)]{M} that for any $m\geq 1$,
\begin{equation}\label{mqinverse}
\sum_{\sigma \in S_m} \prod_{1\leq B\leq A\leq m} \frac{z_{\sigma(A)}-q^{-1} z_{\sigma(B)}}{z_{\sigma(A)}-z_{\sigma(B)}} =m!_{q},
\end{equation}

It follows from the above identity that
\begin{align}\label{kqbosonthmend}
&\big(\KqBoson f\big)(\vec{x}) = C_q^{-1}(\vec{x}) (-1)^k \\
&\nonumber\times \prod_{i=1}^{M} (c_i)!_{q^{-1}}\,  \Res{z_k=\infty}\cdots \Res{z_1=\infty} \prod_{1\leq A<B\leq k} \frac{z_A-z_B}{z_A-qz_B}\, \prod_{j=1}^{k} \frac{1}{1-z_j} \prod_{1\leq j<i\leq M} \prod_{\substack{A\in C(\vec{c},i)\\B\in C(\vec{c},j)}}  \frac{z_{A}-q^{-1}z_{B}}{z_{A}-z_{B}}.
\end{align}
The residue is now easily computed since the expression has simple poles at $\infty$ coming from the $(1-z_j)^{-1}$ terms: It equals
\begin{equation}\label{starstars}
(-1)^k \Lim{z_k=\infty}\cdots \Lim{z_1=\infty} \prod_{1\leq A<B\leq k} \frac{z_A-z_B}{z_A-qz_B}\, \prod_{1\leq j<i\leq M} \prod_{\substack{A\in C(\vec{c},i)\\B\in C(\vec{c},j)}}  \frac{z_{A}-q^{-1}z_{B}}{z_{A}-z_{B}}.
\end{equation}
The limit of the first product is $1$ due to the order in which the limits are taken. For the second multiplicative term, since $j<i$ it implies that $B<A$, each factor of $\frac{z_A-q^{-1}z_B}{z_A-z_B}$ limits to $q^{-1}$. For each $j<i$ there are a total of $c_ic_j$ such factors. Thus (\ref{starstars}) is equal to
$$
(-1)^k \prod_{i=1}^{M} (c_i)!_{q^{-1}} \prod_{1\leq j<i\leq M} q^{-c_i c_j}
$$
Note that
$
(c)!_{q^{-1}} = q^{-\frac{c(c-1)}{2}} (c)!_q
$
and, since $k=\sum_{i=1}^{M} c_i$, that
$$
\prod_{i=1}^{M} q^{-\frac{c_i(c_i-1)}{2}} \prod_{1\leq j<i\leq M} q^{-c_i c_j} = q^{-\frac{k(k-1)}{2}}.
$$
Combining these observations with (\ref{kqbosonthmend}) shows that $\big(\KqBoson f\big)(\vec{x})=1$, and thus completes the proof of the theorem.
\end{proof}

\begin{remark}\label{wrtbwdremark}
The fact that $\KqBoson=\Id$ obviously implies that
\begin{equation}\label{altplan}
(C_q R)^{-1} \KqBoson C_q R = \Id
\end{equation}
as well. This alternative form of the Plancherel formula effectively switches the roles of $\psil$ and $\psir$.
\end{remark}

We turn now to a dual Plancherel formula.

\begin{theorem}\label{KqBosonIdDual}
The inverse $q$-Boson transform $\JqBoson$ induces an isomorphism between the space $\LP{k}$ and its image with inverse given by $\FqBoson$. Equivalently, $\KqBosonDual$ acts as the identity operator on $\LP{k}$.
\end{theorem}

\begin{proof}
Observe that for any $G\in \LP{k}$, $\big(\KqBosonDual G\big)(\vec{z})$ can be written (due to Lemma \ref{321prime}) as
\begin{equation}\label{Gid}
\big(\KqBosonDual G\big)(\vec{z}) =  \sum_{\vec{n}\in \Weyl{k}} \psir_{\vec{z}}(\vec{n}) \oint_{\gamma}\frac{dw_1}{2\pi \i}\cdots \oint_{\gamma}\frac{dw_k}{2\pi \i} d\mu_{(1)^k}(\vec{w}) \prod_{j=1}^{k} \frac{1}{1-w_j} \psil_{\vec{w}}(\vec{n}) G(\vec{w}).
\end{equation}
Thus it suffices to prove that the right-hand side above is equal to $G(\vec{z})$.

Towards this aim, consider any complex valued Laurent polynomial $F(\vec{z})$ in the variables $1-z_1,\ldots, 1-z_k$. In order to finish the proof of the lemma it suffices to show the following integrated version holds for all such $F$:
\begin{align}\label{Gidint}
&\oint_{\gamma}\frac{dz_1}{2\pi \i}\cdots \oint_{\gamma}\frac{dz_k}{2\pi \i}  F(\vec{z})G(\vec{z})  \\
\nonumber &= \oint_{\gamma}\frac{dz_1}{2\pi \i}\cdots \oint_{\gamma}\frac{dz_k}{2\pi \i} F(\vec{z})\sum_{\vec{n}\in \Weyl{k}} \psir_{\vec{z}}(\vec{n}) \oint_{\gamma}\frac{dw_1}{2\pi \i}\cdots \oint_{\gamma}\frac{dw_k}{2\pi \i} d\mu_{(1)^k}(\vec{w}) \prod_{j=1}^{k} \frac{1}{1-w_j} \psil_{\vec{w}}(\vec{n}) G(\vec{w}).
\end{align}

From the Cauchy determinant identity
\begin{equation}\label{23star}
d\mu_{(1)^k}(\vec{w})= (-1)^{\frac{k(k-1)}2}\frac{1}{k!}\, \frac{\V(\vec{w})^2}	{\prod_{A\ne B}(w_A-qw_B)} \,\prod_{j=1}^{k} \frac{dw_j}{2\pi \i}.
\end{equation}
Interchanging the $\vec{z}$ integration and $\vec{n}$ summation the right-hand side of (\ref{Gidint}) can be brought to
\begin{align*}
&(-1)^{\frac{k(k-1)}2}\frac{1}{k!}\sum_{\vec{n}\in \Weyl{k}} \oint_{\gamma}\frac{dz_1}{2\pi \i}\cdots \oint_{\gamma}\frac{dz_k}{2\pi \i}\oint_{\gamma}\frac{dw_1}{2\pi \i}\cdots \oint_{\gamma}\frac{dw_k}{2\pi \i} \frac{\V(\vec{w})\V(\vec{w})}{\prod_{A\ne B}(w_A-qw_B)}\prod_{j=1}^{k} \frac{1}{1-w_j} \psir_{\vec{z}}(\vec{n}) \psil_{\vec{w}}(\vec{n}) F(\vec{z})G(\vec{w})
\\&=\oint_{\gamma}\frac{dw_1}{2\pi \i}\cdots \oint_{\gamma}\frac{dw_k}{2\pi \i}	F(\vec{w}) \frac{1}{k!}\sum_{\sigma\in S_k}G(\sigma \vec{w}) = \oint_{\gamma}\frac{dw_1}{2\pi \i}\cdots \oint_{\gamma}\frac{dw_k}{2\pi \i} F(\vec{w})G(\vec{w}).
\end{align*}
The equality between the first and second lines is an application of Proposition \ref{specorth} (given below and proved independently in Section \ref{semidisc}). This matches the left-hand side of (\ref{Gidint}) and completes the proof of this theorem.
\end{proof}

\begin{remark}\label{conjrem}
It should be possible to extend the dual Plancherel formula to more degenerate classes of functions such as those of the form $G(\vec{z}'\circ\lambda)$ where $\lambda\vdash k$ and $\vec{z}'$ has $\ell(\lambda)$ variables. In this direction we conjecture that the following holds when the contour $\gamma$ and class of functions is chosen suitably:
$$
	G(\vec{z}'\circ\lambda)= \sum_{\vec{n}\in \Weyl{k}} \psir_{\vec{z}'\circ\lambda}(\vec{n})
	\oint_{\gamma}\frac{dw'_1}{2\pi \i}\ldots\oint_{\gamma}\frac{dw'_{\ell(\lambda)}}{2\pi \i} d\mu_{\lambda}(\vec{w}')\prod_{j=1}^{\ell}\frac{1}{(w'_j,q)_{\lambda_j}}
	\psil_{\vec{w}'\circ\lambda}(\vec{n}) G(\vec{w}'\circ\lambda).
$$
This result is related to the conjectured spectral orthogonality given in Remark \ref{orthconj}.
\end{remark}

We may combine Theorems \ref{KqBosonId} and \ref{KqBosonIdDual} we arrive to the following result.

\begin{theorem}\label{isothm}
The $q$-Boson transform $\FqBoson$ induces an isomorphism between $\CP{k}$ and $\LP{k}$ with inverse given by $\JqBoson$. Moreover, for any $f,g\in \CP{k}$
\begin{equation}\label{res1}
\llangle f,g \rrangle = \llangle \FqBoson(\mathcal{P} f),\FqBoson g\rranglesmall,
\end{equation}
and for any $F,G\in \LP{k}$
\begin{equation}\label{res2}
\llangle \mathcal{P}^{-1}(\JqBoson F),\JqBoson G \rrangle = \llangle F,G\rranglesmall.
\end{equation}
Here $\mathcal{P}:\CP{k}\to \CP{k}$ is defined via its action $(\mathcal{P}g)(\vec{n}) = (-1)^k C_q(\vec{n})(R g)(\vec{n})$ and is the operator which maps right eigenfunctions to left eigenfunctions.
\end{theorem}
\begin{proof}
The isomorphism follows immediately by combining Theorems \ref{KqBosonId} and \ref{KqBosonIdDual}. To show (\ref{res1}) and (\ref{res2}) consider an arbitrary $g\in \CP{k}$, and let $\bfone_{\vec{m}}\in \CP{k}$ (for $\vec{m}\in \Weyl{k}$) be the function $\bfone_{\vec{m}}(\vec{n}) = \bfone_{\vec{m}=\vec{n}}$. Also, let $\psil(\vec{n})\in \LP{k}$ be the function $\psil(\vec{n}) (\vec{z}) = \psil_{\vec{z}}(\vec{n})$. Then,
$$
\llangle \psil(\vec{m}) , \big(\FqBoson g\big)\rranglesmall = \big(\JqBoson \FqBoson g\big)(\vec{m}) = g(\vec{m}) = \llangle \bfone_{\vec{m}},g\rrangle.
$$
The first equality follows from Definition \ref{bilinearprime} as well as Lemma \ref{expandlem}. The second equality follows from Theorem \ref{KqBosonId} and the final equality follows immediately from the definition of the bilinear pairing. Any function $f\in \CP{k}$ is a (finite) linear combination of the functions $\bfone_{\vec{m}}$ (over $\vec{m}\in \Weyl{k}$), hence the desired result (\ref{res1}) follows by linearity of the bilinear pairings and the fact that $\mathcal{P}\psir(\vec{m}) =\psil(\vec{m})$.

The result (\ref{res2}) follows due to an application of Theorem \ref{KqBosonIdDual}.
\end{proof}

\medskip
Theorem \ref{KqBosonId} has two immediate corollaries. The first (Corollary \ref{completenesscor}) is the  completeness of the coordinate Bethe ansatz for the $q$-Boson particle system (with respect to a certain complex Plancherel measure), and the second (Corollary \ref{north}) is the orthogonality of $\psir_{\vec{z}}(\vec{n})$ and $\psil_{\vec{z}}(\vec{m})$ when integrated against the Plancherel measure. We call this a {\it spatial} orthogonality since it corresponds to an orthogonality in the variables $\vec{n}$.

There is a second orthogonality which we call {\it spectral} as it is in the variables $\vec{z}$. The spectral orthogonality (given as Proposition \ref{specorth} below) is not a consequence of Theorem \ref{KqBosonId} and, as we have just seen, is the main input in the proof of the dual Plancherel formula of Theorem \ref{KqBosonIdDual}. It follows by taking $\e=1$ in Proposition \ref{trueorththm}, which, in turn, is proved by ultimately appealing to the Cauchy-Littlewood identity for Hall-Littlewood symmetric polynomials.

For what follows, recall the contours $\gamma_1,\ldots,\gamma_k,\gamma,\gamma'$ from Definition \ref{KqTASEPdef}.

\subsection{Completeness}\label{compsec}

\begin{corollary}\label{completenesscor}
Any function $f\in\CP{k}$ can be expanded as
\begin{equation}\label{completenessfwd}
f(\vec{n}) = \sum_{\lambda\vdash k}\, \oint_{\gamma_k} \cdots \oint_{\gamma_k} d\mu_{\lambda}(\vec{w}) \prod_{j=1}^{\ell(\lambda)} \frac{1}{(w_j;q)_{\lambda_j}}  \psil_{\vec{w}\circ\lambda}(\vec{n}) \llangle f,\psir_{\vec{w}\circ \lambda}\rrangle,
\end{equation}
and also as
\begin{equation}\label{completenessbwd}
f(\vec{n}) = \sum_{\lambda\vdash k}\, \oint_{\gamma_k} \cdots \oint_{\gamma_k} d\mu_{\lambda}(\vec{w}) \prod_{j=1}^{\ell(\lambda)} \frac{1}{(w_j;q)_{\lambda_j}}  \psir_{\vec{w}\circ\lambda}(\vec{n}) \llangle \psil_{\vec{w}\circ \lambda},f \rrangle.
\end{equation}
\end{corollary}
\begin{proof}
The expansion in (\ref{completenessfwd}) follows from applying $\KqBoson$ to the function $f$. On the one hand, Theorem \ref{KqBosonId} shows that this returns $f$, while (\ref{Kexpform}) provides the above expansion.

The expression in (\ref{completenessbwd}) comes from the previous expansion as well as the relations (\ref{psisym}). Equivalently, it is related to the expansion of the nested contour integral corresponding to the left-hand side of (\ref{altplan}).
\end{proof}

Due to Lemma \ref{321prime} we also have the expansion
\begin{equation*}
f(\vec{n}) = \oint_{\gamma} \cdots \oint_{\gamma} d\mu_{(1)^k}(\vec{w}) \prod_{j=1}^{k} \frac{1}{1-w_j}  \psil_{\vec{w}}(\vec{n}) \llangle f,\psir_{\vec{w}}\rrangle,
\end{equation*}
and similarly when $r$ and $\ell$ are switched.

\subsection{Biorthogonality}\label{bisec}

The following spatial orthogonality is also an immediate consequence of the Plancherel formula.

\begin{corollary}\label{north}
For $\vec{n},\vec{m}\in \Weyl{k}$, regarding $\psil(\vec{n})$ and $\psil(\vec{m})$ as functions $\big(\psil(\vec{n})\big)(\vec{z})=\psil_{\vec{z}}(\vec{n})$ and $\big(\psir(\vec{m})\big)(\vec{z})=\psir_{\vec{z}}(\vec{m})$ we have that
\begin{equation*}
\llangle \psil(\vec{n}), \psir(\vec{m}) \rranglesmall = \bfone_{\vec{n}=\vec{m}}.
\end{equation*}
\end{corollary}

\begin{proof}
This follows by applying (\ref{completenessfwd}) to the function $f(\vec{n}) = \bfone_{\vec{n}=\vec{m}}$.
\end{proof}

The following spectral orthogonality does not seem to follow from the Plancherel formula (Theorem \ref{KqBosonId}). It is, in fact, the key input in proving the dual Plancherel formula (Theorem \ref{KqBosonIdDual}), so we provide an independent proof of it. The spaces of functions with which we deal are more general than necessary for the dual Plancherel formula.

\begin{proposition}\label{specorth}
Consider a function $F(\vec{z})$ such that for $M$ large enough, $\prod_{i=1}^{k}(1-z_i)^{-M} \V(\vec{z})F(\vec{z})$ is analytic in the closed exterior of $\gamma$, and consider another function $G(\vec{w})$ such that $\V(\vec{w})G(\vec{w})$ is analytic in the closed region between $\gamma$ and $\gamma'$. Then we have that
\begin{align*}
&\sum_{\vec{n}\in \Weyl{k}}\left( \oint_{\gamma}\frac{dz_1}{2\pi \i}\cdots  \oint_{\gamma}\frac{dz_k}{2\pi \i} \psir_{\vec{z}}(\vec{n})\V(\vec{z})F(\vec{z}) \right) \left( \oint_{\gamma}\frac{dw_1}{2\pi \i}\cdots  \oint_{\gamma}\frac{dw_k}{2\pi \i} \psil_{\vec{w}}(\vec{n}) \V(\vec{w}) G(\vec{w})\right)\\
&= \oint_{\gamma}\frac{dw_1}{2\pi \i}\cdots  \oint_{\gamma}\frac{dw_k}{2\pi \i}  (-1)^{\frac{k(k-1)}{2}} \prod_{j=1}^{k} (1 - w_j) \prod_{A\neq B} (w_A-qw_B) \sum_{\sigma\in S_k} \sgn(\sigma) F(\sigma \vec{w}) G(\vec{w}).
\end{align*}
\end{proposition}
\begin{proof}
This is the $\e=1$ case of Proposition \ref{trueorththm}.
\end{proof}

\begin{remark}\label{orthconj}
The above identity may be formally rewritten as
\begin{equation}
\llangle \psil_{\vec{w}},\psir_{\vec{z}}\rrangle \V(\vec{z}) \V(\vec{w}) = (-1)^{\frac{k(k-1)}{2}} \prod_{j=1}^{k} (1 - z_j) \prod_{A\neq B} (z_A-qz_B) \det\big[\delta_{z_i,w_j}\big]_{i,j=1}^{k}
\end{equation}
where $\delta_{z,w}$ is the Dirac delta function for $z=w$. Owing to (\ref{23star}), this can be further (even more formally) rewritten as
$$
\llangle \psil_{\vec{w}},\psir_{\vec{z}}\rrangle   = \det\big[\delta_{z_i,w_j}\big]_{i,j=1}^{k} \frac{\prod_{i=1}^{k} \frac{1}{1-z_i} dz_i}{d\mu_{(1)^k}(\vec{z})}.
$$

It is natural to try to extend the above spectral orthogonality to all eigenfunctions which arise in the completeness results of Corollary \ref{completenesscor}. This prompts the conjecture that for any two partitions $\lambda,\mu\vdash k$, if $\vec{z}= \vec{z}'\circ \lambda$ and $\vec{w} = \vec{w}'\circ \mu$ (for $\vec{z}'$ of length $\ell(\lambda)$ and $\vec{w}'$ of length $\ell(\mu)$) then
\begin{equation}
\llangle \psil_{\vec{w}},\psir_{\vec{z}}\rrangle \V(\vec{z}) \V(\vec{w}) =\bfone_{\lambda,\mu} (-1)^{\frac{k(k-1)}{2}} \prod_{j=1}^{k} (1 - z_j) \prod^{\sim}_{A\neq B} (z_A-qz_B) \det\big[\delta_{z_i,w_j}\big]_{i,j=1}^{k},
\end{equation}
where the product ${\prod\limits_{A\ne B}^{\sim}}(z_A-qz_B)$ means that we omit factors that are zero (by virtue of the definitions of $\vec{z}$ and $\vec{w}$). Presumably, the above formal identity should be understood in a similar manner as Proposition \ref{specorth}. Note that the above may be expressed via $d\mu_{\lambda}(\vec{w})$ owing to the fact that
\begin{equation}\label{24star}
d\mu_\lambda(\vec{w}')=(-1)^{\frac{k(k-1)}2}\frac{1}{m_1!m_2!\ldots}\, \frac{\V(\vec{w}'\circ\lambda)^{2}}{\prod\limits_{A\ne B}^{\sim}\big((\vec{w}'\circ\lambda)_{A}-q(\vec{w}'\circ\lambda)_{B}\big)}\prod_{j=1}^{\ell(\lambda)}\frac{dw'_j}{2\pi \i}.
\end{equation}
%

In particular, for $k=2$, we expect that
\begin{align*}
	\sum_{n_1\ge n_2}
	\psir_{z',qz'}(n_1,n_2)
	\psil_{w',qw'}(n_1,n_2)
	(z'-qz')(w'-qw')
	=
	-
	(1-z')(1-qz')
	(z'-q^{2}z')
	\delta_{z',w'}.
\end{align*}
Note that in the determinant $\det[\delta_{z_i,w_j}]_{i,j=1}^{2}$, only one summand survives, namely, $\delta_{z',w'}\delta_{qz',qw'} =\delta_{z',w'}$. We have checked the above $k=2$ version of the conjecture for test functions of the form $F(z_1,z_2)= \prod_{i=1}^{2} (1-z_i)^{m_i}$ and $G(w_1,w_2)= \prod_{i=1}^{2} (1-w_i)^{r_i}$. For such functions, any integration contours which enclose $1$ suffice (since there is only a pole at 1 to consider).
\end{remark}

\section{Applications of Plancherel formulas}\label{appsec}

In this section we show how the Plancherel formula (in particular the completeness result of Corollary \ref{completenesscor}) provides means to solve the Kolmogorov backward and forward equations for the $q$-Boson particle system. In the case of the backward equation in Section \ref{qtasepsec} we explain how, through a duality with another particle system called $q$-TASEP, this enables us to calculate exact formulas for expectations of certain observables of $q$-TASEP with general initial data (thus extending special cases of initial data studied earlier in \cite{BorCor,BCS}). This should be thought as parallel to the work of Dotsenko \cite{Dot} and Calabrese, Le Doussal and Rosso \cite{CDR} in which they use the eigenfunction decomposition of the delta Bose gas to calculate exact formulas for joint moments of the solution of the stochastic heat equation (cf. Section \ref{deltabosesec}).

\subsection{Solving the backward and forward equations}\label{bwdfwdsec}
The Kolmogorov backward and forward equations can be readily solved when the initial data is expressed as a sum over  eigenfunctions. The Plancherel formula (specifically the completeness result of Corollary \ref{completenesscor}) provides such a decomposition. For the forward generator $\psir_{\vec{z}}(\vec{n})$ is the right eigenfunction, whereas for the backward generator $\psil_{\vec{z}}(\vec{n})$ is the right eigenfunction. The following solutions of the backward and forward equations for general (compactly supported) initial data are consequences of Corollary \ref{completenesscor}.

\begin{corollary}\label{bwdapp}
For any $f_0\in \CP{k}$ the backward equation
\begin{equation*}
\frac{d}{dt} f(t;\vec{n}) = \big(\Abwd f\big)(t;\vec{n})
\end{equation*}
with $f(0;\vec{n}) = f_0(\vec{n})$ is uniquely solved by
\begin{eqnarray}\label{ftvecn}
\nonumber f(t;\vec{n}) &=& \big(e^{t\Abwd} f_0\big)(\vec{n}) = \sum_{\lambda\vdash k}\, \oint_{\gamma_k} \cdots \oint_{\gamma_k} d\mu_{\lambda}(\vec{w}) \prod_{j=1}^{\ell(\lambda)} \frac{1}{(w_j;q)_{\lambda_j}} \,e^{tE(\vec{w}\circ \lambda)} \psil_{\vec{w}\circ\lambda}(\vec{n}) \llangle f_0,\psir_{\vec{w}\circ \lambda}\rrangle \\
 &=& \oint_{\gamma_1} \frac{dz_1}{2\pi \i} \cdots \oint_{\gamma_k} \frac{dz_k}{2\pi \i}  \prod_{1\leq A<B\leq k} \frac{z_A-z_B}{z_A-qz_B}\, \prod_{j=1}^{k} (1-z_{j})^{-n_j-1}\, e^{tE(\vec{z})} \,\llangle f_0,\psir_{\vec{z}}\rrangle
\end{eqnarray}
where $E(\vec{z}) = (q-1) (z_1+\cdots +z_k)$ and where the contours $\gamma_1,\ldots, \gamma_k$ are as in Definition \ref{KqTASEPdef}.
\end{corollary}

\begin{proof}
The uniqueness of solutions to the backward equation with initial data in $\CP{k}$ is evident because of the triangular nature of the generator $\Abwd$. From Definition \ref{leftrighteig} and Proposition \ref{prop211}, $\psil_{\vec{w}\circ \lambda}$ is a right eigenfunction for $\Abwd$ with eigenvalue $E(\vec{w}\circ \lambda)$. This implies that the right-hand side of the first line of (\ref{ftvecn}) solves the desired backward equation. The fact that it satisfies the initial data follows from Corollary \ref{completenesscor}, equation (\ref{completenessfwd}). The equivalence of the first and second lines of (\ref{ftvecn}) is a consequence of Proposition \ref{321} (see also Lemma \ref{expandlem}).
\end{proof}

\begin{corollary}\label{fwdapp}
For any $f_0\in \CP{k}$ the forward equation
\begin{equation*}
\frac{d}{dt} f(t;\vec{n}) = \big(\Afwd f\big)(t;\vec{n})
\end{equation*}
with $f(0;\vec{n}) = f_0(\vec{n})$ is uniquely solved by
\begin{eqnarray}\label{ftvecnfwd}
f(t;\vec{n}) &=& \big(e^{t\Afwd} f_0\big)(\vec{n}) = \sum_{\lambda\vdash k}\, \oint_{\gamma_k} \cdots \oint_{\gamma_k} d\mu_{\lambda}(\vec{w}) \prod_{j=1}^{\ell(\lambda)} \frac{1}{(w_j;q)_{\lambda_j}} e^{tE(\vec{w}\circ \lambda)}  \psir_{\vec{w}\circ\lambda}(\vec{n}) \llangle \psil_{\vec{w}\circ \lambda} , f_0 \rrangle \\
\nonumber &=& q^{-\frac{k(k-1)}{2}} C_q^{-1}(\vec{n})\, \oint_{\gamma_1} \frac{dz_1}{2\pi \i} \cdots \oint_{\gamma_k} \frac{dz_k}{2\pi \i}  \prod_{1\leq A<B\leq k} \frac{z_A-z_B}{z_A-qz_B}\, \prod_{j=1}^{k} (1-z_{j})^{n_{k-j+1}-1} e^{tE(\vec{z})} \,\llangle \psil_{\vec{z}}, f_0\rrangle
\end{eqnarray}
where $E(\vec{z}) = (q-1) (z_1+\cdots +z_k)$ and the contours $\gamma_1,\ldots, \gamma_k$ are as in Definition \ref{KqTASEPdef}.
\end{corollary}

\begin{proof}
This is proved in the same manner as  Corollary \ref{bwdapp} with the replacements of $\Abwd$ with $\Afwd$ and $\psil_{\vec{w}\circ \lambda}$ with $\psir_{\vec{w}\circ \lambda}$. The initial data is satisfied by virtue of Corollary \ref{completenesscor}, equation (\ref{completenessbwd}).
\end{proof}

\subsection{Transition probabilities}\label{transprobsec}
The Kolmogorov forward equation governs the evolution of the transition probabilities of a Markov process. In particular, for the $q$-Boson particle system $\vec{n}(t)$ (recall from Section \ref{stochrepsec}) define
$$P_{\vec{y}}(t;\vec{x}) = \PP\big(\vec{n}(t) = \vec{x} \big\vert \vec{n}(0)=\vec{y}\big),$$
for $\vec{x},\vec{y}\in \Weyl{k}$ and $t\geq 0$. Then $P_{\vec{y}}(t;\vec{x})$ solves the Kolmogorov forward equation
$$
\frac{d}{dt} P_{\vec{y}}(t;\vec{x}) = \big(\Afwd P_{\vec{y}}\big)(t;\vec{x}), \qquad P_{\vec{y}}(0;\vec{x}) = \bfone_{\vec{x}=\vec{y}},
$$
where $\Afwd$ (Definition \ref{bwdgendef}) acts in the $\vec{x}$ variables.

We may directly utilize Corollary \ref{fwdapp} to show:
\begin{corollary}\label{transprobcor}
The $q$-Boson particle system transition probability $P_{\vec{y}}(t;\vec{x})$ from initial state $\vec{y}\in \Weyl{k}$ to terminal state $\vec{x}\in \Weyl{k}$ in time $t\geq 0$ is given by both the right-hand side of the first and second line of (\ref{ftvecnfwd}) with $f_0(\vec{n}) = \bfone_{\vec{n}=\vec{y}}$.
\end{corollary}

This result may be compared to Tracy-Widom's ASEP transition probability formulas \cite{TW1}. Those formulas involve $k!$, $k$-fold contour integrals. With some additional work, the above result (which is a single $k$-fold nested contour integral) can be brought closer to this form. In fact, in work of Korhonen-Lee \cite{Hyun} posted soon after the present work was first post, the authors utilize the Tracy-Widom approach to arrive directly at such a formulas as indicated above. In a companion paper \cite{BCPS2} we develop parallel results for ASEP and the Heisenberg XXZ quantum spin chain and remark on the relationship to Tracy-Widom's work on ASEP.

\subsection{Nested contour integral formulas for $q$-TASEP}\label{qtasepsec}
The Kolmogorov backward equation governs the evolution of the expectation of functions of a Markov process. In particular, for the $q$-Boson particle system $\vec{n}(t)$ and a function $f_0:\Weyl{k}\to \R$, we define
$$h(t;\vec{n})= \EE\big[f_0(\vec{n}(t))\big\vert \vec{n}(0)=\vec{n}\big],$$
which always stays well-defined (at least for bounded functions $f_0$).
Then $h(t;\vec{n})$ solves the Kolmogorov backward equation
$$
\frac{d}{dt} h(t;\vec{n}) = \Abwd h(t;\vec{n}), \qquad h(0;\vec{n}) = f_0(\vec{n}),
$$
where $\Abwd$ (Definition \ref{bwdgendef}) acts in the $\vec{n}$ variables.

We may directly utilize Corollary \ref{bwdapp} to show:
\begin{corollary}\label{exactthm}
For $f_0\in \CP{k}$, the expectation $h(t;\vec{n})$ is given by the right-hand side of the first and the second line of (\ref{ftvecn}).
\end{corollary}

In \cite[Section 2]{BCS} it is shown that there exists a Markov duality between the $q$-Boson particle system and $q$-TASEP. As a result, expectations of certain observables of $q$-TASEP can be computed by solving the above backward equation.

The $q$-deformed totally asymmetric simple exclusion process ($q$-TASEP) is a continuous time, discrete space interacting particle system $\vec{x}(t)$. It was first introduced and studied by Borodin-Corwin \cite{BorCor} (see also \cite{BCS,BCdiscrete,OConPei,BorPet,CorPet} for subsequent developments). Particles occupy sites of $\Z$ with at most one particle at any site at a given time. The location of particle $i$ at time $t$ is written as $x_i(t)$. Particles are ordered so that $x_i(t)>x_j(t)$ for $i<j$. The rate at which the value of $x_i(t)$ increases by one (i.e. the particle jumps right by one) is $1-q^{x_{i-1}(t)-x_i(t)-1}$ (here $x_{i-1}(t)-x_{i}(t)-1$ is the size of the gap between particle $x_i$ and $x_{i-1}$ at time $t$). All jumps occur independently of each other according to exponential clocks.

We will focus on $q$-TASEP with only $N\geq 1$ particles, though since particles only depend on those to their right (i.e., smaller index) this restriction can be weakened to include systems with a right-most particle (though possible infinite particles in total). Define  $\Weyl{k,N}$ as the set of all $\vec{n}\in \Weyl{k}$ such that $N\geq n_1\geq \cdots \geq n_k\geq 1$. For some possibly random initial data $\vec{x}(0)$ for $N$ particle $q$-TASEP define $h_0:\Weyl{k}\to \R$ as $0$ outside of the compact set $\Weyl{k,N}$ and otherwise
\begin{equation}\label{h0eqn}
h_0(\vec{n}) := \EE\left[\prod_{i=1}^{k} q^{x_{n_i}(0)+n_i}\right],
\end{equation}
where the expectation is over the possibly random initial data $x(0)$. Then it follows from Proposition 2.7 (and Definition 2.6) of \cite{BCS} that,
\begin{equation}\label{heqn}
h(t;\vec{n}):=\EE\Big[ \prod_{i=1}^{k} q^{x_{n_i}(t)+n_i}\Big]
\end{equation}
solves the $q$-Boson particle system backward equation with initial data $h_0$. Here the expectation is over both the possibly random initial data and the random evolution of $q$-TASEP. In particular, this, along with Corollary \ref{bwdapp} implies the following.
\begin{corollary}\label{abovecor}
For $q$-TASEP with $N$ particles and initial data $\vec{x}(0)$, $h(t;\vec{n})$ is given by the right-hand side of the first and the second lines of (\ref{ftvecn}) with $f_0(\vec{n})=h_0(\vec{n})$.
\end{corollary}

In order to apply Corollary \ref{abovecor} it is necessary to compute $\big(\FqBoson h_0\big)(\vec{z})= \llangle h_0,\psir_{\vec{z}}\rrangle$. For purposes of asymptotics it is desirable to concisely sum the series given by this bilinear pairing. The dual Plancherel formula given in Theorem \ref{KqBosonIdDual} shows how this can be achieved for a certain class of functions.

\begin{corollary}\label{abovecor2}
Consider any function $G\in\LP{k}$, if $f(\vec{n}) = \big(\JqBoson G\big)(\vec{n})$, then $\big(\FqBoson f\big)(\vec{z}) = G(\vec{z})$.
\end{corollary}
\begin{proof}
The follows immediately from the application of Theorem \ref{KqBosonIdDual}.
\end{proof}

To illustrate, let us apply this corollary to the two types of $q$-TASEP initial data studied previously in \cite{BorCor} and \cite{BCS}. Step initial data for $q$-TASEP is when $\vec{x}_i(0) = -i$ for $1\leq i\leq N$. Half stationary initial data with parameter $\alpha\in [0,1)$ is when $x_1(0)=-1-X_1$ and, for $i>1$, $x_i=x_{i-1}-1-X_i$. Here the $X_i$'s are independent random variables with common distribution
$$
\PP\big(X=k\big) = (\alpha;q)_{\infty} \frac{\alpha^k}{(q;q)_{k}}.
$$
See \cite[Section 3.3]{BorCor} for a justification of why this is half-stationary initial data. When $\alpha=0$, half-stationary reduces to step initial data.

For step initial data it is immediate to see that the corresponding initial data $h_0$ is given by
$\prod_{i=1}^{k} \bfone_{0<n_i\leq N}$. Due to the fact that the operator $\Abwd$ is triangular, the solution $h(t;\vec{n})$ for $\vec{n}\in \Weyl{k,N}$ does not depend on the initial data $h_0(\vec{n})$ for $\vec{n}$ such that $n_k>N$. Therefore in order to solve for $h(t;\vec{n})$ we may just as well work with initial data
$$h_{0}^{{\rm step}}(\vec{n}) := \prod_{i=1}^{k} \bfone_{n_i>0},$$
although observe that this function is no longer compactly supported.

For half-stationary initial data, finding $h_0$ requires a calculation involving the above common distribution of the $X_i$. This is done in equation (12) of \cite{BCS} and gives (after similar considerations involving the indicator functions) the (non-compactly supported) initial data
\begin{equation}\label{idh0}
h_{0}^{{\rm half}}(\vec{n}) :=\prod_{i=1}^{k}\bfone_{n_i>0} \prod_{j=1}^{i} \left(\frac{1}{1-\alpha/q^j}\right)^{n_{i}-n_{i+1}} = \prod_{j=1}^{k} \bfone_{n_i>0} \left(1-\frac{\alpha}{q^j}\right)^{-n_j},
\end{equation}
with the convention above that $n_{k+1} = 0$.

A particular instance of Corollary \ref{abovecor2} (in fact, an immediate extension in the function space to allow for non-compact support) implies the following.
\begin{lemma}\label{belowlemma}
Consider any $0\leq \alpha< q^k$. Assume that there exist positively oriented closed contours $\gamma_1,\ldots,\gamma_k$ and $\tilde\gamma_1,\ldots,\tilde\gamma_k$ which both satisfy the conditions of Definition \ref{KqTASEPdef} but also do not include $\alpha/q$ and are such that for all $z\in \gamma_i$ and $w\in \tilde\gamma_i$, $|1-z|<|1-w|$. Then, for $\vec{z}$ with $z_i\in \gamma_i$ for $1\leq i\leq k$,
\begin{equation}\label{halflem}
\big(\FqBoson h_{0}^{{\rm half}}\big)(\vec{z}) = (-1)^k  q^{\frac{k(k-1)}{2}} \prod_{j=1}^{k} \frac{1-z_j}{z_j-\alpha/q}.
\end{equation}
\end{lemma}
\begin{proof}
Define $G(\vec{z}) = \prod_{i=1}^{k} (z_i-\alpha/q)^{-1}$. We claim that $h_{0}^{{\rm half}}(\vec{n}) = \big(\JqBoson G\big)(\vec{n})$, where the definition of $\JqBoson$ is extended to this function $G$ (which is not in $\LP{k}$) by assuming that the contours $\gamma_1,\ldots,\gamma_k$ have the additional condition that they do not include $\alpha/q$. Checking this is easily achieved through residue considerations. Indeed, when $n_k\leq 0$, the $z_k$ integral in $\big(\JqBoson G\big)(\vec{n})$ has no residue at $1$ and hence evaluates to 0, enforcing the condition $n_i>0$, for $1\leq i\leq k$. Otherwise, we find the desired equality by sequentially deforming the $\gamma_1$ through $\gamma_k$ contours to infinity and evaluating the integral $\big(\JqBoson G\big)(\vec{n})$ via its residue at infinity.

If Corollary  \ref{abovecor2} were valid for $G$ as above, then the proof of the lemma would be complete. Though we have not proved that this is the case, we may still complete the proof via an approximation argument. It is possible to find a sequence of functions $G_{m}\in \LP{k}$ so that as $m\to \infty$, $G_{m}(\vec{z})$ converges uniformly over $\vec{z}$ with $z_i\in \gamma_i$ for $1\leq i\leq k$, to $G(\vec{z})$.

In order to extend the result of Corollary  \ref{abovecor2}, we must show that $\big(\FqBoson \JqBoson G\big)(\vec{z}) = G(\vec{z})$ for all $\vec{z}$ with $z_i\in \gamma_i$ for $1\leq i\leq k$. Theorem \ref{KqBosonIdDual} implies that this holds for all $G_{m}$. Consider the $\vec{n}$ term arising in the summation when $\FqBoson$ is applied to $\big(\JqBoson G_m\big)(\vec{n})$. Call $\vec{w}$ the integration variables involved in the formula for $\big(\JqBoson G_m\big)(\vec{n})$. We can choose the integration contours for $\JqBoson$ to be  $\tilde\gamma_1,\ldots, \tilde\gamma_k$. In that case, the $\vec{n}$ term we are considering can be uniformly (in $m$) bounded by a constant times $\delta^{n_1+\cdots+n_k}$ for some $\delta<1$ (which is the maximal ratio of $|1-z_i|/|1-w_i|$ over $z_i\in \gamma_i$, $w_i\in\tilde\gamma_i$ and $1\leq i\leq k$). Since $G$ does not have a pole at $1$, $\big(\JqBoson G_m\big)(\vec{n})$ is supported on the non-negative $\vec{n}$. The above bound along with uniformity in the convergence of $G_m$ to $G$ implies that we can take the $m\to \infty$ limit inside the dual Plancherel formula relation and conclude that $\big(\FqBoson \JqBoson G\big)(\vec{z}) = G(\vec{z})$ as desired.
\end{proof}

The method of the above proof shows one way to try to extend the Plancherel and dual Plancherel formulas to apply to larger classes of functions.

\begin{remark}
Corollary \ref{abovecor2} (in fact Theorem \ref{KqBosonIdDual}) immediately yields many combinatorial formulas like the one in Lemma \ref{belowlemma} above. We provide, in addition to the above proof, a direct inductive proof of Lemma \ref{belowlemma} in Section \ref{inductiveproof} 
The particular identity of Lemma \ref{belowlemma} is akin to one which arises as \cite[Lemma 5]{ImSa} in considering the half-Brownian KPZ equation, as well as bares similarity to the identity in \cite[Section III]{TWBern}.

In applying Corollary \ref{abovecor2}, it is not a priori so clear if given $f$, how to find $G$ such that $f(\vec{n}) = \big(\JqBoson G\big)(\vec{n})$. On strategy is to try to explicitly compute $\big(\FqBoson f\big)(\vec{z})$ for $k=1,2$ and then guess $G$ for general $k$. In order to check such a guess, one need only compute the integrals in applying $\JqBoson$ to $G$ (a opposed to implementing an inductive proof).

There is a class of initial data for $q$-TASEP which arises as marginal distributions of Macdonald processes for which $G$ is clear (the notation and concepts below are explained in detail in \cite{BorCor}, see also \cite{BCGS,BCdiscrete}). Consider $q$-TASEP initialized according to $x_{i}(0) = \lambda^{(i)}_{i} -i$ where $\lambda$ is a Gelfand-Tselin pattern distributed according to an ascending (Macdonald parameter $t=0$ and $q$ the same in the $q$-TASEP) Macdonald process with parameters $a_i\equiv 1$, and $\rho$ a Macdonald non-negative specialization characterized by (non-negative) parameters $\{\alpha_i\}$, $\{\beta_i\}$ and $\gamma$. Then
$$
h_0(\vec{n}) = \frac{D_{n_k}\cdots D_{n_1} \Pi(a_1,\ldots, a_N;\rho)}{\Pi(a_1,\ldots, a_N;\rho)} \Big\vert_{a_i\equiv 1}
$$
where the $D_j$ are Macdonald first difference operators and $\Pi$ is the normalizing constant for the Macdonald process.
These difference operators are naturally encoded via nested contour integrals, from which one readily identities
$h_0(\vec{n}) = \big(\JqBoson G\big) (\vec{n})$ where
$$
G(\vec{z}) = (-1)^k q^{\frac{k(k-1)}{2}} \prod_{i=1}^{k} \frac{\Pi(qz_i;\rho)}{\Pi(z_i;\rho)} \frac{1-z_i}{z_i}.
$$
This implies that $\big( \FqBoson h_0\big)(\vec{z}) = G(\vec{z})$.

Step initial data for $q$-TASEP is a special case of a marginal of a Macdonald process (in which all $\alpha_i=\beta_i=\gamma=0$). Half stationary should arise from a ``two-sided'' generalization of Macdonald processes  (cf. \cite{twosides} for the Schur version) which has yet to be fully developed. Flat or half-flat initial data for $q$-TASEP do not seem likely to arise from Macdonald processes. However, as they have been (non-rigorously) treated at the KPZ equation level, one might hope to explicitly sum the associated series $\big(\FqBoson h_0\big)(\vec{z})$ and proceed to asymptotics for corresponding $q$-TASEP initial data.
\end{remark}

\begin{proposition}\label{dualiyresult}
For step initial data,
\begin{equation}\label{FdefqTASEP}
\EE\Big[ \prod_{i=1}^{k} q^{x_{n_i}(t)+n_i}\Big]  = (-1)^k q^{\frac{k(k-1)}{2}} \oint_{\gamma_1} \frac{dz_1}{2\pi \i}\cdots \oint_{\gamma_k} \frac{dz_k}{2\pi \i} \prod_{1\leq A<B\leq k} \frac{z_A-z_B}{z_A-qz_B} \prod_{j=1}^{k} (1-z_j)^{-n_j} \frac{e^{(q-1)tz_j}}{z_j},
\end{equation}
where $\gamma_1,\ldots, \gamma_k$ are as in  Definition \ref{KqTASEPdef}, with the additional condition that they do not include $0$.
\vskip.1in
For half-stationary initial data with parameter $\alpha\geq 0$ such that $\alpha <q^{k}$,
\begin{equation}\label{FdefqTASEPhalf}
\EE\Big[ \prod_{i=1}^{k} q^{x_{n_i}(t)+n_i}\Big]  = (-1)^k q^{\frac{k(k-1)}{2}} \oint_{\gamma_1} \frac{dz_1}{2\pi \i}\cdots \oint_{\gamma_k}\frac{dz_k}{2\pi \i}\prod_{1\leq A<B\leq k}  \frac{z_A-z_B}{z_A-qz_B} \prod_{j=1}^{k}  (1-z_j)^{-n_j}  \frac{e^{(q-1)tz_j}}{z_j-\alpha/q},
\end{equation}
where $\gamma_1,\ldots, \gamma_k$ are as in Definition \ref{KqTASEPdef} with the additional condition that they do not include $\alpha/q$ (this is possible due to the restriction on the value of $\alpha$).
\end{proposition}
\begin{proof}
It suffices to prove the half-stationary result. Assume that $\gamma_1,\ldots,\gamma_k$ are such that there exist $\tilde\gamma_1,\ldots,\tilde\gamma_k$ so that the hypotheses of Lemma \ref{belowlemma} are satisfied. Then, applying Lemma \ref{belowlemma} and Corollary \ref{abovecor} we arrive at the desired result. Observe that the contours in (\ref{FdefqTASEPhalf}) can be freely deformed now to any choice of $\gamma_1,\ldots,\gamma_k$ as in the statement of the proposition.
\end{proof}

The step initial data version of this result with all $n_i\equiv n$ was proved in \cite[Section 3.3]{BorCor} using the theory of Macdonald processes. General $\vec{n}$ step and half-stationary result was proved in \cite[Theorem 2.11]{BCS}. The Macdonald process approach was then extended in \cite{BCGS} to also cover general $\vec{n}$.

\section{Complex $q$}\label{complexq}

Though we have assumed $q\in (0,1)$ in this paper, many of the results can either be immediately extended (as well as their proofs) or easily modified to accommodate general $q\in \C^{\prime}_k$ where we define
$$\C^{\prime}_k = \big\{z\in \C\setminus \{0\}: \forall 1\leq j<k, z^j\neq 1\big\}.$$
From the perspective of stochastic processes (like the $q$-Boson particle system or $q$-TASEP), having $q$ outside $(-1,1)$ does not seem to be meaningful. However, from the perspective of quantum mechanics such an extension is natural. In fact, a version of the $q$-Boson Hamiltonian (corresponding with the $\e=0$ limit discussed in Section \ref{vandiejen}) with $q$ on the complex unit circle has been used as an integrable regularization of the real time delta Bose gas in the study of quantum quenches (cf. \cite{quench} and references therein). Taking $q$ on the complex unit circle is akin to taking the ratio $\tau=p/q$ in ASEP to be on the complex unit circle or  taking the coupling constant $\Delta$ in the Heisenberg XXZ quantum spin chain (on $\Z$) to be real with $|\Delta|<1$. The Plancherel formulas and completeness / biorthogonality for ASEP and XXZ will be the subject of future work \cite{BCPS2} (see also \cite{BabThom, BabGut, Gut}).


In Section \ref{qBosonsyst}, though the probabilistic interpretation for the backward and forward generators is not valid when $q$ is general, the equivalence of the generators to free generators with $(k-1)$ boundary conditions, and the Bethe ansatz eigenfunctions all stay valid. In short, all results, definitions and formulas hold for general $q\in \C^{\prime}_k$.

In Section \ref{plansec}, the definitions of $\FqBoson$, $\JqBoson$ and $\KqBoson$ remain valid for general $q$. Recall the conditions on the contours $\gamma_1,\ldots, \gamma_k$  that they be positively oriented, closed contours chosen so that they all contain $1$, so that the $\gamma_A$ contour contains the image of $q$ times the $\gamma_B$ contour for all $B>A$, and so that $\gamma_k$ is a small enough circle around 1 that does not contain $q$. For general $q\in \C^{\prime}_k$, these conditions lead to a rather different picture for the contours than that of Figure \ref{circontours}. For instance, for $q$ on the complex unit circle (but not a $j^{th}$ root of unity for $1\leq j<k$), Figure \ref{unitcircontours} shows contours which satisfy the conditions.
\begin{figure}
\begin{center}
\includegraphics[scale=1]{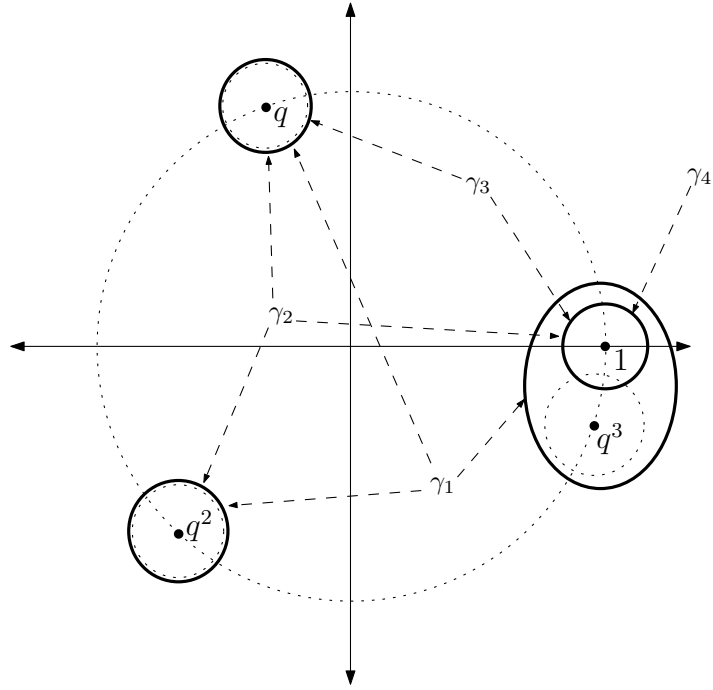}
\end{center}
\caption{Possible $\gamma_1,\ldots,\gamma_k$ contours when $k=4$ and $q$ is on the unit complex circle (approximately $q\approx e^{\frac{5}{8}\pi \i}$). The contour $\gamma_4$ is a small circle around $1$. The contour $\gamma_3$ is the union of the small circle around $1$ and the slightly larger circle around $q$ (so that it contains the image of $\gamma_4$ when multiplied by $q$). The contour $\gamma_2$ is the union of the small circle around $1$, the slightly larger circle around $q$ and the yet larger circle around $q^2$ (so that it contains the image of $\gamma_3$ and $\gamma_4$ when multiplied by $q$). Because the image of $\gamma_2$ when multiplied by $q$ intersects the original small circle around $1$, when defining $\gamma_1$ we take it to be the union of the larger ellipse around $1$ and $q^3$, as well as the circles around $q$ and $q^2$.}\label{unitcircontours}
\end{figure}

For what follows, fix $\gamma_k$ as a circle around 1, small enough so as not to contain $q$. For general $q\in \C^{\prime}_k$, it is possible that the image of $\gamma_k$ when multiplied by a power of $q$ will intersect itself. This will affect the way that a nested contour integral (such as in $\JqBoson$) expands when contours are deformed to $\gamma_k$. Indeed, certain residues previously encountered may not contribute.

With this in mind, define the contour (see Figure \ref{overlap})
$$
\Gamma_{(n)}(q) = \Big\{z\in \gamma_k: \forall 1\leq j< n, q^j z \textrm{ is outside }\gamma_k\Big\}.
$$
These contours are similar in nature (and origin, see \cite{BCPS2}) to the so-called Chebyshev circles considered in \cite{BabGut,Gut} in the context of the XXZ model's Plancherel formula for $|\Delta|<1$.

\begin{figure}
\begin{center}
\includegraphics[scale=1]{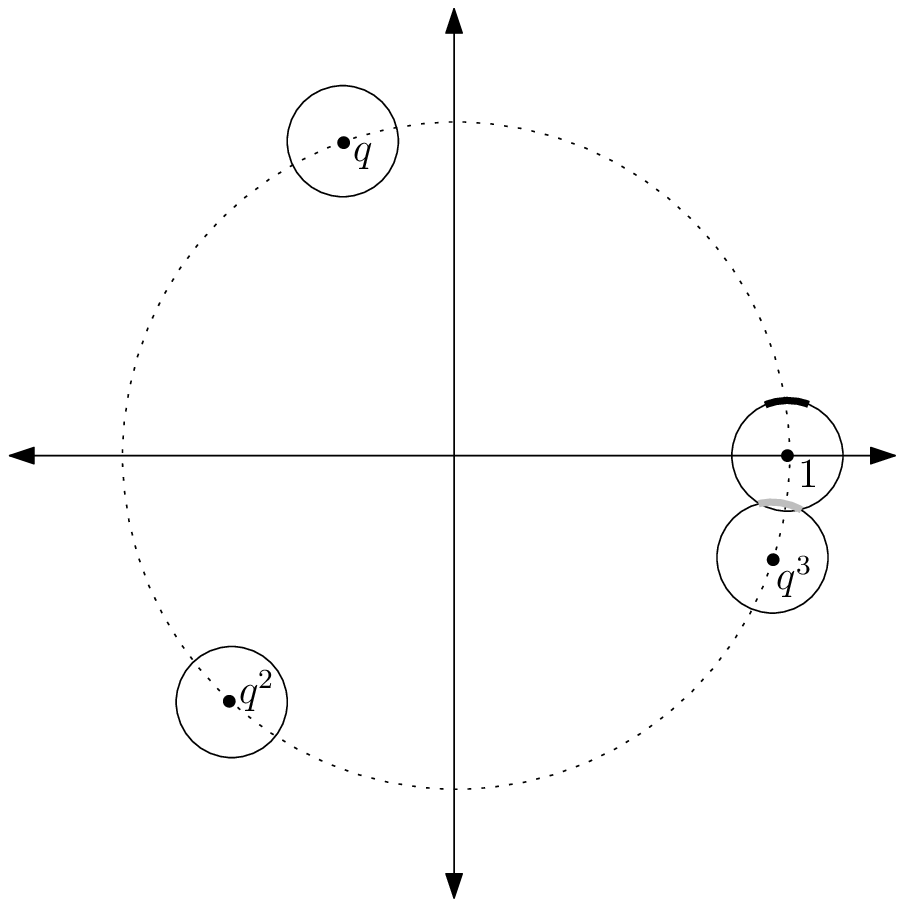}
\end{center}
\caption{For $k=4$, $q\approx e^{\frac{5}{8}\pi \i}$ and a fixed circle $\gamma_4$ around $1$, the contours $\Gamma_{(1)}(q),\Gamma_{(2)}(q)$ and $\Gamma_{(3)}(q)$ are all equal to $\gamma_4$. However, $\Gamma_{(4)}(q)$ is only the portion of $\gamma_4$ which is the complement of the dark black arc. This arc is the set of $z$ such that $q^3 z$ lies inside of $\gamma_4$ (grey arc in the figure).}\label{overlap}
\end{figure}

We may now state the general $q$ version of Lemma \ref{expandlem}.

\begin{lemma}\label{expandlemgeneral}
Fix $q\in \C^{\prime}_k$ and consider a symmetric function $G:\C^k\to \C$ and positively oriented, closed contours $\gamma_1,\ldots, \gamma_k$ such that:
\begin{itemize}
\item The contour $\gamma_k$ is a circle around 1, small enough so as not to contain $q$;
\item For all $1\leq A<B\leq k$, the interior of $\gamma_A$ contains $1$ and the image of $\gamma_B$ multiplied by $q$;
\item For all $1\leq j\leq k$, there exist deformations $D_j$ of $\gamma_j$ to $\gamma_k$ so that for all $z_1,\ldots, z_{j-1},z_{j+1},\ldots, z_k$ with $z_i\in \gamma_i$ for $1\leq i<j$, and $z_i\in \gamma_k$ for $j<i\leq k$, the function $z_j\mapsto \V(\vec{z}) G(z_1,\ldots ,z_j,\ldots, z_k)$ is analytic in a neighborhood of the area swept out by the deformation $D_j$.
\end{itemize}
Then,
\begin{equation*}
\big(\JqBoson G\big)(\vec{n}) = \sum_{\lambda\vdash k}\, \oint_{\Gamma_{(\lambda_1)}(q)} \cdots \oint_{\Gamma_{(\lambda_{\ell(\lambda)})}(q)} d\mu_{\lambda}(\vec{w}) \prod_{j=1}^{\ell(\lambda)} \frac{1}{(w_j;q)_{\lambda_j}}  \psil_{\vec{w}\circ\lambda}(\vec{n}) \, G(\vec{w}\circ \lambda).
\end{equation*}
\end{lemma}
\begin{proof}
This follows immediately from a general $q$ version of Proposition \ref{321}, using the symmetry of $G$ to remove it from the expression for $E^{q}$. The general $q$ version of Proposition \ref{321} has essentially the same proof. Let us illustrate the only modification which is required. As we deformed $\gamma_{k-1}$ through $\gamma_1$ to $\gamma_k$, we picked strings of residues of the form $z_{i_{1}}=qz_{i_{2}}, z_{i_2} = q z_{i_3}, \ldots, z_{i_{\lambda_1-1}}=q z_{i_{\lambda_1}}$ where  $i_{1}< i_{2} < \cdots < i_{\lambda_1}$. However, if $q^j z_{i_{\lambda_1}}$ is contained inside of $\gamma_k$ for some $1\leq j\leq \lambda_1-1$, then the deformation will not cross one of the poles which corresponds with this residue string. Thus, the string will be excluded from the residue expansion. This is the source of the restriction to the $\Gamma_{(n)}(q)$ contours in the above result.
\end{proof}

When $q$ is a root of unity more care must be taken due to the fact that $q^j z$ and $z$ may lie upon the same contour. In the course of proving the above result we encounter denominators which contain terms like $z-q^j z$. Thus, some additional regularization is needed to deal with this case. We do not pursue this here.

Just as with (\ref{Kexpform}), we may use the above lemma to record the analogous general $q$ formula for $\KqBoson$:
\begin{equation}
\big(\KqBoson f\big)(\vec{n}) =\sum_{\lambda\vdash k}\, \oint_{\Gamma_{(\lambda_1)}(q)} \cdots \oint_{\Gamma_{(\lambda_{\ell(\lambda)})}(q)}  d\mu_{\lambda}(\vec{w}) \prod_{j=1}^{\ell(\lambda)} \frac{1}{(w_j;q)_{\lambda_j}}  \psil_{\vec{w}\circ\lambda}(\vec{n}) \llangle f,\psir_{\vec{w}\circ \lambda}\rrangle.
\end{equation}

Theorem \ref{KqBosonId} holds exactly as stated with $\FqBoson$ and $\JqBoson$ as above, and $q\in \C^{\prime}_k$.

This in turn shows that for general $q\in \C^{\prime}_k$, Corollaries \ref{completenesscor} and \ref{north} hold under the one replacement:
$$
\oint_{\gamma_k} \cdots \oint_{\gamma_k} \longrightarrow \oint_{\Gamma_{(\lambda_1)}(q)} \cdots \oint_{\Gamma_{(\lambda_{\ell(\lambda)})}(q)}.
$$
While we expect that suitable modifications are needed for the spectral orthogonality (Proposition \ref{specorth}),  and consequently the dual Plancherel formula (Theorem \ref{KqBosonIdDual}) and the Plancherel isomorphism (Theorem \ref{isothm}), we do not pursue this here.

\section{Two semi-discrete degenerations}\label{semidisc}
This section deals with two different limits of the results contained in Sections \ref{qBosonsyst}, \ref{plansec} and \ref{appsec}. For a parameter $\e>0$ define the operator $M_{\e}$ that acts on functions $f:\Weyl{k}\to \C$ as
\begin{equation}\label{Me}
\big(M_{\e}f\big)(\vec{n}) =  \e^{n_1+\cdots + n_k} f(\vec{n}).
\end{equation}

The first limit (Sections \ref{epdef} and \ref{vandiejen}) involves keeping $q$ fixed but conjugating $\Abwd$ by the operator $M_{\e}$, and then taking $\e\to 0$. The limiting system is equivalent to the delta Bose gas considered by van Diejen \cite{vd} (for root systems of type $A$). It is also related to the $\gamma=0$ version of the $q$-Boson Hamiltonian discussed earlier in Section \ref{afortiori}. For this limiting system we show how the spectral orthogonality of the left and right eigenfunctions follows from the Cauchy-Littlewood identity for Hall-Littlewood symmetric polynomials. This $\e=0$ orthogonality is, in fact, the basis for a general $\e\geq 0$ orthogonality which is proved in Proposition \ref{trueorththm} and implies Proposition \ref{specorth} when $\e=1$.

The second limit (Section \ref{semidiscsec}) involves taking $q=e^{-\e}\to 1$ and also performing an $\e$-dependent conjugation of $\Abwd$ (involving $M_{\e}$ as well). The limiting system here is equivalent to the discrete delta Bose gas considered in \cite[Section 6]{BorCor} and \cite[Section 6]{BCS}. This system arises naturally from studying moments of a probabilistic system called the semi-discrete stochastic heat equation or equivalently the O'Connell-Yor semi-discrete directed polymer partition function (see Section \ref{sdapps}). There is yet a third limit to the continuum delta Bose gas. This is briefly discussed in Section \ref{deltabosesec}.

All other results of this section -- with the exception of Propositions \ref{trueorththm}, \ref{epzeroorth} and the lemmas involved in their proofs --  can either be proved via suitable limits of our earlier $q$-Boson particle system results, or can be proved directly for the limiting system (via the same methods as our earlier results). As such, we do no include these proofs.

In what follows we recall notation from Section \ref{notations} and write $\vec{c} = \vec{c}(\vec{n})$, $M=M(\vec{n})$, and $\vec{g}= \vec{g}(\vec{n})$ (thus suppressing the $\vec{n}$ dependence).

\subsection{An $\e$-deformed $q$-Boson particle system}\label{epdef}

We now develop an $\e$-deformation of the work of Sections \ref{qBosonsyst}, \ref{plansec} and \ref{appsec} (which at $\e=1$ corresponds to these earlier results). The main purpose of doing this is to provide a proof of the spectral orthogonality result of Proposition \ref{specorth}. In Section \ref{orthsecz} we prove that if the spectral orthogonality holds true for some $\e$, then it holds for all $\e$ (by essentially showing the derivative of the relation in $\e$ is zero). It turns out (cf. Section \ref{vandiejenorth}) that the $\e\to 0$ limit of this $\e$-deformation is implied by the Cauchy-Littlewood identity for Hall-Littlewood symmetric polynomials. This limit also makes contact with earlier work of van Diejen \cite{vd}. The general $\e$ deformation corresponds (via setting $\gamma=-\e$) to the $q$-Boson Hamiltonian discussed in Section \ref{afortiori}.

\subsubsection{Coordinate Bethe ansatz}
Fix $q\in (0,1)$ and recall the operator $M_{\e}$ from (\ref{Me}). For $\e>0$ define the operator $\Abwde$ by
$$
\Abwde :=\e  M_{\e}^{-1} \Abwd M_{\e}.
$$
Define $\Afwde$ (the matrix transpose of $\Abwde$) as $\Afwde = \e M_{\e} \Afwd M_{\e}^{-1}$ and $\Amfwde = C_q \Afwde C_{q}^{-1}$.

It is straightforward to see that
\begin{eqnarray}
\big(\Abwde f \big)(\vec{n}) &=& \sum_{i=1}^{M} (1-q^{c_i}) \big( f(\vec{n}_{c_1+\cdots +  c_{i}}^{-}) - \e f(\vec{n})\big)\\
\big(\Amfwde f\big)(\vec{n}) &=& \sum_{i=1}^{M} (1-q^{c_i}) \big(f(\vec{n}_{c_1+\cdots+c_{i-1}+1}^{+}) -\e f(\vec{n})\big).
\end{eqnarray}

Define an $\e$-deformed backward difference operator $\difbwde$ and forward difference operator $\diffwde$ which act on functions $f:\Z\to \C$ as
\begin{equation*}
\big(\difbwde f \big)(n) = f(n-1)-\e f(n),\qquad \big(\diffwde f \big)(n) =f(n+1)- \e f(n).
\end{equation*}

Define the $\e$-deformed $q$-Boson backward free generator $\Freebwde$ in the same way as in (\ref{freebwddef}), with $\difbwd$ replaced by $\difbwde$. Likewise define the $(k-1)$ $\e$-deformed $q$-Boson backward two-body boundary conditions in the same way as in (\ref{star1}), with $\difbwd$ replaced by $\difbwde$. Similarly define the $\e$-deformed $q$-Boson forward free generator $\Freebwde$ and the $(k-1)$ $\e$-deformed $q$-Boson forward two-body boundary conditions.

The result of Proposition \ref{freetrueequiv} then holds with all terms replaced by their $\e$-deformations. Our notation of an eigenfunction remains as in Definition \ref{eigdeffree} and the result of Corollary \ref{eigcorgen} applies under this $\e$-deformation.

Using the coordinate Bethe ansatz (cf. Section \ref{corbetherev}) we construct the below eigenfunctions. For all $z_1,\ldots, z_k\in \C\setminus \{\e\}$, set
\begin{eqnarray*}
\psibwde_{\vec{z}}(\vec{n})  &=& \sum_{\sigma\in S_k} \prod_{1\leq B<A\leq k} \frac{z_{\sigma(A)}-q z_{\sigma(B)}}{z_{\sigma(A)}- z_{\sigma(B)}} \, \prod_{j=1}^{k} (\e-z_{\sigma(j)})^{-n_j}\\
\psimfwde_{\vec{z}}(\vec{n}) &=& \sum_{\sigma\in S_k} \prod_{1\leq B<A\leq k} \frac{z_{\sigma(A)}-q^{-1} z_{\sigma(B)}}{z_{\sigma(A)}- z_{\sigma(B)}} \, \prod_{j=1}^{k} (\e-z_{\sigma(j)})^{n_j}\\
\psifwde_{\vec{z}}(\vec{n}) &=& C_q^{-1}(\vec{n}) \psimfwde_{\vec{z}}(\vec{n}).
\end{eqnarray*}
For fixed $\vec{n}\in \Weyl{k}$, these are symmetric Laurent polynomials in the variables $\e-z_1,\ldots, \e-z_k$.

The result of Proposition \ref{prop211} then holds with all terms  replaced by their $\e$-deformations, and the same eigenvalues $(q-1)(z_1+\cdots +z_k)$ as before. We define the right and left eigenfunctions for the forward generator $\Afwde$ as
$$\psire_{\vec{z}}(\vec{n})=  \psifwde_{\vec{z}}(\vec{n})\qquad \textrm{and}\qquad \psile_{\vec{z}}(\vec{n})=  \psibwde_{\vec{z}}(\vec{n}).$$

\subsubsection{Plancherel formulas}
The Plancherel formulas proved in Section \ref{plansec} are readily adapted to this $\e$-deformation. In order to facilitate this, we must suitably modify Definition \ref{KqTASEPdef}.

\begin{definition}\label{KqTASEPdefe}
Define the space $\LPe{k}$ of symmetric Laurent polynomials in $\e-z_1,\ldots, \e-z_k$.

Fix any set of positively oriented, closed contours $\gamma_1(\e),\ldots, \gamma_k(\e)$ chosen so that they all contain $\e$, so that the $\gamma_A(\e)$ contour contains the image of $q$ times the $\gamma_B(\e)$ contour for all $B>A$, and so that $\gamma_k(\e)$ is a small enough circle around $\e$ so as not to contain $q\e$. For use later, let us also fix contours $\gamma(\e)$ and $\gamma'(\e)$ where $\gamma(\e)$ is a positively oriented closed contour which contains $\e$ and its own image under multiplication by $q$, and $\gamma'(\e)$ contains $\gamma(\e)$ and is such that for all $z\in \gamma(\e)$ and $w\in \gamma'(\e)$, $|\e-w|>|\e-z|$.

We define $\FqBosone$ and $\JqBosone$ (as well as their compositions) in the same manner as $\FqBoson$ and $\JqBoson$ (see Definition \ref{KqTASEPdef}) by replacing $\psir,\psil, \gamma_1,\ldots,\gamma_k$ by their $\e$-deformations, and replacing $1-z_i$ by $\e-z_i$ for $1\leq i\leq k$.
\end{definition}

The Plancherel formula (Theorem \ref{KqBosonId}), the dual Plancherel formula (Theorem \ref{KqBosonIdDual}) and the Plancherel isomorphism (Theorem \ref{isothm}) all hold under the above $\e$-deformation. Defining an $\e$-deformed $q$-Pochhammer symbol $(a;q)^{\e}_{n} := \prod_{i=0}^{n-1} (\e-q^i a)$, we likewise find that Corollary \ref{completenesscor} holds with $(w_j;q)_{\lambda}$ also replaced by $(w_j;q)_{\lambda}^{\e}$. Modifying in the same manner Definition \ref{bilinearprime}, it is immediate that Corollary \ref{north} continues to hold as well.

\subsubsection{Spectral orthogonality}\label{orthsecz}
We turn here to the $\e$-deformation of the spectral orthogonality. We will provide a direct proof of this result here.

\begin{proposition}\label{trueorththm}
Consider a function $F(\vec{z})$ such that for $M$ large enough, $\prod_{i=1}^{k}(\e-z_i)^{-M} \V(\vec{z})F(\vec{z})$ is analytic in the closed exterior of $\gamma(\e)$, and consider another function $G(\vec{w})$ such that $\V(\vec{w})G(\vec{w})$ is analytic in the closed region between $\gamma(\e)$ and $\gamma'(\e)$. Then we have that
\begin{align*}
&\sum_{\vec{n}\in \Weyl{k}}\left( \oint_{\gamma(\e)}\frac{dz_1}{2\pi \i}\cdots  \oint_{\gamma(\e)}\frac{dz_k}{2\pi \i} \psire_{\vec{z}}(\vec{n})\V(\vec{z})F(\vec{z}) \right) \left( \oint_{\gamma(\e)}\frac{dw_1}{2\pi \i}\cdots  \oint_{\gamma(\e)}\frac{dw_k}{2\pi \i} \psile_{\vec{w}}(\vec{n}) \V(\vec{w}) G(\vec{w})\right)\\
&= \oint_{\gamma(\e)}\frac{dw_1}{2\pi \i}\cdots  \oint_{\gamma(\e)}\frac{dw_k}{2\pi \i}  (-1)^{\frac{k(k-1)}{2}} \prod_{j=1}^{k} (\e - w_j) \prod_{A\neq B} (w_A-qw_B) \sum_{\sigma\in S_k} \sgn(\sigma) F(\sigma \vec{w}) G(\vec{w}).
\end{align*}
\end{proposition}
\begin{remark}
The above identity may be formally rewritten as
\begin{equation}
\sum_{\vec{n}\in \Weyl{k}}\psire_{\vec{z}}(\vec{n}) \psile_{\vec{w}}(\vec{n}) \V(\vec{z}) \V(\vec{w}) = (-1)^{\frac{k(k-1)}{2}} \prod_{j=1}^{k} (\e - z_j) \prod_{A\neq B} (z_A-qz_B) \det\big[\delta_{z_i,w_j}\big]_{i,j=1}^{k}.
\end{equation}
where $\delta_{z,w}$ is the Dirac delta function for $z=w$.
\end{remark}

\begin{proof}
We start with the simple example of $k=1$. The general $k$ proof is more involved and is given after this example. When $k=1$, $\psire_{z}(n) = -(\e-z)^n$ and $\psile_{z}(n) = (\e-z)^{-n}$. Thus, the above orthogonality identity states that
\begin{equation}\label{pgfouronestar}
-\sum_{n\in \Z} \left(\oint_{\gamma(\e)} (\e-z)^n F(z) dz\right)\, \left(\oint_{\gamma(\e)} (\e-w)^{-n} G(w) dw\right) =  \oint_{\gamma(\e)} (\e-w) F(w)G(w) dw.
\end{equation}
Consider the left-hand side of (\ref{pgfouronestar}). Since for $M$ large enough, $(\e-z)^{-M} F(z)$ is analytic outside $\gamma(\e)$ if $n< -M$ then there is no residue at infinity and the integral in $z$ vanishes. Therefore the summation can be restricted to $n\geq -M$. By Cauchy's theorem we can deform the contour of integration for $w$ so as to lie on the larger contour $\gamma'(\e)$, so that $|\e-z|<|\e-w|$ for all $z\in \gamma(\e)$ and $w\in \gamma'(\e)$. The summation and integration can now be interchanged since $\sum_{n\geq -M} (\e-z)^{n} (\e-w)^{-n}$ is uniformly absolutely convergent. Evaluating this sum shows that the left-hand side of (\ref{pgfouronestar}) equals
$$
\oint_{\gamma(\e)} dz \oint_{\gamma'(\e)} dw \frac{(\e-w)^{M+1}}{(\e-z)^{M}(w-z)} F(z)G(w).
$$
We can apply the residue theorem to evaluate the $z$ integral. Due to the conditions on $F(z)$, outside $\gamma(\e)$ there is only a residue at $z=w$ (and no singularity at infinity), hence we immediately arrive at the right-hand side of (\ref{pgfouronestar}).

This method of proof does not appear so well adapted to $k>1$. Let us introduce (again for $k=1$) the method which we will employ for the general $k>1$ case. In what follows we will work formally, though all manipulations are easily justified by integrating against test functions. Rewrite the orthogonality identity (formally) as
$$
H_{z,w}(\e):=-\sum_{n\in \Z} (\e-z)^{n-1} (\e-w)^{-n} = \delta_{z,w},
$$
so that now the left-hand side depends on $\e$, whereas the right-hand side does not. Here we have called the above left-hand side $H_{z,w}(\e)$. It suffices then to show that $H_{z,w}(0)= \delta_{z,w}$ and that $\frac{d}{d\e} H_{z,w}(\e) = 0$ for all $\e\geq 0$. For $k=1$ the proof that $H_{z,w}(0)= \delta_{z,w}$ is just as above (see Section \ref{vandiejenorth} for the $k>1$ version of this result). As far as the derivative, due to telescoping
$$
\frac{d}{d\e} H_{z,w}(\e) =- \sum_{n\in \Z}\left( (n-1) (\e-z)^{n-2} (\e-w)^{-n} + (-n) (\e-z)^{n-1} (\e-w)^{-n-1}\right) =0.
$$

The general $k>1$ proof splits into two steps. The identity we wish to prove can be rewritten formally as
\begin{equation}\label{sixeight}
H_{\vec{z},\vec{w}}(\e):= \sum_{\vec{n}\in \Weyl{k}}C^{-1}_{q}(\vec{n}) \psiretilde_{\vec{z}}(\vec{n}-1) \psile_{\vec{w}}(\vec{n}) \V(\vec{z}) \V(\vec{w}) = (-1)^{\frac{k(k-1)}{2}} \prod_{A\neq B} (z_A-qz_B) \det\big[\delta_{z_i,w_j}\big]_{i,j=1}^{k},
\end{equation}
where we define $\psiretilde_{\vec{z}}(\vec{n}) := C_{q} \psire_{\vec{z}}(\vec{n})$ and $\vec{n}\pm 1= (n_1\pm 1,\ldots, n_k\pm 1)$. Here we have called the left-hand side $H_{\vec{z},\vec{w}}(\e)$. The purpose of rewriting the desired identity as above is that $\e$ only arises on the left-hand side, and not on the right.

The first step is to prove (\ref{sixeight}) for $\e=0$. This is proved as Proposition \ref{epzeroorth} and follows by identifying $\psirezero_{\vec{z}}(\vec{n})$ and $\psilezero_{\vec{z}}(\vec{n})$ with Hall-Littlewood polynomials (\ref{pgstarstar}) and (\ref{pgstarstarstar}) and using the Cauchy-Littlewood identity (\ref{pgfoursevenstarstar}).

The second step is to prove that $\frac{d}{d\e} H_{\vec{z},\vec{w}}(\e)\equiv 0$. By the Leibnitz rule the derivative either applies to $\psiretilde$ or $\psile$. Let us calculate how each of these terms transforms under differentiation. We will show that the derivatives can be written as
\begin{eqnarray}\label{decomp}
\frac{d}{d\e} \psiretilde_{\vec{z}}(\vec{n}) &=& \sum_{\vec{m}\in \Weyl{k}} C^{r}(\vec{n}, \vec{m}) \psiretilde_{\vec{z}}(\vec{m})\\
\nonumber \frac{d}{d\e} \psile_{\vec{w}}(\vec{n}) &=& \sum_{\vec{m}\in \Weyl{k}} C^{\ell}(\vec{n}, \vec{m}) \psile_{\vec{w}}(\vec{m})
\end{eqnarray}
where $C^r$ and $C^{\ell}$ are matrices independent of $\e$ (and $\vec{z}$ and $\vec{w}$) such that for each fixed $\vec{n}$, there are only finitely many $\vec{m}$ such that the $(\vec{n},\vec{m})$ entry is non-zero. Assuming the above decomposition, we can rewrite
$$
\frac{d}{d\e} H_{\vec{z},\vec{w}}(\e)= \sum_{\vec{n},\vec{m}\in \Weyl{k}} \psiretilde_{\vec{z}}(\vec{m})\psile_{\vec{w}}(\vec{n}) \left(\frac{C^r(\vec{n}-1, \vec{m})}{C_q(\vec{n})} + \frac{C^{\ell}(\vec{m}+1, \vec{n})}{C_q(\vec{m}+1)}\right).
$$
To show that the above is identically zero, it suffices to  show that for all $\vec{m},\vec{n}\in \Weyl{k}$
\begin{equation}\label{crl}
\frac{C^r(\vec{n}-1, \vec{m})}{C_q(\vec{n})} =-  \frac{C^{\ell}(\vec{m}+1, \vec{n})}{C_q(\vec{m}+1)}.
\end{equation}

Therefore, to complete this proof we must first justify the decomposition (\ref{decomp}) and then check that the matrices which arise from that decomposition satisfy the relation (\ref{crl}).
We will explicitly compute the entries of the matrix $C^r$ (and likewise those of $C^{\ell}$). Observe that
\begin{align}\label{expressab}
&\frac{d}{d\e} \psiretilde_{\vec{z}}(\vec{n}) =\\
\nonumber&\sum_{\sigma\in S_k} \prod_{1\leq B<A\leq k} \frac{z_{\sigma(A)}-q^{-1} z_{\sigma(B)}}{z_{\sigma(A)}- z_{\sigma(B)}} \, \prod_{j=1}^{k} (\e-z_{\sigma(j)})^{n_j-1} \,\left( \sum_{s=1}^{k} n_s (\e-z_{\sigma(1)})\cdots \widehat{(\e-z_{\sigma(s)})}\cdots (\e-z_{\sigma(k)}) \right),
\end{align}
where the term $\widehat{(\e-z_{\sigma(s)})}$ is excluded from the product.

For $\vec{n}$ such that $n_{i}>n_{i+1}$ for all $1\leq i\leq k-1$ (i.e., $\vec{n}$ away from the boundary of $\Weyl{k}$) the above expansion can be easily reexpressed through $\psiretilde_{\vec{z}}(\vec{m})$ with coefficients
\begin{equation*}
C^{r}(\vec{n}, \vec{m}) = \begin{cases}  n_i, & \vec{m} = \vec{n}_{i}^{-} \textrm{ for some }1\leq i\leq k-1,\\ 0, &\textrm{else}. \end{cases}
\end{equation*}

When there are clusters in $\vec{n}$ of equal values, $\vec{n}_{i}^{-}$ will lie outside of $\Weyl{k}$ for some $i$. In order to express the derivative in terms of $\psiretilde_{\vec{z}}(\vec{m})$ for $\vec{m}$ inside $\Weyl{k}$ we will use the following property.

\begin{definition}
For $\vec{n}\in \Weyl{k}$ and two functions $f,g:\C^k\to \C$, we write $f\stackrel{\vec{n}}{\approx} g$ if
\begin{equation}\label{diffeq}
\sum_{\sigma\in S_k} \prod_{1\leq B<A\leq k} \frac{z_{\sigma(A)}-q^{-1} z_{\sigma(B)}}{z_{\sigma(A)}- z_{\sigma(B)}} \, \prod_{j=1}^{k} (\e-z_{\sigma(j)})^{n_j} \big(f(z_{\sigma(1)},\ldots,z_{\sigma(k)})-g(z_{\sigma(1)},\ldots,z_{\sigma(k)})\big) = 0.
\end{equation}
\end{definition}

\begin{lemma}\label{shiftlem}
If $n_i=n_{i+1}$ then for any constants $c,c'$ and any function $f:\C^{k}\to \C$ symmetric in the $i$ and $i+1$ entries, we have that
$$
(c + c'z_{\sigma(i+1)})f(\vec{z}) \stackrel{\vec{n}}{\approx} (c + q c' z_{\sigma(i)}) f(\vec{z}).
$$
\end{lemma}
\begin{proof}
Consider the summand (for a given $\sigma\in S_k$) corresponding to the difference in (\ref{diffeq}) between the left and right sides of the above identity. All that differs between these two sides of the identity is the terms $(c + c' z_{\sigma(i+1)})$ and  $(c + qc'z_{\sigma(i)})$. Taking the difference thus yields the factor $c'q(z_{\sigma(i)}- q^{-1} z_{\sigma(i+1)})$. Combined with the term $(z_{\sigma(i+1)}- q^{-1} z_{\sigma(i)})$ coming from the product over $A>B$, this is symmetric in $z_{\sigma(i)}$ and $z_{\sigma(i+1)}$. Since $n_i=n_{i+1}$, and $F(z_{\sigma(1)},\ldots, z_{\sigma(k)})$ is symmetric in $z_{\sigma(i)}$ and $z_{\sigma(i+1)}$, it follows that everything except for the denominator of the product over $A>B$ is symmetric in $z_{\sigma(i)}$ and $z_{\sigma(i+1)}$. The denominator is antisymmetric, thus, we conclude that the difference between the left and right sides of the desired identity is equal to the symmetrization (summation in $\sigma$) of a summand which is antisymmetric in two of its variables. This implies the symmetrization is zero, and thus so is the difference.
\end{proof}

\begin{lemma}\label{cepslem}
Consider $\vec{n}=(n,\ldots, n)$, then
\begin{equation*}
(\e-z_2)\cdots (\e-z_k) \stackrel{\vec{n}}{\approx} \sum_{i=0}^{k-1} c_{\e}(k,i)(\e -z_1)\cdots (\e - z_i)
\end{equation*}
where
\begin{equation}\label{checked}
c_{\e}(k,i) = \e^{k-i-1}q^{i} \frac{(q;q)_{k-1}}{(q;q)_i}.
\end{equation}
\end{lemma}
\begin{proof}
We will proceed by induction on $k$. For $k=1$, the decomposition is clearly true with $c_{\e}(1,0)=1$. Fix that for all $k$, $c_{\e}(k,-1)=c_{\e}(k,k)=0$.  Assume the inductive hypothesis holds for $k-1$. Then
\begin{eqnarray*}
(\e-z_2)\cdots (\e-z_k) &\stackrel{\vec{n}}{\approx}& \sum_{i=0}^{k-2} c_{\e}(k-1,i)(\e -z_1)\cdots (\e - z_i) (\e-z_k)\\
&\stackrel{\vec{n}}{\approx}& \sum_{i=0}^{k-2} c_{\e}(k-1,i)(\e -z_1)\cdots (\e - z_i) \big( \e - q^{k-i-1} z_{i+1}\big)\\
&\stackrel{\vec{n}}{\approx}& \sum_{i=0}^{k-2} c_{\e}(k-1,i)(\e -z_1)\cdots (\e - z_i) \Big( \e(1-q^{k-i-1}) + q^{k-i-1}(\e-z_{i+1})\Big)\\
&\stackrel{\vec{n}}{\approx}& \sum_{i=0}^{k-1} \Big(c_{\e}(k-1,i-1) q^{k-i} + c_{\e}(k-1,i)\e (1-q^{k-i-1})\Big) (\e -z_1)\cdots (\e - z_i).
\end{eqnarray*}
The equality in the first line is from the inductive hypothesis; the equality between the first and second line is from a repeated application of Lemma \ref{shiftlem}; the equality between the second and third line is from rewriting $\e - q^{k-i-1} z_{i+1} = \e(1-q^{k-i-1}) + q^{k-i-1}(\e-z_{i+1})$; and the equality between the third and fourth line is from gathering terms with the same factor of
$(\e -z_1)\cdots (\e - z_i)$, and using the notational assumption $c_{\e}(k,-1)=c_{\e}(k,k)=0$.

The above reasoning shows that the $c_{\e}(k,i)$ are the solutions to the following recurrence relation:
\begin{equation*}
c_{\e}(k,i) = c_{\e}(k-1,i-1) q^{k-i} + c_{\e}(k-1,i) \e (1-q^{k-i-1})
\end{equation*}
with $c_{\e}(k,-1)=c_{\e}(k,k)=0$ and $c_{\e}(1,0)=1$.
It is immediate to check that the right-hand side of (\ref{checked}) uniquely satisfies this relation.
\end{proof}

These $c_{\e}(k,i)$ coefficients combine to form the entries of $C^{r}$. Define
\begin{equation}\label{dexp}
D_{\e}(k,p) = \sum_{j=0}^{p} c_{\e}(k-p+j,j) = \e^{k-p-1} (1-q)^{k-p-1} \frac{(k)!_q (k-p-1)!_q}{(k-p)!_q (p)!_q}.
\end{equation}
Then it is immediate from Lemma \ref{cepslem} that for $\vec{n}=(n,\ldots, n)$
\begin{equation*}
\sum_{\ell=1}^{k} (\e-z_{1})\cdots \widehat{(\e-z_{s})}\cdots (\e-z_{k}) \stackrel{\vec{n}}{\approx} \sum_{p=0}^{k-1} D_{\e}(k,p) (\e-z_{1})\cdots (\e-z_{p}).
\end{equation*}

This shows that for $\vec{n}=(n,\ldots, n)$, $C^{r}(\vec{n},\vec{m}) = n D_{\e}(k,p)$ for all $\vec{m}=(n,\ldots, n, n-1,\ldots,n-1)$ with exactly $1\leq p\leq k$ entries equal to $n$ (and the remaining equal to $n-1$). For all other $\vec{m}$, $C^{r}(\vec{n},\vec{m})=0$.

To state the general $\vec{n}$ result, define
$$\vec{n}_{[a,b]}^{\pm} = (n_1,\ldots,n_{a-1}, n_{a}\pm 1,n_{a+1}\pm 1,\ldots, n_{b}\pm 1,n_{b+1},\ldots, n_{k}).$$

\begin{lemma}\label{rlem}
For $\vec{n}\in \Weyl{k}$, set $(c_1,\ldots, c_M) = \vec{c}(\vec{n})$ (cf. Section \ref{notations}). Then for all $1\leq i\leq M$ and $0\leq p\leq c_{i}-1$,
if $\vec{m} = \vec{n}_{[c_{1}+\cdots +c_{i-1}+p+1,c_{1}+\cdots +c_{i}]}^{-}$ then
\begin{equation}
C^{r}(\vec{n}, \vec{m}) =  n_{c_1+\cdots +c_{i}} D_{\e}(c_i,p),
\end{equation}
and for all other $\vec{m}$, $C^{r}(\vec{n}, \vec{m})=0$.
\end{lemma}
\begin{proof}
This same reasoning as was applied in the $\vec{n}= (n,\ldots, n)$ case may be applied to each cluster of $\vec{n}$, thus resulting in the statement of the lemma.
\end{proof}

We likewise find that
\begin{lemma}\label{llem}
For $\vec{n}\in \Weyl{k}$, set $(c_1,\ldots, c_M) = \vec{c}(\vec{n})$ (cf. Section \ref{notations}). Then for all $1\leq i\leq M$ and $1\leq p\leq c_{i}$,
if $\vec{m} = \vec{n}_{[c_{1}+\cdots +c_{i-1}+1,c_{1}+\cdots +c_{i-1}+p]}^{+}$ then
\begin{equation}
C^{\ell}(\vec{n}, \vec{m}) =  -n_{c_1+\cdots +c_{i}} D_{\e}(c_i,c_i-p),
\end{equation}
and for all other $\vec{m}$, $C^{\ell}(\vec{n}, \vec{m})=0$.
\end{lemma}
\begin{proof}
This is completely analogous to the $C^{r}$ case. Note that the negative sign comes from differentiation of terms of the type $(1-z_j)^{-n_j}$.
\end{proof}

\begin{lemma}
The matrices defined in Lemmas \ref{rlem} and \ref{llem} satisfy (\ref{crl}).
\end{lemma}

\begin{proof}
Let us consider the example of $\vec{n} = (n,\ldots, n,n-1,\ldots, n-1)$ with $a>0$ entries equal to $n$ and $b\geq 0$ entries  equal to $n-1$ (thus a total length of $a+b$). Then, the only $\vec{m}$ for which $C^r(\vec{n}-1, \vec{m})$ and $C^{\ell}(\vec{m}+1,\vec{n})$ are non-zero is $\vec{m} = (n-1,\ldots, n-1,n-2,\ldots, n-2)$ with $p<a$ entries of $n-1$ and $b+p$ entries of $n-2$. For these $\vec{n}$ and $\vec{m}$ we compute
\begin{align*}
&C^{r}(\vec{n}-1, \vec{m}) = (n-1) D_{\e}(a,p), \qquad C_q(\vec{n}) = (-1)^k q^{-\frac{k(k-1)}{2}} (a)!_q (b)!_q,\\ &C^{\ell}(\vec{m}+1, \vec{n}) =-(n-1) D_{\e}(a+b-p,b),\qquad C_{q}(\vec{m}) =  (-1)^k q^{-\frac{k(k-1)}{2}} (p)!_q (a+b-c)!_q.
\end{align*}
From the explicit formula (\ref{dexp}) for the $D_{\e}$ we immediately confirm the  (\ref{crl}) is satisfied. The general $\vec{n}\in \Weyl{k}$ case follows similarly from the above calculation.
\end{proof}

We return to complete the proof of Proposition \ref{trueorththm}. As we have now justified the decomposition (\ref{decomp}) and checked that the matrices $C^r$ and $C^{\ell}$ satisfy the relation (\ref{crl}), it follows that $\frac{d}{d\e} H_{\vec{z},\vec{w}}(\e)=0$. Since Proposition \ref{epzeroorth} (yet to be proved below) shows that (\ref{sixeight}) holds in the limit $\e\to 0$, these two facts together complete the proof of this proposition.
\end{proof}

\subsection{Van Diejen's delta Bose gas}\label{vandiejen}

We consider the $\e\to 0$ limit of the $\e$-deformed $q$-Boson particle system considered in Section \ref{epdef}. Our primary purpose is to prove Proposition \ref{trueorththm} for $\e=0$, as well as to make contact with earlier work of van Diejen \cite{vd}. We also briefly remark on the limiting versions of the Plancherel formula.

The $\e\to 0$ limits of the generators $\Abwde$, $\Amfwde$ and $\Afwde$ are straightforward (set $\e=0$), and likewise for $\difbwde, \diffwde, \Freebwde, \Freefwde$ and the boundary conditions. In particular the limits of the backward and forward generators ($\Freebwde$ and $\Freefwde$) applied to a function $u(\vec{n})$ give $\sum_{i=1}^{k} u(\vec{n}_{i}^{-})$ and $\sum_{i=1}^{k} u(\vec{n}_{i}^{+})$ (respectively). The backward two-body boundary conditions limit to:
\begin{equation*}
\textrm{for all } 1\leq i\leq k-1\qquad \big(u(\vec{n}_{i}^{-}) - qu(\vec{n}_{i+1}^{-})\big)\big\vert_{\vec{n}:n_i=n_{i+1}} \equiv 0,
\end{equation*}
and the forward two-body boundary conditions limit to:
\begin{equation*}
\textrm{for all } 1\leq i\leq k-1\qquad \big(q u(\vec{n}_{i}^{+}) - u(\vec{n}_{i+1}^{+})\big)\big\vert_{\vec{n}:n_i=n_{i+1}} \equiv 0.
\end{equation*}

This limiting system arose earlier in work of van Diejen \cite{vd}, where he additionally worked in the context of general root systems (the above system corresponds with type $A$). In that work, van Diejen uses Bethe ansatz to construct eigenfunctions for these systems and proves a Plancherel formula (for general root systems) by appealing to earlier work of Macdonald \cite{M1,M3}. Taking the $\e\to 0$ limit of our results also yields these eigenfunctions and such a Plancherel formula (for the type $A$ case).

Let us take the $\e\to 0$ limit of the left and right eigenfunctions of $\Afwde$ and call these $\psilezero_{\vec{z}}(\vec{n})$ and $\psirezero_{\vec{z}}(\vec{n})$. For all $z_1,\ldots, z_k\in \C\setminus \{0\}$, we then have
\begin{eqnarray*}
\psilezero_{\vec{z}}(\vec{n})  &=& \sum_{\sigma\in S_k} \prod_{1\leq B<A\leq k} \frac{z_{\sigma(A)}-q z_{\sigma(B)}}{z_{\sigma(A)}- z_{\sigma(B)}} \, \prod_{j=1}^{k} (-z_{\sigma(j)})^{-n_j}\\
\psirezero_{\vec{z}}(\vec{n}) &=& C_q^{-1}(\vec{n}) \sum_{\sigma\in S_k} \prod_{1\leq B<A\leq k} \frac{z_{\sigma(A)}-q^{-1} z_{\sigma(B)}}{z_{\sigma(A)}- z_{\sigma(B)}} \, \prod_{j=1}^{k} (-z_{\sigma(j)})^{n_j}.
\end{eqnarray*}
We identify these with Hall-Littlewood symmetric polynomials $P_{\lambda}(x;t)$ and $Q_{\lambda}(y;t)$ \cite[Chapter III]{M}: For $\vec{n}$ such that $n_k\geq 1$,
\begin{equation}\label{pgstarstar}
\psilezero_{\vec{z}}(\vec{n}) = (-1)^{\sum_{i=1}^{k} n_i} (1-q)^{-k} Q_{\vec{n}}(z_1^{-1},\ldots, z_{k}^{-1};q),
\end{equation}
while for $\vec{n}$ such that $n_k\geq 0$,
\begin{equation}\label{pgstarstarstar}
\psirezero_{\vec{z}}(\vec{n}) = (-1)^{\sum_{i=1}^{k} n_i} (-1)^k P_{\vec{n}}(z_1,\ldots, z_{k};q).
\end{equation}

The $\e\to 0$ limits of $\FqBosone$ and $\JqBosone$ (as well as their compositions) may thus be expressed in terms of Hall-Littlewood symmetric polynomials. In taking this $\e\to 0$ limit of $\JqBosone$, there is no clear way to take the limit of the nested contours $\gamma_1(\e),\ldots\gamma_k(\e)$. However, we may just as well take all contours to be $\gamma(\e)$ and then there is a clear limit of the integrand and the contour. Indeed, at $\e=0$, we may use any contour $\gamma(0)$ which encloses $0$ and its own image under multiplication by $q$.
In this way we can show $\e\to 0$ limits of the Plancherel formula (Theorem \ref{KqBosonId}), the dual Plancherel formula (Theorem \ref{KqBosonIdDual}) and the Plancherel isomorphism (Theorem \ref{isothm}).

\subsubsection{Proof of $\e=0$ spectral orthogonality}\label{vandiejenorth}
We state and prove the $\e=0$ case of Proposition \ref{trueorththm} by directly appealing to the Cauchy-Littlewood identity for Hall-Littlewood symmetric polynomials. This result then serves as the first step in proving the general $\e$ result of Proposition \ref{trueorththm}, and consequently Proposition \ref{specorth} as well.

\begin{proposition}\label{epzeroorth}
Fix a positively oriented circle $\gamma(0)\subset \C$ which contains 0 and the image of $\gamma(0)$ under multiplication by $q$, and let $\gamma'(0)$ be a circle which contains $\gamma(0)$ and is such that $|z|<|w|$ for all $z\in \gamma(0)$ and $w\in \gamma'(0)$. Consider a function $F(\vec{z})$ such that for $M$ large enough, $\prod_{i=1}^{k}(-z_i)^{-M} \V(\vec{z})F(\vec{z})$ is analytic in the closed exterior of $\gamma(0)$, and consider another function $G(\vec{w})$ such that $\V(\vec{w})G(\vec{w})$ is analytic in the closed region between $\gamma(0)$ and $\gamma'(0)$. Then we have that
\begin{align}\label{epzeroortheqn}
&\sum_{\vec{n}\in \Weyl{k}}\left( \oint_{\gamma(0)}\frac{dz_1}{2\pi \i}\cdots  \oint_{\gamma(0)}\frac{dz_k}{2\pi \i} \psirezero_{\vec{z}}(\vec{n})\V(\vec{z})F(\vec{z}) \right) \left( \oint_{\gamma(0)}\frac{dw_1}{2\pi \i}\cdots  \oint_{\gamma(0)}\frac{dw_k}{2\pi \i} \psilezero_{\vec{w}}(\vec{n}) \V(\vec{w}) G(\vec{w})\right)\\
\nonumber &= \oint_{\gamma(0)}\frac{dw_1}{2\pi \i}\cdots  \oint_{\gamma(0)}\frac{dw_k}{2\pi \i}  (-1)^{\frac{k(k+1)}{2}} \prod_{j=1}^{k} w_j \prod_{A\neq B} (w_A-qw_B) \sum_{\sigma\in S_k} \sgn(\sigma) F(\sigma \vec{w}) G(\vec{w}).
\end{align}
\end{proposition}

\begin{proof}
Consider a single term in the left-hand side of (\ref{epzeroortheqn}) with fixed $\vec{n}\in \Weyl{k}$. Due to the analyticity condition on $F(\vec{z})$, if $n_i< -M$ then the corresponding term integrates to zero (this follows from the lack of singularity at infinity). Thus we can restrict the summation to $n_1\geq \cdots \geq n_{k}\geq -M$. Now we can deform the integration contour for $w$ from $\gamma(0)$ to $\gamma'(0)$ and we may appeal to the Hall-Littlewood polynomial representations for $\psirezero$ and $\psilezero$ given in (\ref{pgstarstar}) and (\ref{pgstarstarstar}). Using homogeneity of the Hall-Littlewood polynomials, and taking the integration in $z$ and $w$ outside of the summation, we are left with (for $M$ sufficiently positive) the left-hand side of (\ref{epzeroortheqn}) equal to
\begin{align*}
&(1-q)^{-k} (-1)^k\oint_{\gamma(0)}\frac{dz_1}{2\pi \i}\cdots \oint_{\gamma(0)}\frac{dz_k}{2\pi \i} \oint_{\gamma'(0)} \frac{dw_1}{2\pi \i}\cdots \oint_{\gamma'(0)} \frac{dw_k}{2\pi \i} F(\vec{z}) G(\vec{w})\\
&\left(\frac{w_1\cdots w_k}{z_1\cdots z_k}\right)^{M+1} \sum_{n_1\geq \cdots \geq n_k\geq 1} P_{\vec{n}}(z_1,\ldots, z_k;q)Q_{\vec{n}}(w_1^{-1},\ldots, w_k^{-1};q) \V(\vec{z})\V(\vec{w}).
\end{align*}
The fact that the summation and the two integrations may be interchanged follows from the fact that for $z_j\in \gamma(0)$ and $w_j\in \gamma'(0)$ the summation is uniformly absolutely convergent. By residue considerations we may change the summation over $n_1\geq \cdots \geq n_k\geq 1$ above to $n_1\geq \cdots \geq n_k\geq 0$ without changing the value of the integrals. We now utilize the Cauchy-Littlewood identity for Hall-Littlewood polynomials \cite[III (4.4)]{M}  to explicitly evaluate the summation
\begin{equation}\label{pgfoursevenstarstar}
\sum_{n_1\geq \cdots \geq n_k\geq 0} P_{\vec{n}}(z_1,\ldots, z_k;q)Q_{\vec{n}}(w_1^{-1},\ldots, w_k^{-1};q)  = \prod_{i,j=1}^{k} \frac{w_j- q z_i}{w_j - z_i}.
\end{equation}

Therefore, the left-hand side of (\ref{epzeroortheqn}) is equal to
\begin{align*}
&(1-q)^{-k} (-1)^k\oint_{\gamma(0)}\frac{dz_1}{2\pi \i}\cdots \oint_{\gamma(0)}\frac{dz_k}{2\pi \i} \oint_{\gamma'(0)} \frac{dw_1}{2\pi \i}\cdots \oint_{\gamma'(0)} \frac{dw_k}{2\pi \i}\\
&\left(\frac{w_1\cdots w_k}{z_1\cdots z_k}\right)^{M+1} \prod_{i,j=1}^{k} \frac{w_j- q z_i}{w_j - z_i} \V(\vec{z})\V(\vec{w}) F(\vec{z})G(\vec{w}).
\end{align*}
We now apply the residue theorem to evaluate the $z$ integrals. For each $z_i$ there are residues outside of $\gamma(0)$ at $w_1,\ldots, w_k$ (there are no singularities at $z_j=\infty$). However, due to the term $\V(\vec{z})$, no two of these $z$ variables can pick the same $w$ variable (otherwise the Vandermonde gives zero). Thus, the $z$ integrals expand as a sum over residues from $z_i= w_{\sigma(i)}$ for $1\leq i\leq k$ and $\sigma\in S_k$:
\begin{align*}
&\oint_{\gamma(0)}\frac{dz_1}{2\pi \i}\cdots \oint_{\gamma(0)}\frac{dz_k}{2\pi \i} \left(\frac{w_1\cdots w_k}{z_1\cdots z_k}\right)^{M+1} \prod_{i,j=1}^{k} \frac{w_j- q z_i}{w_j - z_i} \V(\vec{z})\V(\vec{w}) F(\vec{z}) \\
& = (1-q)^{k} (-1)^{-\frac{k(k-1)}{2}} \prod_{j=1}^{k} w_j \prod_{A\neq B} (w_A-qw_B)  \sum_{\sigma\in S_k} \sgn(\sigma) F(w_{\sigma(1)},\ldots, w_{\sigma(k)}).
\end{align*}
The integral of this expression times $(1-q)^{-k} (-1)^k G(\vec{w})$ is readily matched to the integral of the right-hand side of (\ref{epzeroortheqn}).
\end{proof}
%
%
%
%
%

\subsection{Another discrete delta Bose gas}\label{semidiscsec}
Let $q=e^{-\e}$ and recall the operator $M_{\e}$ from (\ref{Me}). Define the operator $\Abwdsd_{\e}$ by
$$
\Abwdsd_{\e} = \e^{-2} M_{\e}^{-1} \Abwd M_{\e} + k\big(\e^{-1}  -\tfrac{3}{2}\big) \Id.
$$
Then, as $\e\to 0$, $\Abwdsd_{\e} \to \Abwdsd$ where we call this limit the semi-discrete backward generator (because of the connection with the semi-discrete stochastic heat equation, see below). Below we record (without proof) scaling limits (under the above scaling) of various results we have proved earlier in Sections \ref{qBosonsyst}, \ref{plansec} and \ref{appsec}.

\subsubsection{Coordinate Bethe ansatz}

The semi-discrete backward generator\footnote{This is not a Markov generator since it is not stochastic.} $\Abwdsd$ acts on functions $f:\Weyl{k}\to \C$ as
\begin{equation*}
\big(\Abwdsd f \big)(\vec{n}) = \sum_{i=1}^{M} \left(c_i \big(\difbwd_{c_1+\cdots + c_{i}}f\big)(\vec{n}) + \frac{c_i(c_i-1)}{2} f(\vec{n})\right).
\end{equation*}
The semi-discrete forward generator is the matrix transpose of $\Abwdsd$ and acts on functions $f:\Weyl{k}~\to~\C$ as
\begin{equation*}
\big(\Afwdsd f \big)(\vec{n}) = \sum_{i=1}^{M} \left(\big( (c_{i-1}+1) \bfone_{g_i=1} + \bfone_{g_i>1}\big) f(\vec{n}_{c_1+\cdots+c_{i-1}+1}^{+})  - c_i f(\vec{n}) +  \frac{c_i(c_i-1)}{2} f(\vec{n})\right)
\end{equation*}
where by convention, for $i=1$ we set $c_{i-1}+1=c_1+\cdots + c_{i-1}+1=1$.

The function $C:\Weyl{k}\to \R$ depends only on the list $\vec{c}(\vec{n})=(c_1,\ldots, c_M)$ of cluster sizes for $\vec{n}$ via
\begin{equation}\label{Cn}
C(\vec{n}) = (-1)^k \prod_{i=1}^{M} (c_i)!.
\end{equation}
We will use the notation of $C$ and $C^{-1}$ for multiplication operators by $C(\vec{n})$ and $\big(C(\vec{n})\big)^{-1}$. It is clear that these operators commute with $R$ (see Section \ref{notations}).

The semi-discrete conjugated forward generator is defined by
$$\Amfwdsd = C \Afwdsd C^{-1},$$
and it acts as
\begin{equation*}
\big(\Amfwdsd f\big)(\vec{n}) =  \sum_{i=1}^{M}\left( c_i \big( \diffwd_{c_1+\cdots + c_{i-1}+1} f\big)(\vec{n}) + \frac{c_i(c_i-1)}{2} f(\vec{n})\right).
\end{equation*}

Using the backward and forward difference operators (Section \ref{notations}) we define the semi-discrete backward free generator $\Freebwdsd$ which acts on functions $u:\Z^k\to \C$ as
\begin{equation*}
\big(\Freebwdsd u \big)(\vec{n}) = \sum_{i=1}^{k} \big(\difbwd_i u\big)(\vec{n})
\end{equation*}
where $\difbwd_i$ acts as $\difbwd$ in the variable $n_i$.
We say that the function $u:\Z^k\to \C$  satisfies the $(k-1)$ semi-discrete backward two-body boundary conditions if
\begin{equation*}
\textrm{for all } 1\leq i\leq k-1\qquad \big(\difbwd_i - \difbwd_{i+1} -1\big)u\big\vert_{\vec{n}:n_i=n_{i+1}} \equiv 0.
\end{equation*}

Similarly, the semi-discrete forward free generator $\Freefwdsd$ acts on functions $u:\Z^k\to \C$ as
\begin{equation*}
\big(\Freefwdsd u\big)(\vec{n}) =\sum_{i=1}^{k} \diffwd_i u(\vec{n})
\end{equation*}
where $\diffwd_i$ acts as $\diffwd$ in the variable $n_i$. We say that the function $u:\Z^k\to \C$  satisfies the $(k-1)$ semi-discrete forward two-body boundary conditions if
\begin{equation*}
\textrm{for all } 1\leq i\leq k-1\qquad \big(1+\diffwd_i - \diffwd_{i+1}\big)u\big\vert_{\vec{n}:n_i=n_{i+1}} \equiv 0.
\end{equation*}

\begin{proposition}\label{freetrueequivsd}
If $u:\Z^k\to \C$ satisfies the $(k-1)$ semi-discrete backward (respectively, forward) two-body boundary conditions, then for $\vec{n}\in \Weyl{k}$,
\begin{equation*}
\big(\Freebwdsd u\big)(\vec{n}) = \big(\Abwdsd u\big)(\vec{n})\qquad  \big(\textrm{respectively, } \big(\Freefwdsd u\big)(\vec{n}\big) = \big(\Amfwdsd u\big)(\vec{n})\big).
\end{equation*}
\end{proposition}

\begin{remark}\label{commremsd}
Proposition \ref{freetrueequivsd} implies that
$R^{-1} \Abwdsd R = \Amfwdsd = C \Afwdsd C^{-1}$ or equivalently
$$
\Abwdsd = (R C) \Afwdsd (R C)^{-1}
$$
showing that $\Abwdsd$ and $\Afwdsd$ are related via a similarity transform.
\end{remark}

\begin{definition}\label{eigdeffreesd}
A function $\psibwdsd:\Z^k\to \C$ is an eigenfunction of the semi-discrete backward free generator with $(k-1)$ two-body boundary conditions if $\psibwdsd$ is an eigenfunction for the semi-discrete backward free generator that satisfies the $(k-1)$ semi-discrete backward two-body boundary conditions. We likewise define what it means for a function $\psimfwdsd:\Z^k\to \C$ to be an eigenfunction of the semi-discrete forward free generator with $(k-1)$ two-body boundary conditions.
\end{definition}

The following is a corollary of Proposition \ref{freetrueequivsd}.

\begin{corollary}\label{eigcorgensd}
Any eigenfunction $\psibwdsd:\Z^k\to \C$ for the semi-discrete backward free generator with $(k-1)$ two-body boundary conditions is, when restricted to $\vec{n}\in \Weyl{k}$, an eigenfunction for the semi-discrete backward generator $\Abwdsd$ with the same eigenvalue.

Similarly, any eigenfunction $\psimfwdsd:\Z^k\to \C$ for the semi-discrete forward free evolution equation with $(k-1)$ two-body boundary conditions is, when restricted to $\vec{n}\in \Weyl{k}$, an eigenfunction for the semi-discrete conjugated forward generator $\Amfwdsd$ with the same eigenvalue. In turn, $C^{-1}\psimfwdsd$ is an eigenfunction for the semi-discrete forward generator $\Afwdsd$ with the same eigenvalue.
\end{corollary}

We may apply the coordinate Bethe ansatz of Section \ref{corbetherev} to the semi-discrete backward and forward free generators with $(k-1)$ two-body boundary conditions. The backward generator eigenfunction $\psibwdsd_{\vec{z}}(\vec{n})$ given below arose in earlier work of \cite{takeyama}.

\begin{definition}\label{eigdefnsd}
For all $z_1,\ldots, z_k\in \C\setminus \{0\}$, set
\begin{eqnarray*}
\psibwdsd_{\vec{z}}(\vec{n})  &=& \sum_{\sigma\in S_k} \prod_{1\leq B<A\leq k} \frac{z_{\sigma(A)}-z_{\sigma(B)}-1}{z_{\sigma(A)}- z_{\sigma(B)}} \, \prod_{j=1}^{k} (z_{\sigma(j)})^{-n_j}\\
\psimfwdsd_{\vec{z}}(\vec{n}) &=& \sum_{\sigma\in S_k} \prod_{1\leq B<A\leq k} \frac{z_{\sigma(A)}-z_{\sigma(B)}+1}{z_{\sigma(A)}- z_{\sigma(B)}} \, \prod_{j=1}^{k} (z_{\sigma(j)})^{n_j}\\
\psifwdsd_{\vec{z}}(\vec{n}) &=& C^{-1}(\vec{n}) \psimfwdsd_{\vec{z}}(\vec{n}).
\end{eqnarray*}
These are symmetric Laurent polynomials in the variables $z_1,\ldots, z_k$.
\end{definition}

\begin{proposition}\label{prop211sd}
For all $z_1,\ldots, z_k\in \C\setminus \{0\}$, $\psibwdsd_{\vec{z}}(\vec{n})$
is an eigenfunction for the semi-discrete backward free generator with $(k-1)$ two-body boundary conditions with eigenvalue $\sum_{i=1}^{k} (z_i-1)$. The restriction of $\psibwdsd_{\vec{z}}(\vec{n})$ to $\vec{n}\in \Weyl{k}$ is consequently an eigenfunction for the semi-discrete backward generator $\Abwdsd$ with the same eigenvalue.

Similarly, for all $z_1,\ldots, z_k\in \C\setminus \{0\}$, $\psimfwdsd_{\vec{z}}(\vec{n})$
is an eigenfunction  for the semi-discrete forward free generator with $(k-1)$ two-body boundary conditions with eigenvalue $\sum_{i=1}^{k} (z_i-1)$. The restriction of $\psimfwdsd_{\vec{z}}(\vec{n})$ to $\vec{n}\in \Weyl{k}$ is consequently an eigenfunction for the semi-discrete conjugated forward generator $\Amfwdsd$ with the same eigenvalue. The restriction of $\psifwdsd_{\vec{z}}(\vec{n})$ to $\vec{n}\in \Weyl{k}$ is likewise an eigenfunction for the semi-discrete forward generator $\Afwdsd$ with the same eigenvalue.
\end{proposition}

\begin{remark}\label{BosonHamiltoniansymsd}
We may extend the eigenfunctions of Definition \ref{eigdefnsd} so as to be defined for all of $\Z^k$ (rather than $\Weyl{k}$) by fixing that the value for a general $\vec{n}\in \Z^k$ is the same as the value of $\sigma \vec{n}\in \Weyl{k}$ where $\sigma\in S_k$ is a permutation of the elements of $\vec{n}$ taking it into $\Weyl{k}$. It is possible to write down an operator on all of $\Z^k$ for which these extensions (which we write with the same notation) are still eigenfunctions. One finds (cf. \cite[Proposition 6.3, (C)]{BCS}) that
$$
\left[\sum_{i=1}^{k} \difbwd_i + \sum_{1\leq i<j\leq k} \bfone_{n_i=n_j}\difbwd_i\right] \psibwdsd_{\vec{z}}(\vec{n}) = \left(\sum_{i=1}^{k} (z_i-1)\right)\, \psibwdsd_{\vec{z}}(\vec{n}).
$$
This can be thought of as a discrete version of the delta Bose gas considered later in Section \ref{deltabosesec}.

\end{remark}

Proposition \ref{prop211sd} along with the fact that $\Afwdsd$ is the transpose of $\Abwdsd$ implies that
\begin{equation}\label{eqn8sd}
\Afwdsd \psifwdsd_{\vec{z}}(\vec{n}) = \left(\sum_{i=1}^{k} (z_i-1)\right) \psifwdsd_{\vec{z}}(\vec{n}), \qquad \psibwdsd_{\vec{z}}(\vec{n}) \Afwdsd  = \psibwdsd_{\vec{z}}(\vec{n}) \left(\sum_{i=1}^{k} (z_i-1)\right),
\end{equation}
showing the $\psifwdsd_{\vec{z}}(\vec{n})$ and $\psibwdsd_{\vec{z}}(\vec{n})$ are (respectively) right and left eigenfunctions for $\Afwdsd$ with eigenvalue $\sum_{i=1}^{k} (z_i-1)$. This motivates the following.

\begin{definition}\label{leftrighteigsd}
For any $\vec{z}=(z_1,\ldots, z_k)\in \big(\C\setminus\{0\}\big)^k$ define
\begin{equation*}
\psirsd_{\vec{z}}(\vec{n})=  \psifwdsd_{\vec{z}}(\vec{n})=  C^{-1}(\vec{n}) \, \sum_{\sigma\in S_k} \prod_{1\leq B<A\leq k} \frac{z_{\sigma(A)}- z_{\sigma(B)}+1}{z_{\sigma(A)}- z_{\sigma(B)}} \, \prod_{j=1}^{k} (z_{\sigma(j)})^{n_j}.
\end{equation*}
Likewise define
\begin{equation*}
\psilsd_{\vec{z}}(\vec{n})=  \psibwdsd_{\vec{z}}(\vec{n})= \sum_{\sigma\in S_k} \prod_{1\leq B<A\leq k} \frac{z_{\sigma(A)}-z_{\sigma(B)}-1}{z_{\sigma(A)}- z_{\sigma(B)}} \, \prod_{j=1}^{k} (z_{\sigma(j)})^{-n_j}.
\end{equation*}
\end{definition}

These eigenfunctions are limits of $\psil_{\vec{z}}$ and $\psir_{\vec{z}}$ of Definition \ref{leftrighteig} when $q=e^{-\e}\to 1$ and the $z_j$ variables in Definition \ref{leftrighteig} are replaced by $q^{z_j}$.

We could have just as well defined the right and left eigenfunctions with respect to the operator $\Abwd$. Finally, observe the symmetry of $\psilsd$ and $\psirsd$ with respect to the space-reflection operator $R$
\begin{equation}\label{psisymsd}
\big(R \psilsd_{\vec{z}}\big)(\vec{n}) =C(\vec{n}) \psirsd_{\vec{z}}(\vec{n}), \qquad\qquad
\big(R\psirsd_{\vec{z}}\big)(\vec{n}) = C^{-1}(\vec{n}) \psilsd_{\vec{z}}(\vec{n}).
\end{equation}

\subsubsection{Plancherel formulas}

\begin{definition}\label{KqTASEPdefsd}
Define the space $\LPsd{k}$ of symmetric Laurent polynomials in the variables $z_1,\ldots, z_k$.

The semi-discrete transform $\FqBosonsd$ takes functions $f\in \CP{k}$ into functions $\FqBosonsd f\in\LPsd{k}$ via
\begin{equation*}
\big(\FqBosonsd f\big)(\vec{z}) = \llangle f, \psirsd_{\vec{z}}\rrangle.
\end{equation*}

Fix any set of positively oriented, closed contours $\tilde{\gamma}_1,\ldots, \tilde{\gamma}_k$ chosen so that they all contain $0$, so that the $\tilde{\gamma}_A$ contour contains the image of $1$ plus the $\tilde{\gamma}_B$ contour for all $B>A$, and so that $\tilde{\gamma}_k$ is a small enough circle around 0 so as not to contain $1$.

The (candidate) semi-discrete inverse transform $\JqBosonsd$ takes functions $G\in\LPsd{k}$ into functions $\JqBosonsd G\in \CP{k}$ via
\begin{equation}\label{JqBosontranssd}
\big(\JqBosonsd G\big)(\vec{n}) = (-1)^k \oint_{\tilde{\gamma}_1} \frac{dz_1}{2\pi \i} \cdots \oint_{\tilde{\gamma}_k} \frac{dz_k}{2\pi \i}  \prod_{1\leq A<B\leq k} \frac{z_A-z_B}{z_A-z_B -1}\, \prod_{j=1}^{k} (z_{j})^{-n_j-1}\, G(\vec{z}).
\end{equation}

The composition of the transform and (candidate) inverse transform takes functions $f\in \CP{k}$ into functions $\KqBosonsd f\in \CP{k}$ via
\begin{eqnarray}\label{KqBosondefsd}
\big(\KqBosonsd f \big)(\vec{n}) &=& \big(\JqBosonsd\FqBosonsd f\big)(\vec{n}) \\
\nonumber &=&  (-1)^k\oint_{\tilde{\gamma}_1} \frac{dz_1}{2\pi \i} \cdots \oint_{\tilde{\gamma}_k} \frac{dz_k}{2\pi \i}  \prod_{1\leq A<B\leq k} \frac{z_A-z_B}{z_A-z_B-1}\, \prod_{j=1}^{k} (z_{j})^{-n_j-1} \, \llangle f,\psirsd_{\vec{z}}\rrangle.
\end{eqnarray}

The composition of the (candidate) inverse transform and the transform takes functions $G\in \LPsd{k}$ into functions $\KqBosonsdDual G\in \LPsd{k}$ via
\begin{eqnarray}\label{KqBosonsddefDual}
\big(\KqBosonsdDual G \big)(\vec{n}) &=& \big(\FqBosonsd \JqBosonsd G\big)(\vec{z}) \\
\nonumber &=& (-1)^k \sum_{\vec{n}\in \Weyl{k}} \psirsd_{\vec{z}}(\vec{n}) \oint_{\tilde\gamma_1} \frac{dz_1}{2\pi \i} \cdots \oint_{\tilde\gamma_k} \frac{dz_k}{2\pi \i}  \prod_{1\leq A<B\leq k} \frac{z_A-z_B}{z_A-z_B-1}\, \prod_{j=1}^{k} (z_{j})^{-n_j-1} G(\vec{z}).
\end{eqnarray}
\end{definition}

\begin{lemma}\label{expandlemsd}
Consider a symmetric function $F:\C^k\to \C$ and positively oriented, closed contours $\tilde{\gamma}_1,\ldots, \tilde{\gamma}_k$ such that:
\begin{itemize}
\item The contour $\tilde{\gamma}_k$ is a circle around 0, small enough so as not to contain $1$;
\item For all $1\leq A<B\leq k$, the interior of $\tilde{\gamma}_A$ contains $0$ and the image of $\tilde{\gamma}_B$ plus 1;
\item For all $1\leq j\leq k$, there exist deformations $D_j$ of $\tilde{\gamma}_j$ to $\tilde{\gamma}_k$ so that for all $z_1,\ldots, z_{j-1},z_{j+1},\ldots, z_k$ with $z_i\in \tilde{\gamma}_i$ for $1\leq i<j$, and $z_i\in \tilde{\gamma}_k$ for $j<i\leq k$, the function $z_j\mapsto \V(\vec{z}) F(z_1,\ldots ,z_j,\ldots, z_k)$ is analytic in a neighborhood of the area swept out by the deformation $D_j$.
\end{itemize}
Then,
\begin{equation*}
\big(\JqBosonsd G\big)(\vec{n}) = \sum_{\lambda\vdash k}\, \oint_{\tilde{\gamma}_k} \cdots \oint_{\tilde{\gamma}_k} d\tilde{\mu}_{\lambda}(\vec{w}) \prod_{j=1}^{\ell(\lambda)} \frac{1}{(w_j)_{\lambda_j}}  \psilsd_{\vec{w}\tilde{\circ}\lambda}(\vec{n}) \, G(\vec{w}\tilde{\circ} \lambda),
\end{equation*}
where the definition of the notation  $\vec{w}\tilde{\circ}\lambda$ and $(w)_n$ is given in Section \ref{notations} and where, for a partition  $\lambda = 1^{m_1}2^{m_2}\cdots$,
\begin{equation*}
d\tilde{\mu}_{\lambda}(\vec{w}) = \frac{1}{m_1! m_2!\cdots} \det\left[\frac{1}{w_i +\lambda_i -w_j}\right]_{i,j=1}^{\ell(\lambda)} \prod_{j=1}^{\ell(\lambda)} \frac{dw_j}{2\pi \i}.
\end{equation*}
\end{lemma}


We may now record the scaling limits of the Plancherel formula (Theorem \ref{KqBosonId}), the dual Plancherel formula (Theorem \ref{KqBosonIdDual}) and the Plancherel isomorphism (Theorem \ref{isothm}).

We start with the Plancherel formula.

\begin{theorem}\label{KqBosonIdsd}
The semi-discrete transform $\FqBosonsd$ induces an isomorphism between the space $\CP{k}$ and its image with inverse given by $\JqBosonsd$.  Equivalently, $\KqBosonsd$ acts as the identity operator on $\CP{k}$.
\end{theorem}

A key fact in the proof of this (which follows in the exact same manner as Theorem \ref{KqBosonId}) is the the following property of $\KqBosonsd$: For any functions $f,g\in\CP{k}$,
\begin{equation}\label{claimedkbossd}
\llangle \KqBosonsd f,g\rrangle = \llangle f ,(C R)^{-1} \KqBosonsd (C R g)\rrangle.
\end{equation}
Lemma \ref{expandlemsd} yields the following expansion
\begin{equation}\label{Kexpformsd}
\big(\KqBosonsd f\big)(\vec{n}) = \sum_{\lambda\vdash k}\, \oint_{\tilde{\gamma}_k} \cdots \oint_{\tilde{\gamma}_k} d\tilde{\mu}_{\lambda}(\vec{w}) \prod_{j=1}^{\ell(\lambda)} \frac{1}{(w_j)_{\lambda_j}}  \psilsd_{\vec{w}\tilde{\circ}\lambda}(\vec{n}) \llangle f,\psirsd_{\vec{w}\tilde{\circ} \lambda}\rrangle.
\end{equation}

We have shown in Theorem \ref{KqBosonIdsd} that $\KqBoson=\Id$ when restricted to functions in $\CP{K}$. Similarly,
\begin{equation*}
(C R)^{-1} \KqBosonsd (C R) = \Id.
\end{equation*}

Let us turn now to the dual Plancherel formula.

\begin{theorem}\label{KqBosonIdDualsd}
The inverse semi-discrete transform $\JqBosonsd$ induces an isomorphism between the domain $\LPsd{k}$ and its image with inverse given by $\FqBosonsd$. Equivalently, $\KqBosonsdDual$ acts as the identity operator on $\LPsd{k}$.
\end{theorem}
It should also be possible (as conjectured in Remark \ref{conjrem}) to extend this dual Plancherel formula to a more degenerate class of functions than presently considered.

We may combine Theorems \ref{KqBosonIdsd} and \ref{KqBosonIdDualsd} to arrive at the following result.
\begin{definition}\label{bilinearprimesd}
Define a bilinear pairing acting on two functions $F(\vec{z})$ and $G(\vec{z})$  by
\begin{equation*}
\llangle F,G\rranglesd = \sum_{\lambda\vdash k}\, \oint_{\tilde{\gamma}_k} \cdots \oint_{\tilde{\gamma}_k} d\tilde{\mu}_{\lambda}(\vec{w}) \prod_{j=1}^{\ell(\lambda)} \frac{1}{(w_j)_{\lambda_j}}  F(\vec{w}\tilde{\circ}\lambda) G(\vec{w}\tilde{\circ} \lambda).
\end{equation*}
\end{definition}

\begin{theorem}\label{isothmsd}
The semi-discrete transform $\FqBosonsd$ induces an isomorphism between $\CP{k}$ and $\LPsd{k}$ with inverse given by $\JqBosonsd$. Moreover, for any $f,g\in \CP{k}$
\begin{equation}\label{res1sd}
\llangle f,g \rrangle = \llangle \FqBosonsd f,\FqBosonsd g\rranglesd,
\end{equation}
and for any $F,G\in \LPsd{k}$
\begin{equation}\label{res2sd}
\llangle \JqBosonsd F,\JqBosonsd G \rrangle = \llangle F,G\rranglesd.
\end{equation}
\end{theorem}

\subsubsection{Completeness and biorthogonality}
Corollaries \ref{completenesscorsd} and \ref{northsd} are immediate consequences of Theorem \ref{KqBosonIdsd}. Proposition \ref{specorthsd} can either be proved directly (as done for Proposition \ref{specorth}) or derived from Theorem \ref{KqBosonIdDualsd} (assuming that theorem is proved directly via a limit transition from Theorem \ref{KqBosonIdDual}).

\begin{corollary}\label{completenesscorsd}
Any function $f\in\CP{k}$ can be expanded as
\begin{equation}\label{completenessfwdsd}
f(\vec{n}) = \sum_{\lambda\vdash k}\, \oint_{\tilde{\gamma}_k} \cdots \oint_{\tilde{\gamma}_k} d\tilde{\mu}_{\lambda}(\vec{w}) \prod_{j=1}^{\ell(\lambda)} \frac{1}{(w_j)_{\lambda_j}}  \psilsd_{\vec{w}\tilde{\circ}\lambda}(\vec{n}) \llangle f,\psirsd_{\vec{w}\tilde{\circ} \lambda}\rrangle,
\end{equation}
and also as
\begin{equation}\label{completenessbwdsd}
f(\vec{n}) = \sum_{\lambda\vdash k}\, \oint_{\tilde{\gamma}_k} \cdots \oint_{\tilde{\gamma}_k} d\tilde{\mu}_{\lambda}(\vec{w}) \prod_{j=1}^{\ell(\lambda)} \frac{1}{(w_j)_{\lambda_j}}  \psirsd_{\vec{w}\tilde{\circ}\lambda}(\vec{n}) \llangle \psilsd_{\vec{w}\tilde{\circ} \lambda},f \rrangle.
\end{equation}
\end{corollary}

\begin{corollary}\label{northsd}
For $\vec{n},\vec{m}\in \Weyl{k}$, regard $\psilsd(\vec{n})$ and $\psilsd(\vec{m})$ as functions $\big(\psilsd(\vec{n})\big)(\vec{z})=\psilsd_{\vec{z}}(\vec{n})$ and $\big(\psirsd(\vec{m})\big)(\vec{z})=\psirsd_{\vec{z}}(\vec{m})$. Then
\begin{equation*}
\llangle \psilsd(\vec{n}), \psirsd(\vec{m}) \rranglesd = \bfone_{\vec{n}=\vec{m}}.
\end{equation*}
\end{corollary}


\begin{proposition}\label{specorthsd}
Consider functions $F,G\in \LPsd{k}$. Then we have that
\begin{align*}
&\sum_{\vec{n}\in \Weyl{k}}\left( \oint_{\tilde\gamma_1}\frac{dz_1}{2\pi \i}\cdots  \oint_{\tilde\gamma_k}\frac{dz_k}{2\pi \i} \psirsd_{\vec{z}}(\vec{n})\V(\vec{z})F(\vec{z}) \right) \left( \oint_{\tilde\gamma_1}\frac{dw_1}{2\pi \i}\cdots  \oint_{\tilde\gamma_k}\frac{dw_k}{2\pi \i} \psilsd_{\vec{w}}(\vec{n}) \V(\vec{w}) G(\vec{w})\right)\\
&= \oint_{\tilde\gamma_1}\frac{dw_1}{2\pi \i}\cdots  \oint_{\tilde\gamma_k}\frac{dw_k}{2\pi \i}  (-1)^{\frac{k(k+1)}{2}} \prod_{j=1}^{k} w_j \prod_{A\neq B} (w_A-w_B-1) \sum_{\sigma\in S_k} \sgn(\sigma) F(\sigma \vec{w}) G(\vec{w}).
\end{align*}
\end{proposition}

\subsubsection{Applications to the semi-discrete stochastic heat equation}\label{sdapps}
The following results follow from Corollary \ref{completenesscorsd}. We then apply them to the analysis of the semi-discrete stochastic heat equation.

\begin{corollary}\label{bwdappsd}
For $f_0\in\CP{k}$ the backward equation
\begin{equation*}
\frac{d}{dt} f(t;\vec{n}) = \big(\Abwdsd f\big)(t;\vec{n})
\end{equation*}
with $f(0;\vec{n}) = f_0(\vec{n})$ is uniquely solved by
\begin{eqnarray}\label{ftvecnsd}
\nonumber f(t;\vec{n}) &=& \big(e^{t\Abwdsd} f_0\big)(\vec{n}) = \sum_{\lambda\vdash k}\, \oint_{\tilde{\gamma}_k} \cdots \oint_{\tilde{\gamma}_k} d\tilde{\mu}_{\lambda}(\vec{w}) \prod_{j=1}^{\ell(\lambda)} \frac{e^{t\tilde{E}(\vec{w}\tilde{\circ} \lambda)}}{(w_j)_{\lambda_j}}  \psilsd_{\vec{w}\tilde{\circ}\lambda}(\vec{n}) \llangle f_0,\psirsd_{\vec{w}\tilde{\circ} \lambda}\rrangle \\
 &=& \oint_{\tilde{\gamma}_1} \frac{dz_1}{2\pi \i} \cdots \oint_{\tilde{\gamma}_k} \frac{dz_k}{2\pi \i}  \prod_{1\leq A<B\leq k} \frac{z_A-z_B}{z_A-z_B-1}\, \prod_{j=1}^{k} (z_{j})^{-n_j-1} e^{t \tilde{E}(\vec{z})} \,\llangle f_0,\psirsd_{\vec{z}}\rrangle
\end{eqnarray}
where $\tilde{E}(\vec{z}) = \sum_{i=1}^{k} (z_i-1)$.
\end{corollary}

\begin{corollary}\label{fwdappsd}
For $f_0\in\CP{k}$ the forward equation
\begin{equation*}
\frac{d}{dt} f(t;\vec{n}) = \big(\Afwdsd f\big)(t;\vec{n})
\end{equation*}
with $f(0;\vec{n}) = f_0(\vec{n})$ is uniquely solved by
\begin{eqnarray}\label{ftvecnfwdsd}
\nonumber f(t;\vec{n}) &=& \big(e^{t\Afwdsd} f_0\big)(\vec{n}) = \sum_{\lambda\vdash k}\, \oint_{\tilde{\gamma}_k} \cdots \oint_{\tilde{\gamma}_k} d\tilde{\mu}_{\lambda}(\vec{w}) \prod_{j=1}^{\ell(\lambda)} \frac{1}{(w_j)_{\lambda_j}}e^{t\tilde{E}(\vec{w}\tilde{\circ} \lambda)} \psirsd_{\vec{w}\tilde{\circ}\lambda}(\vec{n}) \llangle \psilsd_{\vec{w}\tilde{\circ} \lambda} , f_0 \rrangle \\
&=& \oint_{\tilde{\gamma}_1} \frac{dz_1}{2\pi \i} \cdots \oint_{\tilde{\gamma}_k} \frac{dz_k}{2\pi \i}  \prod_{1\leq A<B\leq k} \frac{z_A-z_B}{z_A-z_B-1}\, \prod_{j=1}^{k} (z_{j})^{n_{k-j+1}-1} e^{t\tilde{E}(\vec{z})} \,\llangle \psilsd_{\vec{z}}, f_0\rrangle
\end{eqnarray}
where $\tilde{E}(\vec{z}) = \sum_{i=1}^{k} (z_i-1)$.
\end{corollary}

We may utilize the above solution to the backward equation to compute integral formulas for the moments of the probabilistic model we now introduce. Though we phrase it in terms of a coupled system of stochastic ODEs, it is equivalent to the partition function of the O'Connell-Yor semi-discrete directed polymer \cite{OY} (see also \cite{OCon} for more recent developments in the study of this model).

\begin{definition}\label{def:semidiscSHE}
The {\it semi-discrete stochastic heat equation} with initial data $Z_0$ is the system of stochastic ODEs
\begin{equation*}
dZ(t,n) = \difbwd Z(t,n) dt + Z(t,n) dB_n,\qquad Z(0,n) = Z_0(n),
\end{equation*}
where $Z:\R_{\geq 0}\times \Z \to \R_{\geq 0}$ and the $B_i$ are independent standard Brownian motions. We will assume that $Z_0(n)$ is compactly supported.
\end{definition}

\begin{remark}
As shown in \cite[Theorem 4.1.26]{BorCor} (see also \cite[Proposition 6.2]{BCS}) the above system is a $q\to 1$ scaling limit of $q$-TASEP with finitely many particles (cf. Section \ref{qtasepsec}).
\end{remark}

Let $\bar{Z}(t;\vec{n})=\EE\left[\prod_{i=1}^{k} Z(t,n_i)\right]$. Using the Feynman-Kac representation (and the so-called replica method) one shows (cf. \cite[Section 6.2]{BCS}) that these moments satisfy
$$
\frac{d}{dt}\bar{Z}(t,\vec{n}) = \left[\sum_{i=1}^{k} \difbwd_i + \sum_{1\leq i<j\leq k} \bfone_{n_i=n_j}\right] \bar{Z}(t,\vec{n}).
$$
This can be thought of as a discrete version of the delta Bose gas considered in Section \ref{deltabosesec}.

As in Remark \ref{BosonHamiltoniansymsd} (see also \cite[Proposition 6.3, (C)]{BCS}) one sees that for $\vec{n}\in \Weyl{k}$, $\bar{Z}(t;\vec{n}) = h(t;\vec{n})$ where
$h(t;\vec{x})$ solves the backward equation
$$
\frac{d}{dt} h(t;\vec{n}) = \Abwdsd h(t;\vec{n}), \qquad h(0;\vec{n}) = \EE\left[\prod_{i=1}^{k} Z_0(n_i)\right].
$$

It follows from Corollary \ref{bwdappsd} that we may solve the above equation succinctly in terms of contour integrals.
\begin{corollary}
For $f_0(\vec{n})=h(0;\vec{n})$ of compact support, $h(t;\vec{n})$ is given by the right-hand side of the first and the second lines of (\ref{ftvecnsd}).
\end{corollary}

Let us apply this for the example of the semi-discrete stochastic heat equation with delta initial data $Z_0(n) = \bfone_{n=1}$. This is the limit of step initial data for $q$-TASEP. One should be able to likewise work out the half-stationary limit, though we do not pursue that here. The following result previously appeared in \cite[Section 5.2.2]{BorCor}.

\begin{proposition}\label{dualiyresultsd}
For delta initial data $Z_0(n) = \bfone_{n=1}$ we have
\begin{equation}\label{FdefqTASEPsd}
 \EE\left[\prod_{i=1}^{k} Z(t,n_i)\right]  =  \oint_{\tilde{\gamma}_1} \frac{dz_1}{2\pi \i} \cdots \oint_{\tilde{\gamma}_k} \frac{dz_k}{2\pi \i} \prod_{1\leq A<B\leq k} \frac{z_A-z_B}{z_A-z_B-1} \prod_{j=1}^{k} (z_j)^{-n_j} e^{t(z_j-1)}.
\end{equation}
\end{proposition}

\section{Appendix}\label{appendixsec}

\subsection{Continuum delta Bose gas}\label{deltabosesec}
We briefly describe a further scaling limit of the results of Section \ref{semidiscsec} (as well as those earlier in the paper). We are informal here and do not give precise statements (see references for such statements). The analog (and limit) here of the $q$-Boson particle system and the discrete delta Bose gas is the continuum delta Bose gas, and the analog (and limit) of $q$-TASEP and the semi-discrete stochastic heat equation is the continuum stochastic heat equation (or equivalently the Kardar-Parisi-Zhang equation), see \cite[Section 6]{BorCor}, \cite{BCF}, or \cite{QMR}. A rigorous treatment of this sort of limiting procedure (starting from the system considered in Section \ref{vandiejen}) is given in \cite{vd}.

\begin{definition}
The $1+1$ dimensional continuum stochastic heat equation with initial data $\mathcal{Z}_0$ is the solution to the stochastic partial differential equation:
\begin{equation}
\frac{d}{dt} \mathcal{Z} = \frac{1}{2} \frac{d^2}{dx^2} \mathcal{Z} + \xi \mathcal{Z}, \qquad \mathcal{Z}(0,x)=\mathcal{Z}_0(x),
\end{equation}
where $\xi$ is space-time Gaussian white noise (see \cite{ACQ,ICReview} for more details and background).
\end{definition}

In a similar manner as for the semi-discrete stochastic heat equation, the joint moments (time $t$ fixed, spatial variables $x_i$ varying) satisfy simple closed evolution equations. In particular (cf. \cite{BertiniCancrini} or \cite[Section 6]{BorCor}) for $\vec{x}\in \R^k$, defining $\bar{\mathcal{Z}}(t,\vec{x})=  \EE\left[\prod_{i=1}^{k} \mathcal{Z}(t,x_i)\right]$
$$
\frac{d}{dt}\bar{\mathcal{Z}}(t,\vec{x}) = \left[\sum_{i=1}^{k} \frac{1}{2}\frac{d^2}{dx_i^2} + \sum_{1\leq i<j\leq k} \delta(x_i-x_j)\right] \bar{\mathcal{Z}}(t,\vec{x}),
$$
where $\delta(x)$ is the Dirac delta function. This is often called the (imaginary time) delta Bose gas with attractive coupling constant, or the (imaginary time) Lieb-Liniger model with attractive delta interaction. Uniqueness of solutions to the above equation (within a class of solutions with suitable control on their growth in $\vec{x}$) should be attainable, but we are unaware of an exact reference in the literature.

It is standard in the physics literature (cf. \cite{Dot,CDR}) to postulate the following reduction: Let
$$\Weylct{k} = \Big\{\vec{x} = (x_1,\ldots,x_k)\in \R^k\big\vert x_1\leq \cdots \leq x_k\Big\}.$$
Then, for $\vec{x}\in \Weylct{k}$, $\bar{\mathcal{Z}}(t;\vec{x}) = u(t;\vec{x})$ where
$u(t;\vec{x})$ solves
$$
\frac{d}{dt} u(t;\vec{x}) = \sum_{i=1}^{k} \frac{1}{2} \frac{d^2}{dx_i^2} u(t;\vec{x})
$$
and satisfies the two-body boundary conditions that
\begin{equation*}
\textrm{for all } 1\leq i\leq k-1\qquad \left(\frac{d}{dx_i} - \frac{d}{dx_{i+1}} - 1\right)u\Big\vert_{\vec{x}:x_i+0=x_{i+1}} \equiv 0.
\end{equation*}
Here $x_i+0=x_{i+1}$ means the limit as $x_i\to x_{i+1}$ from below. To our knowledge, this reduction has not been rigorously justified.

In the work of Heckman and Opdam \cite{HO}, the above system (Laplacian plus two-body boundary conditions) is called the Yang system \cite{Yang1,Yang2} (see also \cite{Gaudin2,HO} for other root systems) corresponding to the type $A$ root system and an attractive coupling constant (of magnitude 1 in our setting).

\begin{definition}\label{conteig}
For any $\vec{z}\in \C^k$ define
\begin{equation*}
\psirct_{\vec{z}}(\vec{x})= \tilde{C}^{-1}(\vec{x})  \sum_{\sigma\in S_k} \prod_{1\leq B<A\leq k} \frac{z_{\sigma(A)}-z_{\sigma(B)}+1}{z_{\sigma(A)}- z_{\sigma(B)}} \, \prod_{j=1}^{k} e^{-x_j z_j}.
\end{equation*}
As in (\ref{Cn}), the function $\bar{C}:\Weylct{k}\to \R$ depends only on the list $\vec{c}(\vec{x})=(c_1,\ldots, c_M)$ of cluster sizes for $\vec{x}$ via
\begin{equation}\label{Cnsd}
\bar{C}(\vec{x}) = (-1)^k \prod_{i=1}^{M} (c_i)!.
\end{equation}

Likewise define
\begin{equation*}
\psilct_{\vec{z}}(\vec{x})=   \sum_{\sigma\in S_k} \prod_{1\leq B<A\leq k} \frac{z_{\sigma(A)}-z_{\sigma(B)}-1}{z_{\sigma(A)}- z_{\sigma(B)}} \, \prod_{j=1}^{k} e^{x_j z_j}.
\end{equation*}
\end{definition}

These are right and left eigenfunctions for the Yang system.

Following Heckman-Opdam \cite{HO}, we now define the analogs of $\FqBoson$ and $\JqBoson$, though without carefully specifying the domains on which they act. The Fourier-Yang transform $\FqBosonct$ takes functions $f:\Weylct{k}\to \C$ into symmetric functions $\FqBosonct f:\C^k\to \C$ via
\begin{equation*}
\big(\FqBosonct f\big)(\vec{z}) = \int_{\vec{x}\in \Weylct{k}} f(\vec{x}) \psirct_{\vec{z}}(\vec{x})d\vec{x}.
\end{equation*}
Note that the $\bar{C}^{-1}(\vec{x})$ factor in defining $\psirct$ plays no role in this integration as it is identically equal to $(-1)^k$ except on a set of measure 0.

The (candidate) Fourier-Yang inverse transform $\JqBosonct$ takes symmetric functions $F:\C^k\to \C$ into functions $\JqBosonct F:\Weylct{k}\to \C$ via
\begin{equation}\label{JqBosontransct}
\big(\JqBosonct F\big)(\vec{x}) = \int_{\alpha_1+\i\R} \frac{dz_1}{2\pi \i} \cdots \int_{\alpha_k+\i\R} \frac{dz_k}{2\pi \i}  \prod_{1\leq A<B\leq k} \frac{z_A-z_B}{z_A-z_B-1}\, \prod_{j=1}^{k} e^{x_jz_j}\, F(\vec{z})
\end{equation}
where $\alpha_1>\alpha_2+1>\alpha_3+2>\cdots>\alpha_k + k-1$.


\begin{lemma}\label{expandlemct}
Consider a symmetric function $F:\C^k\to \C$ and a set of real numbers $\alpha_1,\ldots,\alpha_k$ which satisfy
\begin{enumerate}
\item For all $1\leq j \leq k-1$, $\alpha_j>\alpha_{j+1}+1$;
\item For all $1\leq j \leq k$ and $z_1,\ldots,z_k$ such that $z_i\in \alpha_i+\i \R$ for $1\leq i<j$ and $z_i\in \alpha_k+\i \R$ for $j<i\leq k$, the function $z_j\mapsto \V(\vec{z}) F(z_1,\ldots ,z_j,\ldots, z_k)$ is analytic in the complex domain $\{z: \alpha_k\leq \Real(z) \leq \alpha_j\}$ and is bounded in modulus on that domain by $\const\, \Imag(z_j)^{-\delta}$ for some constants $\const,\delta>0$ (depending on $z_1,\ldots,z_{j-1},z_{j+1},\ldots,z_k$ but not $z_j$).
\end{enumerate}
Then,
\begin{equation*}
\big(\JqBosonct F\big)(\vec{n}) = \sum_{\lambda\vdash k}\, \int_{\alpha_k+\i\R} \cdots \int_{\alpha_k+\i\R} d\tilde{\mu}_{\lambda}(\vec{w}) \psilct_{\vec{w}\tilde{\circ}\lambda}(\vec{n}) \, F(\vec{w}\tilde{\circ} \lambda),
\end{equation*}
where $\vec{w}\tilde{\circ} \lambda$ is as in (\ref{wlambdasd}) and $d\tilde{\mu}_{\lambda}$ is as in Lemma \ref{expandlemsd}.
\end{lemma}


Heckman and Opdam \cite{HO} showed that when restricted to a suitable class of functions, the Fourier-Yang transform $\FqBosonct$ is an isomorphism onto its image with inverse $\JqBosonct$. Or, equivalently, on this set of functions $\KqBosonct=\Id$. This is the analog of Theorem \ref{KqBosonId} and the method of proof we employed herein is adapted from their work. Just as before, from this Plancherel formula and Lemma \ref{expandlemct}, one proves completeness of the Bethe ansatz (see \cite[Remark 6.2.5]{BorCor} for a detailed history of this question).

There should be an analog of the dual Plancherel formula (Theorem \ref{KqBosonIdDual}) and the Plancherel isomorphism (Theorem \ref{isothm}) in this setting, though we have not seen this in the mathematics literature. The physics literature contains some work in the direction of spectral orthogonality in this context \cite{CalCaux, Dot}.

In the past few years there have been many works in the physics literature using the above delta Bose gas to study the moments of the stochastic heat equation with various types of initial data (flat / half-flat \cite{CDprl,CDlong}, stationary \cite{ImSa,ImSaKPZ,ImSaKPZ2}, or more general \cite{CorwinQuastel}) or at different times \cite{Dot4}. These calculations are counterparts to the $q$-Boson Plancherel theorem discussed in Section \ref{qtasepsec}.


\subsection{Shrinking nested contours}

The following proposition is a variant of Proposition 3.2.1 of \cite{BorCor}. The proof therein is a verification type inductive proof. The proof we present below is direct and inspired by \cite{HO}.

\begin{proposition}\label{321}
Given a set of positively oriented, closed contours $\gamma_1,\ldots,\gamma_k$ and a function $F(z_1,\ldots, z_k)$ which satisfy
\begin{itemize}
\item The contour $\gamma_k$ is a circle around 1, small enough so as not to contain $q$;
\item For all $1\leq A<B\leq k$, the interior of $\gamma_A$ contains the image of $\gamma_B$ multiplied by $q$;
\item For all $1\leq j\leq k$, there exist deformations $D_j$ of $\gamma_j$ to $\gamma_k$ so that for all $z_1,\ldots, z_{j-1},z_{j+1},\ldots, z_k$ with $z_i\in \gamma_i$ for $1\leq i<j$, and $z_i\in \gamma_k$ for $j<i\leq k$, the function $z_j\mapsto \V(\vec{z}) F(z_1,\ldots ,z_j,\ldots, z_k)$ is analytic in a neighborhood of the area swept out by the deformation $D_j$.
\end{itemize}
Then we have the following residue expansion identity:
\begin{equation}\label{qAK}
\oint_{\gamma_1} \frac{dz_1}{2\pi \i} \cdots \oint_{\gamma_k} \frac{dz_k}{2\pi \i} \prod_{1\leq A<B\leq k} \frac{z_A-z_B}{z_A-q z_B} F(z_1,\ldots, z_k)= \sum_{\lambda\vdash k}\, \oint_{\gamma_k} \cdots \oint_{\gamma_k} d\mu_{\lambda}(\vec{w}) E^q(\vec{w}\circ\lambda),
\end{equation}
where $\vec{w}\circ\lambda$ was defined in (\ref{wlambda}). Here we define
\begin{equation*}
E^q(z_1,\ldots, z_k) =   \sum_{\sigma\in S_k} \prod_{1\leq B<A\leq k} \frac{z_{\sigma(A)}-qz_{\sigma(B)}}{z_{\sigma(A)}-z_{\sigma(B)}} F(z_{\sigma(1)},\ldots, z_{\sigma(k)})
\end{equation*}
and
\begin{equation*}
d\mu_{\lambda}(\vec{w}) = \frac{(1-q)^{k}(-1)^k q^{-\frac{k(k-1)}{2}}}{m_1! m_2!\cdots} \det\left[\frac{1}{w_i q^{\lambda_i} -w_j}\right]_{i,j=1}^{\ell(\lambda)} \prod_{j=1}^{\ell(\lambda)} w_j^{\lambda_j} q^{\frac{\lambda_j(\lambda_j-1)}{2}} \frac{dw_j}{2\pi \i}
\end{equation*}
where $\lambda= 1^{m_1}2^{m_2}\cdots$.
\end{proposition}

\begin{proof}
First notice that for $k=1$ the result follows immediately. Hence, in what follows we assume $k\geq 2$.

We proceed sequentially and deform (using the deformation $D_{k-1}$ afforded from the hypotheses of the theorem) the $\gamma_{k-1}$ contour to $\gamma_k$, and then deform (using the deformation $D_{k-2}$ afforded from the hypotheses of the theorem) $\gamma_{k-2}$ contour to $\gamma_k$, and so on until all contours have been deformed to $\gamma_k$. However, due to the $z_A-qz_B$ terms in the denominator of the integrand, during the deformation of $\gamma_A$ we may encounter simple poles at the points $z_A=q z_B$, for $B>A$. The residue theorem implies that the integral on the left-hand side of (\ref{qAK}) can be expanded into a summation over integrals of possibly few variables (all along $\gamma_k$) whose integrands correspond to the various possible residue subspaces coming from these poles.

Our proof splits into three basic steps. First, we identify the residual subspaces upon which our integral is expanded via residues. This brings us to equation (\ref{RHS2}). Second, we show that these subspaces can be brought to a canonical form via the action of some $\sigma\in S_k$. This enables us to simplify the summation over the residual subspaces to a summation over partitions $\lambda\vdash k$ and certain subsets of permutations $\sigma$ in $S_k$. Inspection of those terms corresponding to $\sigma\in S_k$ not in these subsets shows that they have zero residue contribution and hence the summation can be completed to include all of $S_k$. This brings us to equation (\ref{above3}). And third, we rewrite the function whose residue we are computing as the product of an $S_k$ invariant function (which contains all of the poles related to the residual subspace) and a remainder function. We use Lemma \ref{reslemma} to evaluate the residue of the $S_k$ invariant function and we identify the summation over $\sigma\in S_k$ of the substitution into the remainder function as exactly giving the $E^q$ function in the statement of the proposition we are presently proving.

\smallskip
\noindent {\bf Step 1:} It is worthwhile to start with an example ($k=3$). Figure \ref{qcontourshifting} accompanies the example and illustrates the deformations and locations of poles.

\begin{figure}
\begin{center}
\includegraphics[scale=.8]{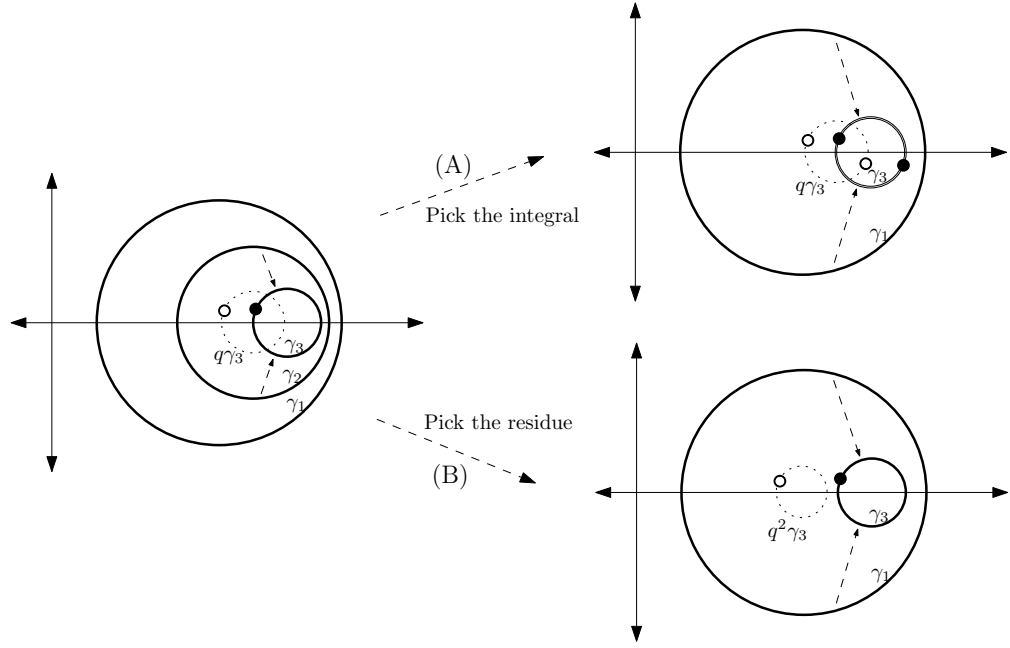}
\end{center}
\caption{The expansion of the $k=3$ nested contour integral as explained in Example \ref{ex2}. On the left-hand side, the $\gamma_2$ contour is deformed to the $\gamma_3$ contour and a pole is crossed along $q\gamma_3$ at the point $qz_3$ (here the point $z_3$ is drawn as a small black bullet and its location is on the solid line circle; and the point $q z_3$ is drawn as a small white bullet and its location is on the dotted line circle). On the upper right-hand side the effect of picking the integral is shown. The $\gamma_3$ contour is represented as a (doubled) solid line circle since both $z_3$ and $z_2$ are integrated along it (these correspond to the two black bullets). The white bullets are along $q\gamma_3$ (the dotted line circle) and represent $qz_3$ and $qz_2$. As the $\gamma_1$ contour is deformed to $\gamma_3$ these residues must be taken. On the lower right-hand side the effect of picking the initial residue at $z_2=qz_3$ is shown. As the $\gamma_1$ contour is deformed to $\gamma_3$ a pole is encountered along $q^2\gamma_3$ at the location $q^2 z_3$ (as before $z_3$ is the black bullet and $q^2 z_3$ is the white bullet.}\label{qcontourshifting}
\end{figure}

\begin{example}\label{ex2}
When $k=3$ the integrand on the left-hand side of (\ref{qAK}) contains the fractions
\begin{equation}\label{cross}
\frac{z_1-z_2}{z_1-qz_2}\, \frac{z_1-z_3}{z_1-qz_3}\,\frac{z_2-z_3}{z_2-qz_3},
\end{equation}
times the function $F(z_1,z_2,z_3)$ which (by the hypotheses of the proposition) does not have poles between $\gamma_j$ and $\gamma_k$ (for $j=1,2,3$). Thinking of $z_3$ as fixed along the contour $\gamma_3$, we begin by deforming the $\gamma_2$ contour to $\gamma_3$. As we proceed, we necessarily encounter a single simple pole at $z_2=qz_3$. The residue theorem implies that our initial integral equals the sum of (A) the integral where $z_2$ is along $\gamma_3$, and (B) the integral with only $z_1$ and $z_3$ remaining and integrand given by taking the residue of the initial integrand at $z_2=qz_3$.

Let us consider separately these two pieces. For (A) we now think of $z_2$ and $z_3$ as fixed along $\gamma_3$ and deform the $\gamma_1$ contour to $\gamma_3$, encountering two simple poles at $z_1=qz_2$ and $z_1=q z_3$. Thus (A) is expanded into a sum of three terms: the integral with $z_1, z_2$ and $z_3$ along $\gamma_3$; and the integral with only $z_2$ and $z_3$ remaining (along $\gamma_3$) and the residue taken at either $z_1=q z_2$ or $z_1=q z_3$.

For (B) we now think of $z_3$ as fixed along $\gamma_3$ and deform the $\gamma_1$ contour to $\gamma_3$, encountering a simple pole at $z_1=q^2 z_3$. This is because the residue of (\ref{cross}) at $z_2=q z_3$ equals
$$
\frac{z_1-z_3}{z_1-q^2 z_3}\,(qz_3-z_3).
$$
Thus (B) is expanded into a sum of two terms: the integral with $z_1$ and $z_3$ (along $\gamma_3$ and with the above expression in the integrand); and the integral with only $z_3$ remaining (along $\gamma_3$) and the residue of the above term taken at $z_1=q^2 z_3$.

Gathering the various terms in this expansion, we see that the residue subspaces we sum over are indexed by partitions of $k$ (here $k=3$) and take the form of geometric strings with the parameter $q$. For example, $\lambda = (1,1,1)$ corresponds to the term in the residue expansion in which all three variables $z_1,z_2$ and $z_3$ are still integrated, but along $\gamma_3$. On the other hand, $\lambda = (3)$ corresponds to the term in which the residue is taken at $z_1=qz_2=q^2z_3$ and the only variable which remains to be integrated along $\gamma_3$ is $z_3$. The partition $\lambda=(2,1)$ corresponds to the three remaining terms in the above expansion in which two integration variables remain. In general, $\ell(\lambda)$ corresponds to the number of variables which remain to be integrated in each term of the expansion.
\end{example}

Let us now turn to the general $k$ case. By the hypotheses of the theorem, the function $z_j\mapsto \prod_{1\leq A<B\leq k} (z_A-z_B) F(z_1,\ldots ,z_j,\ldots, z_k)$ has no poles which are encountered during contour deformations -- hence it plays no role in the residue analysis. As we deform sequentially the contours in the left-hand side of (\ref{qAK}) to $\gamma_k$ we find that the resulting terms in the residue expansion can be indexed by partitions $\lambda\vdash k$ along with a list (ordered set) of disjoint ordered subsets of $\{1,\ldots, k\}$ (whose union is all of $\{1,\ldots, k\}$)
\begin{eqnarray}\label{ijvar}
\nonumber &i_{1}< i_{2} < \cdots < i_{\lambda_1}&\\
&j_{1}< j_{2} < \cdots < j_{\lambda_2}&\\
\nonumber & \vdots&
\end{eqnarray}
Let us call such a list $I$. For a given partition $\lambda$ call $S(\lambda)$ the collection of all such lists $I$ corresponding to $\lambda$.
For $k=3$ and $\lambda=(2,1)$, in Example \ref{ex2} we saw there are three such lists which correspond with
$$S(\lambda) = \Big\{ \big\{1<2 , 3\big\}, \big\{1<3 ,2\big\}, \big\{2<3 , 1\big\}\Big\}.$$

For such a list $I$, we write $\Res{I} f(z_1,\ldots, z_k)$ as the residue of the function $f$ at
\begin{eqnarray*}
&z_{i_{1}}=qz_{i_{2}}, z_{i_2} = q z_{i_3}, \quad \ldots, \quad z_{i_{\lambda_1-1}}=q z_{i_{\lambda_1}}&\\
&z_{j_{1}}=qz_{j_{2}}, z_{j_2} = q z_{j_3}, \quad \ldots, \quad z_{j_{\lambda_2-1}}=q z_{j_{\lambda_2}}&\\
&\vdots&
\end{eqnarray*}
and regard the output as a function of the terminal variables $(z_{i_{\lambda_1}}, z_{j_{\lambda_2}}, \ldots)$. There are $\ell(\lambda)$ such strings and consequently that many remaining variables (though we have only written the first two strings above).

With the above notation in place, we may write the expansion of the integral on the left-hand side of (\ref{qAK}) as
\begin{eqnarray}\label{RHS2}
{\rm LHS} (\ref{qAK}) = \sum_{\substack{\lambda\vdash k\\ \lambda= 1^{m_1}2^{m_2}\cdots}} \frac{1}{m_1! m_2!\cdots} \sum_{I\in S(\lambda)}\oint_{\gamma_k} \frac{dz_{i_{\lambda_1}}}{2\pi \i} \oint_{\gamma_k} \frac{dz_{j_{\lambda_{2}}}}{2\pi \i} \cdots \\
\nonumber \times\,\Res{I} \left(\prod_{1\leq A<B\leq k} \frac{z_A-z_B}{z_A-q z_B}  F(z_1,\ldots, z_k) \right).
\end{eqnarray}
The factor of $\frac{1}{m_1! m_2!\cdots}$ arose from multiple counting of terms in the residue expansion due to symmetries of $\lambda$. For example, for the partition $\lambda=(2,2,1)$, each $I\in S(\lambda)$, corresponds uniquely with a different $I'\in S(\lambda)$ in which the $i$ and $j$ variables in (\ref{ijvar}) are switched. Since these correspond with the same term in the residue expansion, this constitutes double counting and hence the sum should be divided by $2!$. The reason why our residual subspace expansion only corresponds with strings is because if we took a residue which was not of the form of a string, then for some $A\neq A'$ we would be evaluating the residue at $z_A=qz_B$ and $z_{A'}=qz_B$. However, the Vandermonde determinant in the numerator of our integrand would then necessarily evaluate to zero. Therefore, such possible non-string residues (coming from the denominator) in fact have zero contribution.

\noindent {\bf Step 2:} For each $I\in S(\lambda)$ relabel the $z$ variables as
\begin{eqnarray*}
&(z_{i_{1}}, z_{i_{2}}, \ldots, z_{i_{\lambda_1}}) \mapsto (y_{\lambda_1},y_{\lambda_1-1},\ldots, y_1)&\\
&(z_{j_{1}}, z_{j_{2}}, \ldots, z_{j_{\lambda_2}}) \mapsto (y_{\lambda_1+\lambda_2}, \ldots, y_{\lambda_1+1})&\\
&\vdots&
\end{eqnarray*}
and observe that there exists a unique permutation $\sigma\in S_k$ for which $(z_1,\ldots, z_k) = (y_{\sigma(1)},\ldots, y_{\sigma(k)})$. Let us also call $w_j = y_{\lambda_1+\cdots \lambda_{j-1}+1}$, for $1\leq j\leq \ell(\lambda)$.

We introduce a canonical form for taking strings of residues. For a function $f(y_1,\ldots,y_k)$ define $\Resq{\lambda}f(y_1,\ldots, y_k)$ to be the residue of $f(y_1,\ldots, y_k)$ at
\begin{eqnarray}\label{resqvalues}
\nonumber &y_{\lambda_{1}}=qy_{\lambda_{1}-1},\quad y_{\lambda_{1}-1}=qy_{\lambda_{1}-2}, \quad \ldots, \quad y_{2}=qy_{1}\\
&y_{\lambda_{1}+\lambda_2}=qy_{\lambda_1+\lambda_2-1},\quad  y_{\lambda_{1}+\lambda_2-1}=qy_{\lambda_1+\lambda_2-2}, \quad \ldots, \quad y_{\lambda_{1}+2}=qy_{\lambda_1+1}\\
\nonumber &\vdots &
\end{eqnarray}
with the output regarded as a function of the terminal variables $\big(y_1,y_{\lambda_1+1},\ldots, y_{\lambda_1+\cdots+\lambda_{\ell(\lambda)-1}-1}\big)$. We call each sequence of identifications of variables a {\it string}. As all poles which we encounter in what follows are simple, the above residue evaluation amounts to
\begin{equation}\label{amts}
\Resq{\lambda}f(y_1,\ldots, y_k) = \lim_{\substack{y_2\to q y_1\\ \cdots\\y_{\lambda_1}\to q y_{\lambda_1-1}}}\prod_{i=2}^{\lambda_1}(y_i-qy_{i-1}) \quad \lim_{\substack{y_{\lambda_1+2}\to q y_{\lambda_1+1}\\ \cdots \\ y_{\lambda_1+\lambda_2}\to q y_{\lambda_1+\lambda_2-1}}}\prod_{i=\lambda_1+2}^{\lambda_1+\lambda_2}(y_i-qy_{i-1})\quad \cdots \,\, f(y_1,\ldots, y_k).
\end{equation}
It is also convenient to define $\Subq{\lambda}f(y_1,\ldots, y_k)$ as the function of $\big(y_1,y_{\lambda_1+1},\ldots, y_{\lambda_1+\cdots+\lambda_{\ell(\lambda)-1}-1}\big)$, which is the result of substituting the relations of (\ref{resqvalues}) into $f(y_1,\ldots, y_k)$.

With the above notation we rewrite the term corresponding to a partition $\lambda$ in the right-hand side of (\ref{RHS2}) as
\begin{equation}\label{above3}
 \frac{1}{m_1! m_2!\cdots} \sum_{\sigma\in S_k} \oint_{\gamma_k} \frac{dw_1}{2\pi \i} \cdots \oint_{\gamma_k} \frac{dw_{\ell(\lambda)}}{2\pi \i}
\Resq{\lambda} \left(\prod_{1\leq A<B\leq k} \frac{y_{\sigma(A)}-y_{\sigma(B)}}{y_{\sigma(A)}-q y_{\sigma(B)}} F(\sigma(\vec{y})) \right),
\end{equation}
where $\Resq{\lambda}$ is defined in equation (\ref{resqvalues}). Note that the output of the residue operation is a function of the variables $(y_1,y_{\lambda_1+1},\ldots) = (w_1,w_2,\ldots)$.

One should observe that the above expression includes the summation over all $\sigma\in S_k$, and not just those which arise from an $I\in S(\lambda)$ as above. This, however, is explained by the fact that if $\sigma$ does not arise from some $I\in S(\lambda)$, then the residue necessarily evaluates to zero (hence adding these terms is allowed). To see this, observe that in the renumbering of variables discussed above, only permutations with
$$\sigma^{-1}(1)>\sigma^{-1}(2)\cdots>\sigma^{-1}(\lambda_1), \qquad \sigma^{-1}(\lambda_1+1)>\sigma^{-1}(\lambda_1+2)\cdots>\sigma^{-1}(\lambda_1+\lambda_2),\qquad \ldots$$
participated. Any other $\sigma$ must violate one of these strings of conditions. Consider, for example, some $\sigma$ with $\sigma(\lambda_1-1) < \sigma(\lambda_1)$. This implies that the term $y_{\lambda_1-1}-qy_{\lambda_1}$ shows up in the denominator of (\ref{above3}), as opposed to the term $y_{\lambda_1}-qy_{\lambda_1-1}$. Residues can be taken in any order, and if we first take the residue at $y_{\lambda_1} = qy_{\lambda_1-1}$, we find that the above denominator does not have a pole (nor do any other parts of (\ref{above3})) and hence the residue is zero. Similar reasoning works in general.

\noindent {\bf Step 3:} All that remains is to compute the residues in (\ref{above3}) and identify the result (after summing over all $\lambda\vdash k$) with the right-hand side of (\ref{qAK}) as necessary to prove the proposition.

It is convenient to rewrite the product over $A<B$ as an $S_k$ invariant function, times a function that is analytic at the points in which the residue is being taken:
$$
\prod_{1\leq A<B\leq k} \frac{y_{\sigma(A)}-y_{\sigma(B)}}{y_{\sigma(A)}-q y_{\sigma(B)}}  = \prod_{1\leq A\neq B\leq k} \frac{y_A-y_B}{y_A-qy_B} \prod_{1\leq B<A\leq k} \frac{y_{\sigma(A)} - qy_{\sigma(B)}}{y_{\sigma(A)} - y_{\sigma(B)}}.
$$
Since it is only the $S_k$ invariant function that contains the poles with which we are concerned, it allows us to rewrite (\ref{above3}) as
\begin{eqnarray*}
\frac{1}{m_1! m_2!\cdots} \oint \frac{w_1}{2\pi \i} \cdots \oint \frac{w_{\ell(\lambda)}}{2\pi \i} \Resq{\lambda} \left( \prod_{1\leq A\neq B\leq k} \frac{y_A-y_B}{y_A-qy_B} \right)\\
\nonumber \times\,\Subq{\lambda} \left( \sum_{\sigma\in S_k} \prod_{1\leq B<A\leq k} \frac{y_{\sigma(A)} - qy_{\sigma(B)}}{y_{\sigma(A)} - y_{\sigma(B)}} F(\sigma(\vec{y})) \right).
\end{eqnarray*}
We use Lemma \ref{reslemma} to evaluate the above residue, and we easily identify the substitution on the second line with $E^q(w \circ \lambda)$ as in the statement of the proposition.
Combining these two expressions and summing the resulting expression over $\lambda\vdash k$ we arrive at the desired residue expansion claimed in the statement of the proposition.
\end{proof}

\begin{lemma}\label{reslemma}
For all $k\geq 1$, $\lambda\vdash k$ and $q\in (0,1)$, we have that
\begin{equation}\label{reslemmaeqn}
 \Resq{\lambda} \left(\prod_{1\leq i\neq j\leq k} \frac{y_i-y_j}{y_i-qy_j} \right) =(-1)^k (1-q)^k q^{-\frac{k^2}{2}} \prod_{j=1}^{\ell(\lambda)} w_j^{\lambda_j} q^{\frac{\lambda_j^2}{2}} \det\left[\frac{1}{w_i q^{\lambda_i}-w_j}\right]_{i,j=1}^{\ell(\lambda)},
\end{equation}
where we have renamed the variables remaining after the residue operator as $w_j = y_{\lambda_1+\cdots \lambda_{j-1}+1}$, for $1\leq j\leq \ell(\lambda)$.
\end{lemma}

\begin{proof} The product on the left-hand side of (\ref{reslemmaeqn}) involves terms in which $i$ and $j$ are in the same string of variables in (\ref{resqvalues}) as well as terms in which they are in different strings. We need to compute the residue of the same string terms and multiply it by the substitution of variables into the different string terms.

Let us first evaluate same string residues. Consider variables $y_1,\ldots, y_{\ell}$, $\ell\geq 2$, and observe that
\begin{equation*}
\Res{\substack{y_2=qy_1\\y_3=qy_2\\\cdots\\y_{\ell}= qy_{\ell-1}}}  \left( \prod_{1\leq i\neq j\leq \ell} \frac{y_i-y_j}{y_i-qy_j}\right) = (-1)^{\ell-1} y_1^{\ell-1} \frac{(1-q)^{\ell}}{1-q^{\ell}}.
\end{equation*}

Now turn to the cross term between two strings of variables. Consider one set of variables  $y_1,\ldots, y_{\ell}$ with $\ell\geq 2$ and a second set of variables $y'_1\ldots, y'_{\ell'}$ with $\ell'\geq 2$. Then
\begin{equation*}
\Sub{\substack{y_2=qy_1\\y_3=qy_2\\\cdots\\y_{\ell}= qy_{\ell-1}}}  \Sub{\substack{y'_2=qy'_1\\y'_3=qy'_2\\\cdots\\y'_{\ell'}= qy'_{\ell'-1}}} \left( \prod_{i=1}^{\ell}\prod_{j=1}^{\ell'} \frac{y_i-y_j}{y_i-qy_j}\right)= \prod_{i=1}^{\ell} \frac{y_1 q^{i-1}-y'_1}{y q^{i-1} - y'_1 q^{\ell'}}.
\end{equation*}

Since the strings of variables also come interchanged, we should multiply the above expression by the same term with $(y_1,\ell)$ and $(y'_1,\ell')$ interchanged.

Returning to the statement of the lemma, we see that we can evaluate the desired residue by multiplying the same string terms over all strings in (\ref{resqvalues}) as well as multiplying all terms corresponding to pairs of different strings. Using the above calculations we obtain
\begin{equation*}
\Resq{\lambda} \left(\prod_{1\leq i\neq j\leq k} \frac{y_i-y_j}{y_i-qy_j} \right) = \prod_{j=1}^{\ell(\lambda)} \frac{(1-q)^{\lambda_j}}{(1-q^{\lambda_j})} (-1)^{\lambda_j-1} w_{j}^{\lambda_j-1} \, \prod_{1\leq i<j\leq \ell(\lambda)} q^{-\lambda_i\lambda_j} \frac{(w_i-w_j)(w_i q^{\lambda_i} - w_j q^{\lambda_j})}{(w_i q^{\lambda_i} - w_j)(w_i -w_j q^{\lambda_j})}.\qquad\qquad
\end{equation*}
It is easy now to rewrite (using the Cauchy determinant) the above expression so as to produce the equality of the lemma, as desired.
\end{proof}

%

\subsection{Inductive proof of Lemma \ref{belowlemma}}\label{inductiveproof} 
The proof we give in the main body of the text for Lemma \ref{belowlemma} relied upon Theorem \ref{KqBosonIdDual} (the dual Plancherel theorem). We record here an inductive proof of this same result.

The desired identity (\ref{halflem}) can be rewritten as
\begin{equation}\label{halflemrewrite}
\sum_{n_1\geq \cdots \geq n_k\geq 0} \psir_{\vec{z}}(\vec{n}) \prod_{j=1}^{k} \left(1-\alpha/q^j\right)^{-n_j} = (-1)^k q^{\frac{k(k-1)}{2}} \prod_{j=1}^{k} \frac{1-\alpha/q^j}{z_j-\alpha/q}.
\end{equation}

From the definition of $\psir$ we can rewrite this identity as
\begin{equation}\label{Ik}
\sum_{n_1\geq \cdots \geq n_k\geq 0}\, \prod_{i=1}^{M} \frac{1}{(c_i)!_q} \sum_{\sigma\in S_k} \prod_{1\leq B<A\leq k} \frac{z_{\sigma(A)}-q^{-1}z_{\sigma(B)}}{z_{\sigma(A)}-z_{\sigma(B)}} \prod_{j=1}^{k} \left(\frac{1-z_{\sigma(j)}}{1-\alpha/q^j}\right)^{n_j}  = \prod_{j=1}^k \frac{1-\alpha/q^j}{z_j - \alpha/q}
\end{equation}
where we recall $M=M(\vec{n})$ and $(c_1,\ldots,c_M)=\vec{c}(\vec{n})$ from Section \ref{notations}. Define $I_k$ to be the left-hand side of the above desired equality.

We will proceed by induction in $k$. It is straightforward to check that the identity holds for $k=1$, and for later use define $I_0=1$. We may write
\begin{eqnarray*}
I_k &=& \sum_{n_k\geq 0}\, \sum_{I,J}\, \sum_{\sigma_I}\, \sum_{\sigma_J} \sum_{n_1\geq \cdots \geq n_{k-m} >n_k} \frac{1}{(m)!_q}\, \prod_{i=1}^{M'} \frac{1}{(c_i')!_q}\\
&&\times \prod_{\substack{B<A\\A,B\in I}} \frac{z_{\sigma_I(i_A)}-q^{-1}z_{\sigma_I(i_B)}}{z_{\sigma_I(i_A)}-z_{\sigma_I(i_B)}}
\prod_{\substack{B<A\\A,B\in J}} \frac{z_{\sigma_J(j_A)}-q^{-1}z_{\sigma_J(j_B)}}{z_{\sigma_J(j_A)}-z_{\sigma_J(j_B)}}
\prod_{\substack{i\in I\\j\in J}} \frac{z_i-q^{-1}z_j}{z_i-z_j} \\
&& \times \prod_{\ell=k-m+1}^{k} \left(\frac{1}{1-\alpha/q^{\ell}}\right)^{n_k} \prod_{i\in I} (1-z_i)^{n_k} \prod_{\ell=1}^{k-m} \left(\frac{1-z_{\sigma_J(j_{\ell})}}{1-\alpha/q^{\ell}}\right)^{n_{\ell}}.
\end{eqnarray*}
The above formula involves some notation which should be explained. The term $\sum_{I,J}$  is the summation over all subsets $I=\{i_1<\cdots <i_m\}$ and $J=\{j_1<\cdots <j_{k-m}\}$ such that $m=|I|\geq 1$, $I\cup J = \{1,\ldots, k\}$ and $I\cap J = \emptyset$. The term $\sum_{\sigma_I}$ is the summation over all permutations $\sigma_I$ which permute the elements of $I$ and fix the elements of $J$, and, conversely, the term $\sum_{\sigma_J}$ is the summation over all permutations $\sigma_J$ which permute the elements of $J$ and fix the elements of $I$. The final summation is over $n_1\geq \cdots \geq n_{k-m}>n_k$ with $n_k$ fixed. We call $\vec{n}' = (n_1,\ldots, n_{k-m})$ and set $M' = M(\vec{n}')$ and $(c'_1,\ldots, c'_{M'})=\vec{c}(\vec{n}')$.

Notice that the summation over $\sigma_I$ can be readily computed to be $(m)!_{q^{-1}}$ using (\ref{mqinverse}). Note that when $(m)!_{q^{-1}}$ is divided by $(m)!_q$ this yields $q^{-\frac{m(m-1)}{2}}$. Further rearrangement of terms leads to
\begin{eqnarray}\label{almostt}
I_k &=& \sum_{n_k\geq 0} \,\sum_{I,J}  \, q^{-\frac{k(k-1)}{2}} \prod_{i\in I} (1-z_i)^{n_k} \prod_{\ell=k-m+1}^{k} \left(\frac{1}{1-\alpha/q^{\ell}}\right)^{n_k} \prod_{\substack{i\in I\\j\in J}} \frac{z_i-q^{-1}z_j}{z_i-z_j} \\
&& \times \sum_{n_1\geq \cdots \geq n_{k-m} >n_k} \, \prod_{i=1}^{M'} \frac{1}{(c_i')!_q}\, \sum_{\sigma_J} \prod_{\substack{B<A\\A,B\in J}} \frac{z_{\sigma_J(j_A)}-q^{-1}z_{\sigma_J(j_B)}}{z_{\sigma_J(j_A)}-z_{\sigma_J(j_B)}} \prod_{\ell=1}^{k-m} \left(\frac{1-z_{\sigma_J(j_\ell)}}{1-\alpha/q}\right)^{n_{\ell}}.
\end{eqnarray}

From the definition of $I_{k-m}$ we can rewrite the second line as
$$
I_{k-m}(z_{j_1},\ldots, z_{j_{k-m}}) \prod_{\ell=1}^{k-m} \left(\frac{1-z_{j_{\ell}}}{1-\alpha/q^j}\right)^{n_k+1}.
$$

By the inductive hypothesis $I_{k-m}$ is given by the right-hand side of (\ref{Ik}). Let us make this substitution into (\ref{almostt}):
\begin{equation*}
I_k = \sum_{n_k\geq 0}\, \prod_{\ell=1}^{k} \left(\frac{1-z_{\ell}}{1-\alpha/q^{\ell}}\right)^{n_k} \sum_{I,J}q^{-\frac{m(m-1)}{2}} \prod_{\substack{i\in I\\j\in J}} \frac{z_i-q^{-1}z_j}{z_i-z_j} \prod_{j\in J} \frac{1-z_j}{z_j-\alpha/q}.
\end{equation*}
The summation in $n_k$ can now be easily computed.

In order to prove the inductive step for $k$ it suffices to check that the above expression equals the right-hand side of (\ref{Ik}). This reduces the proof of the lemma to showing the following equality:
$$
\frac{1-\prod_{\ell=1}^{k} \frac{1-z_{\ell}}{1-\alpha/q^{\ell}}}{\prod_{\ell=1}^k \frac{z_\ell - \alpha/q}{1-\alpha/q^{\ell}}} = \sum_{I,J} q^{-\frac{m(m-1)}{2}} \prod_{\substack{i\in I\\j\in J}} \frac{z_i-q^{-1}z_j}{z_i-z_j} \prod_{j\in J} \frac{1-z_j}{z_j-\alpha/q}.
$$
The summation on the left-hand side does not include the term corresponding to $I=\emptyset$. Including this term cancels a portion of the left-hand side and further cross multiplying brings the desired equality to:
\begin{equation}\label{desireeqn}
\prod_{\ell=1}^{k} \left(1-\frac{\alpha}{q^\ell}\right) =  \sum_{I,J} q^{-\frac{k(k-1)}{2}} \prod_{\substack{i\in I\\j\in J}} \frac{z_i-q^{-1}z_j}{z_i-z_j} \prod_{i\in I} \left(z_i -\frac{\alpha}{q}\right) \prod_{j\in J} (1-z_j)
\end{equation}
where the term $\sum_{I,J}$ (without the $*$) represents the summation over all $I$ and $J$ such that $I\cup J =\{1,\ldots,k\}$ and $I\cap J = \emptyset$.

This desired equality is a $q$-deformation of the binomial expansion of $(1-\alpha)^k$. To prove it, let $V(\vec{z}) = \prod_{i<j} (z_j-z_i)$ be the Vandermonde determinant and let $T_{\ell}$ be the $q$-shift operator defined as $(T_\ell f)(z_1,\ldots, z_k) = (z_1,\ldots, q^{-1} z_{\ell},\ldots, z_k)$. Making the change of variables $z_i\mapsto 1/z_i$ we can rewrite the desired equality (\ref{desireeqn}) as
$$
\prod_{\ell=1}^{k} \left(\frac{T_\ell -1}{z_\ell} + (1-\frac{\alpha}{q} T_\ell )\right)  V(z)  = \prod_{j=1}^{k} (1-\frac{\alpha}{q^j}) V(z).
$$
This eigenfunction relationship for $V(z)$ is easily proved by writing $V(z) = \det\left[z_i^{j-1}\right]_{i,j=1}^{k}$ and observing that
$$
 \left(\frac{T_\ell -1}{z_\ell} + (1-\frac{\alpha}{q} T_\ell)\right)  V(z)  = \det\left[ (q^{-j+1} -1)z_i^{j-2} + (1-\alpha/q^j) z_i^{j-1}\right]_{i,j=1}^{k}.
$$
By elementary column transformations the above matrix can be brought to the form
$$
 \det\left[(1-\alpha/q^j) z_i^{j-1}\right]_{i,j=1}^{k} = \prod_{j=1}^{k} (1-\alpha/q^j) V(z),
$$
thus proving the desired equality (\ref{desireeqn}) and ultimately the inductive step and the lemma.


\begin{thebibliography}{alpha}

\bibitem{AKK}
M.~Alimohammadi, V.~Karimipour, M.~Khorrami.
\newblock Exact solution of a one-parameter family of asymmetric exclusion processes.
\newblock {\it Phys. Rev. E.}, {\bf 57}:6370--6376, 1998

\bibitem{ACQ}
G.~Amir, I.~Corwin, J.~Quastel.
\newblock Probability distribution of the free energy of the continuum directed random polymer in $1+1$ dimensions.
\newblock {\it Commun. Pure Appl. Math.},{\bf 64}:466--537, 2011.

\bibitem{BabThom}
D.~Babbitt, L.~Thomas.
\newblock Ground state representation of the infinite one-dimensional Heisenberg ferromagnet. II. An explicit Plancherel formula.
\newblock {\it Comm. Math. Phys.} {\bf 54}:255--278, 1977.

\bibitem{BabGut}
D.~Babbitt, E.~Gutkin.
\newblock The plancherel formula for the infinite XXZ Heisenberg spin chain.
\newblock {\it Lett. Math. Phys.} {\bf 20}:91--99, 1990.

\bibitem{BertiniCancrini}
L.~Bertini, N.~Cancrini.
\newblock  The stochastic heat equation: Feynman-Kac formula and intermittence.
\newblock {\em J. Stat. Phys.}, {\bf 78}:1377--1401, 1995.

\bibitem{Bethe}
H.~Bethe.
\newblock Zur Theorie der Metalle. I. Eigenwerte und Eigenfunktionen der linearen Atomkette. (On the theory of metals. I. Eigenvalues and eigenfunctions of the linear atom chain)
\newblock {\it Zeitschrift fur Physik}, {\bf 71}:205--226, 1931.

\bibitem{BBT}
N.~M.~ Bogoliubov, R.~K.~Bullough, J.~Timonen.
\newblock Critical behavior for correlated strongly coupled Boson systems in $1+1$ dimensions.
\newblock {\it Phys. Rev. Lett.} {\bf 25}:3933-3936, 1994.

\bibitem{BIK}
N.~M.~ Bogoliubov, A.~G.~Izergin, N.~A.Kitanine.
\newblock Correlation functions for a strongly correlated Boson system.
\newblock {\it Nucl. Phys. B} {\bf 516}:501--528, 1998.

\bibitem{twosides}
A.~Borodin
\newblock Schur dynamics of the Schur processes.
\newblock {\it Adv. Math.} {\bf 228}:2268--2291, 2011.

\bibitem{BorCor}
A.~Borodin, I.~Corwin.
\newblock Macdonald processes.
\newblock {\it Probab. Theor. Rel. Fields}, to appear. arXiv:1111.4408.

\bibitem{BCdiscrete}
A.~Borodin, I.~Corwin.
\newblock Discrete time $q$-TASEPs.
\newblock arXiv:1305.2972.

\bibitem{BCF}
A.~Borodin, I.~Corwin, P.~L.~Ferrari.
\newblock Free energy fluctuations for directed polymers in random media in $1+1$ dimension.
\newblock {\it Comm. Pure Appl. Math.}, to appear. arXiv:1204.1024.

\bibitem{BCFV}
A.~Borodin, I.~Corwin, P.~L.~Ferrari, B.~Vet\H{o}.
\newblock Height fluctuations for the stationary KPZ equation.
\newblock In preparation.


\bibitem{BCGS}
A.~Borodin, I.~Corwin, V.~Gorin, S.~Shakirov.
\newblock Observables of Macdonald processes.
\newblock arXiv:1306.0659.

\bibitem{BCPS2}
A.~Borodin, I.~Corwin, L.~Petrov, T.~Sasamoto.
\newblock In preparation.

\bibitem{BCR}
A.~Borodin, I.~Corwin, D.~Remenik.
\newblock Log-Gamma polymer free energy fluctuations via a Fredholm determinant identity.
\newblock {\it Commun. Math. Phys.}, online first.

\bibitem{BCS}
A.~Borodin, I.~Corwin, T.~Sasamoto.
\newblock From duality to determinants for $q$-TASEP and ASEP.
\newblock {\it Ann. Probab.}, to appear. arXiv:1207.5035.

\bibitem{BF}
A.~Borodin, P.L. Ferrari.
\newblock Anisotropic growth of random surfaces in $2+1$ dimensions.
\newblock {\it Commun. Math. Phys.}, to appear. arXiv:0804.3035.

\bibitem{BorGor}
A.~Borodin, V.~Gorin.
\newblock Lectures on integrable probability.
\newblock  arXiv:1212.3351.

\bibitem{BorPet}
A.~Borodin, L.~Petrov.
\newblock Nearest neighbor Markov dynamics on Macdonald processes.
\newblock arXiv:1305.5501.

\bibitem{CalCaux}
P.~Calabrese, J.~S.~Caux.
\newblock Dynamics of the attractive 1D Bose gas: analytical treatment from integrability.
\newblock {\it J. Stat. Mech.}, P08032, 2007.


\bibitem{CDR}
P.~Calabrese, P.~Le Doussal, A.~Rosso.
\newblock Free-energy distribution of the directed polymer at high temperature.
\newblock {\it Euro. Phys. Lett.}, {\bf 90}:20002, 2010.

\bibitem{CDlong}
P.~Calabrese, P.~Le Doussal.
\newblock The KPZ equation with flat initial condition and the directed polymer with one free end.
\newblock {\it J. Stat. Mech.} P06001, 2012.

\bibitem{CDprl}
P.~Calabrese, P.~Le Doussal.
\newblock An exact solution for the KPZ equation with flat initial conditions.
\newblock {\it Phys.Rev.Lett.} {\bf 106}:250603, 2011.

\bibitem{ICReview}
I. Corwin.
\newblock The {K}ardar-{P}arisi-{Z}hang equation and universality class.
\newblock {\em Random Matrices Theory Appl.}, {\bf 1}, 2012.


\bibitem{COSZ}
I.~Corwin, N.~O'Connell, T.~Sepp\"{a}l\"{a}inen, N.~Zygouras.
\newblock Tropical combinatorics and Whittaker functions.
\newblock {\it Duke. Math. J}, to appear. arXiv:1110.3489.

\bibitem{CorPet}
I.~Corwin, L.~Petrov.
\newblock The $q$-PushASEP: A new integrable traffic model in $1+1$ dimension.
\newblock arXiv:1308.3124.

\bibitem{CQ}
I.~Corwin, J.~Quastel.
\newblock Crossover distributions at the edge of the rarefaction fan.
\newblock {\em Ann. Probab.}, {\bf 41}:1243--1314, 2013.

\bibitem{CorwinQuastel}
I.~Corwin, J.~Quastel.
\newblock The renormalization fixed point of the Kardar-Parisi-Zhang universality class.
\newblock arXiv:1103.3422.


\bibitem{DiaconisFill}
P.~Diaconis, J.~A.~Fill.
\newblock Strong stationary times via a new form of duality.
\newblock {\it Ann. Probab.}, {\bf 18}:1483--1522 (1990).


\bibitem{Dot}
V.~Dotsenko.
\newblock Bethe ansatz derivation of the Tracy-Widom distribution for one-dimensional directed polymers.
\newblock {\it Euro. Phys. Lett.}, {\bf 90}:20003, 2010.


\bibitem{Dot2}
V.~Dotsenko.
\newblock Replica Bethe ansatz derivation of the GOE Tracy-Widom distribution in one-dimensional directed polymers with free boundary conditions.
\newblock {\it J. Stat. Mech.}, P11014, 2012.

\bibitem{Dot3}
V.~Dotsenko.
\newblock Distribution function of the endpoint fluctuations of one-dimensional directed polymers in a random potential.
\newblock  arXiv:1209.6166.

\bibitem{Dot4}
V.~Dotsenko.
\newblock Two-time free energy distribution function in $(1+1)$ directed polymers.
\newblock arXiv:1304.0626.

\bibitem{Dot5}
V.~Dotsenko.
\newblock Two-point free energy distribution function in $(1+1)$ directed polymers.
\newblock arXiv:1304.6571.

\bibitem{GJL}
D.~Gabrielli, G.~Jona-Lasinio, C.~Landim.
\newblock Onsager symmetry from microscopic TP invariance.
\newblock {\it J. Stat. Phys.}, {\bf 96}:639--652, 1999.

\bibitem{GV}
R.~Gangolli, V.~S.~Varadarajan.
\newblock {\it Harmonic analysis of spherical functions on real reductive spaces.}
\newblock Ergebnisse der Mathematik 101, Springer Verlag, 1988.

\bibitem{Gaudin}
M.~Gaudin.
\newblock {\it La fonction d'onde de Bethe}
\newblock Masson, Paris, 1983.

\bibitem{Gaudin2}
M.~Gaudin.
\newblock Boundary energy of a Bose gas in one dimension.
\newblock {\it Phys. Rev. A}, {\bf 4}:386--394, 1971.

\bibitem{GLO}
A.~Gerasimov, D.~Lebedev, S.~Oblezin.
\newblock On q-deformed $\mathfrak{gl}_{\ell+1}$-Whittaker function I
\newblock {\it Commun. Math. Phys.}, {\bf 294}97--119 (2010).


\bibitem{GLOqlim}
A.~Gerasimov, D.~Lebedev, S.~Oblezin.
\newblock On a classical limit of q-deformed Whittaker functions.
\newblock {\it Lett. Math. Phys.}, {\bf 100}:279--290, 2012.

\bibitem{Gut}
E.~Gutkin.
\newblock Heisenberg-Ising spin chain: Plancherel decomposition and Chebyshev polynomials.
\newblock In {\it Calogero-Moser-Sutherland Models. CRM Series in Mathematical Physics}. 177--192, 2000.


\bibitem{HO}
G.~J.~Heckman, E.~M.~Opdam.
\newblock Yang's system of particles and Hecke algebras.
\newblock {\it Ann. Math.}, {\bf 145}:139--173, 1997.

\bibitem{Helgason66b}
S.~Helgason.
\newblock An analogue of the Paley-Wiener theorem for the Fourier transform on certain symmetric spaces.
\newblock {\it Math. Ann.} {\bf 165}:297--308, 1966.

\bibitem{Helbook}
S.~Helgason.
\newblock {\it Groups and geometric analysis.}
\newblock Academic press, London, 1984.

\bibitem{ImSa}
T.~Imamura, T.~Sasamoto.
\newblock  Replica approach to the KPZ equation with half Brownian motion initial condition.
\newblock {\it J. Phys. A: Math. Theor.} {\bf 44}:385001, 2011.

\bibitem{ImSaKPZ}
T.~Imamura, T.~Sasamoto.
\newblock  Exact solution for the stationary Kardar-Parisi-Zhang equation.
\newblock {\it Phys. Rev. Lett.}, {\bf 108}:190603, 2012.

\bibitem{ImSaKPZ2}
T.~Imamura, T.~Sasamoto.
\newblock  Stationary correlations for the 1D KPZ equation.
\newblock {\it J. Stat. Phys.}, {\bf 150}:908--939, 2013.

\bibitem{ISS}
T.~Imamura, T.~Sasamoto, H.~Spohn.
\newblock On the equal time two-point distribution of the one-dimensional KPZ equation by replica.
\newblock arXiv:1305.1217.


\bibitem{K}
M.~Kardar.
\newblock Replica-Bethe Ansatz studies of two-dimensional interfaces with quenched random impurities.
\newblock {\it Nucl. Phys. B}, {\bf 290}:582--602 (1987).

\bibitem{Kim}
D.~Kim.
\newblock Bethe ansatz solution for crossover scaling functions of the asymmetric XXZ chain and the Kardar-Parisi-Zhang-type growth model.
\newblock {\it Phys. Rev. E} {\bf 52}:3512, 1995.

\bibitem{Korepin}
V.~E.~Korepin.
\newblock Calculation of norms of Bethe wave functions.
\newblock {\it Commun. Math. Phys.} {\bf 86}:391, 1982.

\bibitem{Korff}
C.~Korff.
\newblock Cylindric versions of specialised Macdonald functions and a deformed Verlinde algebra.
\newblock {\it Commun. Math. Phys.}, {\bf 318}:173--246, 2013.

\bibitem{Hyun}
M.~Korhonen, E.~Lee.
\newblock The transition probability and the probability of the left-most particle's position of the $q$-TAZRP.
\newblock arXiv:1308.4769.

\bibitem{quench}
M.~Kormos, A.~Shashi, Y-Z.~Chou, J-S.~Caux, A.~Imambekov.
\newblock Interaction quenches in the 1D Bose gas
\newblock arXiv:1305.7202.

\bibitem{Lee}
E.~Lee.
\newblock Transition probabilities of the Bethe Ansatz solvable interacting particle systems.
\newblock {\it J. Stat. Phys.} {\bf 142}:643--656, 2011.

\bibitem{LL}
E.H.~Lieb, W.~Liniger.
\newblock Exact Analysis of an Interacting Bose Gas. I. The General Solution and the Ground State.
\newblock {\it Phys. Rev. Lett.}, {\bf 130}:1605--1616, 1963.


\bibitem{Lig}
T.~Liggett.
\newblock {\it Interacting particle systems}.
\newblock Spinger-Verlag, Berlin, 2005.

\bibitem{M}
I.G.~Macdonald.
\newblock {\it Symmetric Functions and Hall Polynomials.}
\newblock 2nd ed. Oxford University Press, New York. 1999.

\bibitem{M1}
I.G.~Macdonald.
\newblock Spherical Functions of $p$-adic Type.
\newblock {\it Publ. of the Ramanujan Inst.}, No. 2, 1971

\bibitem{M3}
I.G.~Macdonald.
\newblock Orthogonal polynomials associated with root systems,
\newblock {\it Sem. Lothar. Combin.} {\bf 45}:B45a, 2000.

\bibitem{McGuire}
J.~B.~McGuire.
\newblock Study of exactly soluble one-dimensional N-body problems.
\newblock {\it J. Math. Phys.}, {\bf 5}:622, 1964.


\bibitem{QMR}
G.~Moreno~Flores, J.~Quastel, D.~Remenik.
\newblock In preparation.

\bibitem{OCon}
N.~O'Connell.
\newblock Directed polymers and the quantum Toda lattice.
\newblock  {\em Ann. Probab.}, {\bf 40}:437--458, 2012.

\bibitem{OConPei}
N.~O'Connell, Y.~Pei.
\newblock A $q$-weighted version of the Robinson-Schensted algorithm.
\newblock arXiv:1212.6716.

\bibitem{OY}
N.~O'Connell, M.~Yor.
\newblock Brownian analogues of Burke's theorem.
\newblock {\it Stoch. Proc. Appl.}, {\bf 96}:285--304, 2001.

\bibitem{Ok}
A.~Okounkov.
\newblock Infinite wedge and random partitions.
\newblock {\it Selecta Math.} {\bf 7}:57--81 (2001).

\bibitem{OkResh}
A.~Okounkov, N.~Reshetikhin.
\newblock Correlation function of Schur process with application to local geometry of a random 3-dimensional Young diagram.
\newblock {\it J. Amer. Math. Soc.}, {\bf 16}:581--603 (2003).

\bibitem{Oxford}
S.~Oxford.
\newblock {\it The Hamiltonian of the quantized nonlinear Schr\"{o}dinger equation}.
\newblock Ph.D. thesis, UCLA, 1979.


\bibitem{Pov1}
A.~M.~Povolotsky.
\newblock Bethe ansatz solution of zero-range process with non-uniform stationary state.
\newblock {\it Phys. Rev. E} {\bf 69}:061109, 2004.

\bibitem{Pov2}
A.M.~Povolotsky.
\newblock On integrability of zero-range chipping models with factorized steady state.
\newblock arXiv:1308.3250.

\bibitem{PPH}
A.M.~Povolotsky, V.~B.~Priezzhev, Chin-Kun~Hu.
\newblock The asymmetric avalanche process.
\newblock {\it J. Stat. Phys.}, {\bf 111}:1149--1182, 2003.  


\bibitem{ProS1}
S.~Prolhac, H.~Spohn.
\newblock Two-point generating function of the free energy for a directed polymer in a random medium.
\newblock {\it J. Stat. Mech.} P01031, 2011.

\bibitem{ProS2}
S.~Prolhac, H.~Spohn.
\newblock The one-dimensional KPZ equation and the Airy process.
\newblock {\it J. Stat. Mech.} P03020, 2011.

\bibitem{ProSpoComp}
S.~Prolhac, H.~Spohn.
\newblock The propagator of the attractive delta-Bose gas in one dimension.
\newblock {\it J. Math. Phys.}, {\bf 52}:122106, 2011.


\bibitem{SaSp}
T.~Sasamoto, H.~Spohn.
\newblock One-dimensional KPZ equation: an exact solution and its universality.
\newblock {\em Phys. Rev. Lett.}, {\bf 104}:23, 2010.

\bibitem{SasWad}
T.~Sasamoto, M.~Wadati.
\newblock Exact results for one-dimensional totally asymmetric diffusion models.
\newblock {\it J. Phys. A}, {\bf 31}:6057--6071, 1998.


\bibitem{Schutz}
G.~M.~Sch\"{u}tz.
\newblock Exact solution of the master equation for the asymmetric exclusion process.
\newblock {\it J. Stat. Phys.}, {\bf 88}:427--445, 1997.

\bibitem{TS}
F.~Tabatabei, G.~M.~Sch\"{u}tz.
\newblock Shocks in the asymmetric exclusion process with internal degree of freedom.
\newblock {\it Phys. Rev. E.}, {\bf 74}:051108, 2006.


\bibitem{takeyama}
Y.~Takeyama.
\newblock A discrete analogoue of period delta Bose gas and affine Hecke algebra.
\newblock arXiv:1209.2758.

\bibitem{TW1}
C.~Tracy, H.~Widom.
\newblock Integral formulas for the asymmetric simple exclusion process.
\newblock {\em Commun. Math. Phys.}, {\bf 279}:815--844, 2008.
\newblock Erratum: {\em Commun. Math. Phys.} {\bf 304}:875--878, 2011.

\bibitem{TW3}
C.~Tracy, H.~Widom.
\newblock Asymptotics in ASEP with step initial condition.
\newblock {\em Commun. Math. Phys.}, {\bf 290}:129-154, 2009.

\bibitem{TWhalfspace}
C.~Tracy, H.~Widom.
\newblock The Bose Gas and Asymmetric Simple Exclusion Process on the Half-Line.
\newblock {\it J. Stat. Phys.}, {\bf 150}:1--12, 2013.

\bibitem{TWBern}
C.~Tracy, H.~Widom.
\newblock On ASEP with step Bernoulli initial condition.
\newblock {\it J. Stat. Phys.}, {\bf 137}:825--838, 2009.

\bibitem{Tsilevich}
N.~V.~Tsilevich.
\newblock Quantum inverse scattering method for the $q$-boson model and symmetric functions.
\newblock {\it Funct. Analy. Appl.}, {\bf 40}:207--217, 2006.


\bibitem{vd}
J.~F. van Diejen.
\newblock On the Plancherel formula for the (discrete) Laplacian in a Weyl chamber with republive boundary conditions at the walls.
\newblock {\it Ann. Inst. H. Poincare}, {\bf 5}:135--168, 2004.

\bibitem{VDcircle}
J.~F.~van~Diejen.
\newblock Diagonalization of an integrable discretization of the repulsive delta Bose gas on the circle.
\newblock {\it Commun. Math. Phys.}, {\bf 267}:451--476,  2006.

\bibitem{VDEE}
J.~F.~van~Diejen, E.~Emsiz.
\newblock Diagonalization of the infinity $q$-Boson system.
\newblock arXiv:1308.2237.

\bibitem{VDEE2}
J.~F.~van~Diejen, E.~Emsiz.
\newblock The semi-infinity $q$-Boson system with boundary interaction.
\newblock arXiv:1308.2242.

\bibitem{VS}
E.~P.~van Den Ban, H.~Schlichtkrull.
\newblock The most continuous part of the Plancherel decomposition for a reductive symmetric space
\newblock {\it Ann. Math} {\bf 145}:267--364, 1997.

\bibitem{Yang1}
C.~N.~Yang.
\newblock Some exact results for the many body problem in one dimension with repulsive delta function interaction.
\newblock {\it Phys. Rev. Lett.}, {\bf 19}:1312--1314, 1967.

\bibitem{Yang2}
C.~N.~Yang.
\newblock S matrix for the one dimensional N-body problem with repulsive or attractive delta-function interaction.
\newblock {\it Phys. Rev.}, {\bf 168}:1920--1923, 1968.

\bibitem{Yudson}
V.~I.~ Yudson.
\newblock Dynamics of integrable quantum systems.
\newblock {\it Zh. Eksp. Teor. Fiz.}, {\bf 88}: 1757--1770, 1985.

\end{thebibliography}
\end{document}